\documentclass[11pt,letterpaper]{article}
\usepackage[T1]{fontenc}
\usepackage[utf8]{inputenc}
\usepackage[margin=1in]{geometry}
\usepackage{comment}
\usepackage{amsmath,amssymb,amsthm,mathtools}
\usepackage{newtxtext,newtxmath}
\usepackage{mathrsfs}
\usepackage{booktabs,tabularx,array}
\usepackage{enumitem}
\usepackage{microtype}
\usepackage{needspace}
\usepackage[section]{placeins}
\usepackage{graphicx}
\usepackage{pgfplots}
\definecolor{RunBlue}{HTML}{0072B2}
\definecolor{HalfTeal}{HTML}{009E73}
\definecolor{BaselineOrange}{HTML}{D55E00}
\definecolor{UnionGray}{HTML}{777777}
\definecolor{LevGreen}{HTML}{009E73}
\definecolor{EliasPurple}{HTML}{AA4499}
\definecolor{HSSCyan}{HTML}{56B4E9}
\definecolor{HSSGold}{HTML}{B07B00}

\pgfplotsset{compat=1.18}
\pgfplotsset{comparison axis/.style={
 width=.95\textwidth,height=7.1cm,xmin=0,ymin=0,ymax=1,
 xlabel={Error fraction},ylabel={Rate},grid=major,
 tick label style={font=\small},label style={font=\small},
 scaled ticks=false,
 legend style={at={(.5,-.20)},anchor=north,legend columns=2,
 font=\footnotesize,draw=none,cells={anchor=west},column sep=1em},
 every axis plot/.append style={mark=none},
 xticklabel style={/pgf/number format/fixed,/pgf/number format/precision=2}
}}

\usepackage{xcolor}
\usepackage[nosort,nocompress]{cite}
\usepackage{url}
\usepackage{hyperref}
\newcommand\myshade{70}
\hypersetup{
  linkcolor  = red!\myshade!black,
  citecolor  = blue!\myshade!black,
  urlcolor   = blue!\myshade!black,
  colorlinks = true,
}
\usepackage[nameinlink,noabbrev]{cleveref}
\allowdisplaybreaks[2]

\newtheorem{theorem}{Theorem}
\newtheorem{lemma}[theorem]{Lemma}
\newtheorem{proposition}[theorem]{Proposition}
\newtheorem{corollary}[theorem]{Corollary}
\theoremstyle{definition}
\newtheorem{definition}[theorem]{Definition}
\newtheorem{example}[theorem]{Example}
\theoremstyle{remark}
\newtheorem{remark}[theorem]{Remark}
\newcommand{\cC}{\mathcal C}
\newcommand{\cB}{\mathcal B}
\newcommand{\cD}{\mathcal D}
\newcommand{\cI}{\mathcal I}
\newcommand{\cR}{\mathcal R}
\newcommand{\cU}{\mathscr U}
\newcommand{\R}{\mathbb R}

\newcommand{\Z}{\mathbb Z}
\newcommand{\E}{\mathbb E}
\newcommand{\vx}{\boldsymbol{x}}
\newcommand{\vy}{\boldsymbol{y}}
\newcommand{\vz}{\boldsymbol{z}}
\newcommand{\vu}{\boldsymbol{u}}
\newcommand{\vs}{\boldsymbol{s}}
\newcommand{\vt}{\boldsymbol{t}}
\newcommand{\vc}{\boldsymbol{c}}
\newcommand{\vw}{\boldsymbol{w}}
\newcommand{\va}{\boldsymbol{a}}
\newcommand{\vb}{\boldsymbol{b}}
\newcommand{\vvec}{\boldsymbol{v}}
\newcommand{\vj}{\boldsymbol{j}}
\newcommand{\vq}{\boldsymbol{q}}
\newcommand{\ve}{\boldsymbol{e}}
\newcommand{\vh}{\boldsymbol{h}}
\newcommand{\vB}{\boldsymbol{B}}
\newcommand{\vC}{\boldsymbol{C}}
\newcommand{\vV}{\boldsymbol{V}}

\newcommand{\vmu}{\boldsymbol{\mu}}
\newcommand{\vtheta}{\boldsymbol{\theta}}
\newcommand{\vlambda}{\boldsymbol{\lambda}}
\newcommand{\vzero}{\boldsymbol{0}}
\newcommand{\vdagger}{\boldsymbol{\dagger}}
\newcommand{\vone}{\boldsymbol{1}}
\newcommand{\eqdef}{\coloneqq}
\newcommand{\subseq}{\preceq}
\newcommand{\Int}[1]{\left[#1\right]}
\DeclarePairedDelimiter{\abs}{\lvert}{\rvert}

\DeclareMathOperator{\runs}{runs}
\DeclareMathOperator{\LCS}{LCS}
\DeclareMathOperator{\SCS}{SCS}
\newcommand{\ldel}[1]{\lambda_{\mathrm{del}}^{(#1)}}
\newcommand{\lins}{\lambda_{\mathrm{ins}}^{(L)}}
\newcommand{\Tdel}[1]{\mathsf T_{\mathrm{del}}^{(#1)}}
\newcommand{\Tins}{\mathsf T_{\mathrm{ins}}^{(L)}}
\DeclareMathOperator{\supp}{supp}

\newcommand{\Udel}{\mathsf U_{q,\mathrm{del}}^{\mathrm{HSS}}}
\newcommand{\Uins}{\mathsf U_{q,\mathrm{ins}}^{\mathrm{HSS}}}
\newcommand{\Rhalf}{\underline R_{\mathrm{half}}^{(L)}}
\newcommand{\RdelEval}{\underline R_{\mathrm{del}}^{(L)}}
\newcommand{\Rirun}{\underline R_{\mathrm{ins}}^{(L)}}
\newcommand{\Runi}{\underline R_{\mathrm{uni}}}
\newcommand{\Gsame}{G_{\mathrm s}}
\newcommand{\Gflip}{G_{\mathrm f}}
\newcommand{\Gsplus}{G_{\mathrm s}^{+}}
\newcommand{\Gfplus}{G_{\mathrm f}^{+}}
\newcommand{\Ttail}{T_{\mathrm{tail}}}
\newcommand{\Ddomain}{\mathcal D_{\mathrm{del}}}
\newcommand{\pos}[1]{\left[#1\right]_+}
\newcommand{\set}[1]{\left\{#1\right\}}
\newcommand{\image}[1]{\operatorname{Im}(#1)}

\hypersetup{pdftitle={Exponentially Larger Deletion-Correcting Codes and Constant-List-Decodable Insertion–Deletion Codes}}

\title{Exponentially Larger Deletion-Correcting Codes and Constant-List-Decodable Insertion–Deletion Codes%
\thanks{The work was supported by the National Natural Science
Foundation of China under Grant 12371523.}}

\author{%
Han Mao Kiah%
\thanks{H. M. Kiah is with the School of Physical and Mathematical
Sciences, Nanyang Technological University, Singapore. Email: \href{mailto:hmkiah@ntu.edu.sg}{\texttt{hmkiah@ntu.edu.sg}}.}
\and
Hengjia Wei%
\thanks{H. Wei is with the School of Mathematics and Statistics,
Xi'an Jiaotong University, Xi'an 710049, China. Email: \href{mailto:hjwei05@gmail.com}{\texttt{hjwei05@gmail.com}}.}
\and
Ruixiao Zeng%
\thanks{R. Zeng is with the College of Computing and Data Science, Nanyang Technological University, Singapore.
 Email: \href{mailto:zeng0146@e.ntu.edu.sg}{\texttt{zeng0146@e.ntu.edu.sg}}.}
}

\date{}
\begin{document}
\pagenumbering{roman}
\maketitle

\begin{abstract}
We establish improved asymptotic rate bounds for adversarial insertion--deletion codes. For binary unique-deletion correction in the linear-deletion regime, we strictly improve Levenshtein's classical lower-bound expression. At deletion fraction $\delta=0.1$, our achievable rate increases from $0.162009$ to $0.194890$, a relative gain of approximately $20.30\%$.

For list decoding with any prescribed finite list size $L$, we obtain new achievable-rate bounds for binary deletion and insertion codes. Our approach combines the Lov\'asz local lemma with refined counting of common subsequences and supersequences, using generating functions and run-restricted sampling.

We also establish a combinatorial reduction that trades $L$ units of insertion budget for one
unit of deletion budget without changing the code or increasing the list size. It transfers insertion lower bounds at $\gamma+L\delta$
and deletion upper bounds at $\delta+\gamma/L$ to mixed errors.
We also derive a run-counting Levenshtein-type list-decoding upper bound
and a higher-order Elias-type upper bound. The latter recovers
Yasunaga's asymptotic unique-decoding bound at $L=1$ and strictly
improves the Haeupler--Shahrasbi--Sudan insertion upper bound for every
fixed $L$ and $0<\gamma<q-1$. Numerical comparisons quantify the gains
and the remaining gaps.
\end{abstract}

\clearpage
\setcounter{tocdepth}{2}
\begingroup
\setlength{\parskip}{0pt}
\hypersetup{linkcolor=black}
\tableofcontents
\endgroup
\clearpage
\pagenumbering{arabic}

\section{Introduction}\label{sec:introduction}

Insertion and deletion errors arise when communication or storage
systems lose symbol synchronization. They include timing errors in
digital communication and length-changing errors in recorded sequences;
surveys of these models and their coding theory appear
in~\cite{Mitzenmacher2009,MercierBhargavaTarokh2010,CheraghchiRibeiro2021,HaeuplerShahrasbi2021Survey}.
DNA-based storage provides another application: synthesis and sequencing
introduce insertions and deletions as well as substitutions
\cite{HeckelMikutisGrass2019,SabaryKiahSiegelYaakobi2024}.
Unlike a substitution, an insertion or deletion changes the alignment
between the transmitted and received words. Insertion--deletion codes,
or \emph{insdel codes}, must protect the information despite that loss
of alignment.

We consider adversarial insertion and deletion errors, requiring decoding guarantees to hold for every error pattern within the prescribed budgets. For fixed-length codes, correcting up to $t$ deletions is equivalent to correcting up to $t$ insertions, or any combination of at most $t$ insertions and deletions~\cite{Levenshtein1966}. Thus, we focus on deletions when studying unique decoding in Section~\ref{sec:unique}. This equivalence no longer holds for list decoding, but we establish a weaker relation between insertions and deletions in Section~\ref{sec:reduction}.

\subsection{Unique decoding and classical bounds}

Levenshtein established fundamental upper and lower bounds on the size of insertion–deletion codes~\cite{Levenshtein1966,Levenshtein2002}. A central difficulty, absent in the Hamming metric, is that the number of distinct words obtained after a fixed number of deletions depends on the transmitted word. Levenshtein accounted for this nonuniformity through various estimates of deletion-sphere sizes. Sharper counting arguments therefore offer a natural route to improving both upper and lower bounds.

Subsequent work has substantially developed the upper-bound side.
Kulkarni and Kiyavash~\cite{KulkarniKiyavash2013} modeled deletion
correction as a hypergraph matching problem and used linear programming
to obtain finite-length bounds and asymptotic consequences.
Fazeli, Vardy, and Yaakobi~\cite{FazeliVardyYaakobi2015} developed a
generalized sphere-packing framework for channels with nonuniform error
balls. Cullina and Kiyavash~\cite{CullinaKiyavash2014} improved packing
bounds by varying the mixture of insertions and deletions used in the
counting argument. Yasunaga~\cite{Yasunaga2024ImprovedBounds} derived an
Elias-type upper bound, using an averaging argument that will also be
important here. More recently, Kong, Tamo, and Wei~\cite{KongTamoWei2025}
obtained alphabet-dependent combinatorial bounds, including a
linear-programming sphere-packing bound with improvements for large
distance or alphabet size. 

On the lower-bound side, Levenshtein's asymptotic rate expression has long stood as {\em the} benchmark for binary deletion codes in the linear-deletion regime:
\begin{equation}\label{eq:levenshtein-baseline}
R_{\mathrm{Lev}}(\delta)=1+\delta-2h_2(\delta),
\qquad 0<\delta<1/2,
\end{equation}
where $h_2$ is binary entropy; a negative value is replaced by zero
\cite{Levenshtein2002,Yasunaga2024ImprovedBounds}. Here the number of
deletions is $d=\lfloor\delta n\rfloor$, with $\delta>0$ fixed as
$n\to\infty$. Improving the asymptotic rate requires an exponential
improvement in the guaranteed code size: multiplying a lower bound on
$|\cC|$ by a polynomial in $n$ does not change its normalized logarithm.

There have been important improvements in related regimes.
Alon, Bourla, Graham, He, and Kravitz~\cite{ABGHK2024} obtained a
logarithmic-factor improvement for a fixed number of deletions.
Kong, Tamo, and Wei~\cite{KongTamoWei2025} refined this improvement and
extended it to additional alphabet and distance regimes; they also
obtained asymptotically tight code-size bounds as the alphabet size tends to infinity while the block length and
distance remain fixed. More recently, En Gad~\cite{EnGad2026} obtained
polynomially larger codes for a fixed number of deletions, while
Lin~\cite{Lin2026} improved fixed-deletion-count bounds with an explicit
list-size dependence. These results concern a fixed deletion count,
rather than $d=\lfloor\delta n\rfloor$ with fixed $\delta>0$.
Yasunaga's finite-length lower bound improves the classical estimate, but has the same asymptotic rate
in the linear-deletion regime~\cite{Yasunaga2024ImprovedBounds}.
A different line of work, including Bukh, Guruswami, and
H\aa stad~\cite{BGH2017}, gives positive-rate codes for much larger
deletion fractions. Our question concerns a quantitative improvement
of the classical rate curve where it already gives a nonnegative bound,
rather than only a larger positive-rate decoding radius or a
subexponential gain in code size. This distinction leads to the
following question about a strict improvement throughout the
nonnegative part of Levenshtein's rate curve.

\begin{quote}
\noindent\textbf{Question 1.}
\emph{For positive deletion fractions in the linear-deletion regime, can
Levenshtein's classical rate lower bound be strictly improved throughout
the range in which it is nonnegative?}
\end{quote}
We answer Question~1 affirmatively.


\subsection{List decoding: constant list sizes and insertion--deletion asymmetry}

In \emph{list decoding}, the receiver may return at most $L$ candidate
codewords, including the transmitted word. For insertion budget $a$
and deletion budget $b$, this means that every received word belongs
to the insertion--deletion error balls of at most $L$ codewords:
\[
|\{\vc\in\cC:\vw\in\cB(\vc;a,b)\}|\le L.
\]
Section~\ref{sec:preliminaries} defines these balls precisely. List size
therefore provides a third parameter in the rate--error-tolerance tradeoff.

Wachter-Zeh~\cite{WachterZeh2018} studied list decoding in the
Levenshtein metric and list-decoding algorithms for
Varshamov--Tenengolts codes. Subsequent explicit constructions include
the synchronization-string-based codes of Haeupler, Shahrasbi, and
Sudan (HSS)~\cite{HaeuplerShahrasbiSudan2018SyncStringsList}, the
positive-rate constructions attaining the optimal insdel feasibility
region of Guruswami, Haeupler, and
Shahrasbi~\cite{GuruswamiHaeuplerShahrasbi2020Optimally}, and the
high-rate linear edit-correcting codes of Li, Gabrys, and
Farnoud~\cite{LiGabrysFarnoud2025}.
These constructions provide list-size guarantees that are at most polynomial in the block length $n$ for fixed error fractions.

Constant-list guarantees are also available. Guruswami and
Wang~\cite{GuruswamiWang2017} constructed binary codes list-decodable
from deletion fraction $1/2-\varepsilon$, with list size depending only
on $\varepsilon$. Hayashi and Yasunaga~\cite{HayashiYasunaga2020}
developed a Johnson-type bound and constant-list constructions for
mixed errors. The random-coding results of Guruswami--Wang, HSS, and
Liu--Tjuawinata--Xing give list sizes $O(1/\varepsilon)$ at rates
$\varepsilon$ below the corresponding benchmarks
\cite{GuruswamiWang2017,HaeuplerShahrasbiSudan2018SyncStringsList,LiuTjuawinataXing2021Concatenation}.
Con, Doron, and Ribeiro~\cite{CDR2026} determined binary insertion
list-decoding capacity, with list size $O(1/\varepsilon)$ at rate
$C_{\mathrm{ins}}(\gamma)-\varepsilon$.

These results establish that constant-list decoding is possible in various insertion and deletion regimes. However, guarantees such as $L=O(1/\varepsilon)$ do not directly quantify the achievable rate for a prescribed list size, such as $L=2$ or $L=3$. To obtain suitable benchmarks, we make the dependence on $L$ explicit in the HSS deletion random-coding analysis and the Con--Doron--Ribeiro insertion union-bound argument. The resulting rate bounds, given in Sections~\ref{sec:baseline-rates} and~\ref{sec:insertion-comparison}, provide a basis for comparison with our improvements.

Unlike unique decoding, insertion and deletion list decodability
are not equivalent in general. This asymmetry is visible in both
rate and decoding-radius guarantees. HSS achieve rate
$1-\delta-\varepsilon$ against a deletion fraction $\delta$ and
any fixed insertion fraction $\gamma$, over an alphabet whose
size may depend on $\delta$, $\gamma$, and $\varepsilon$
\cite{HaeuplerShahrasbiSudan2018SyncStringsList}.
The deletion fraction enters the rate loss directly, whereas
the insertion fraction does not appear in this rate expression.
The allowance for a parameter-dependent alphabet and a list-size
bound depending on $n$ is important here.

A concrete binary example is provided by Hayashi and
Yasunaga~\cite[Theorem~3]{HayashiYasunaga2020}: positive-rate binary
codes can be list-decoded from any fixed insertion fraction
$\gamma<1/\sqrt{2}$ with a list size independent of $n$.
In contrast, correcting a deletion fraction $1/2$ with polynomially
bounded lists forces the binary code rate to zero. Indeed, every
binary word contains either $0^{\lceil n/2\rceil}$ or
$1^{\lceil n/2\rceil}$ as a subsequence, so a code list-decodable
from $\lfloor n/2\rfloor$ deletions with list size $L_n$ has at
most $2L_n$ codewords. These contrasting guarantees motivate a
more direct question about the relationship between the two
error models.

\begin{quote}
\noindent\textbf{Question 2.}
\emph{Given a list size $L$, what quantitative relations hold
between insertion and deletion list decodability?}
\end{quote}

We therefore seek sharper achievable-rate and converse bounds for a fixed list size $L$, independent of the block length $n$.

\begin{quote}
\noindent\textbf{Question 3.}
\emph{For a fixed alphabet and a given constant list size $L$,
what rates are achievable under insertions, deletions, or both,
and what upper bounds capture the dependence on $L$?}
\end{quote}

\subsection{Our contributions}

We first give a strict improvement of Levenshtein's asymptotic rate
bound, answering Question~1. We also establish a budget-conversion
result addressing Question~2 and new achievability and converse bounds
toward Question~3. Our main achievability results are
for binary codes, while the converses hold over every fixed alphabet.
We write $\cR_{q,\mathrm{del}}^{(L)}(\delta)$ and
$\cR_{q,\mathrm{ins}}^{(L)}(\gamma)$ for the corresponding supremal
asymptotic rates. Here $\delta$ denotes a deletion fraction and $\gamma$
an insertion fraction, both measured relative to the transmitted block
length. Precise definitions are in Section~\ref{sec:preliminaries}.

\Needspace{12\baselineskip}

Our first rate improvement already occurs for unique decoding.
Define
\[
\lambda_{\mathrm{del}}^{(1)}(x,y)
=1+\frac{y^2(1+3x)}{1-x}.
\]

\begin{theorem}[Improving the Levenshtein rate bound]
\label{thm:unique-main}\label{cor:deletion-strict-improvement}
For $0<\delta<1/2$, let
\begin{equation}\label{eq:unique-main-bound}
\underline R_{\mathrm{del}}^{(1)}(\delta)
=
\sup_{\substack{x,y>0\\x(1+y)<1}}
\left\{
2\log_2(1+y)
-(1-\delta)\log_2\lambda_{\mathrm{del}}^{(1)}(x,y)
+2\delta\log_2 x
\right\}.
\end{equation}
Then
\begin{equation}\label{eq:unique-main}
\cR_{2,\mathrm{del}}^{(1)}(\delta)
\ge \pos{\underline R_{\mathrm{del}}^{(1)}(\delta)}.
\end{equation}
Furthermore,
$
\underline R_{\mathrm{del}}^{(1)}(\delta)
>R_{\mathrm{Lev}}(\delta)
$.
\end{theorem}

This bound is obtained by sampling codewords uniformly from
a set of binary words with a prescribed number of runs.
Interestingly, the strict improvement over Levenshtein
already holds at run density $1/2$
(see Remark~\ref{rem:uniform-rate}).
Allowing the run density to vary gives a stronger bound.
For example, at $\delta=0.1$, the feasible parameters
$x=4/27$ and $y=5/7$ give
\begin{equation}\label{eq:intro-numerical-improvement}
R_{\mathrm{Lev}}(0.1)=0.162008812821\ldots,
\qquad
\cR_{2,\mathrm{del}}^{(1)}(0.1)\ge0.194890.
\end{equation}
This represents a relative improvement of approximately $20.30\%$.

The second contribution is a quantitative relation between the error budgets.

Proposition~\ref{prop:ins2del} converts error budgets without changing
the code or the list size. Writing $(\gamma,\delta,L)$ for list
decodability from fractions $\gamma$ of insertions and $\delta$ of
deletions with at most $L$ candidates, its integer-budget statement gives
the asymptotic implication
\begin{equation}\label{eq:reduction-intro}
(\gamma,\delta,L)\ \Longrightarrow\
\left(0,\delta+\frac{\gamma}{L},L\right).
\end{equation}
Thus a deletion converse can be applied at the effective deletion
fraction $\tau=\delta+\gamma/L$, with the finite-length budgets and
limiting argument made precise in Section~\ref{sec:reduction}.
The same proposition transfers insertion-only achievability at fraction
$\gamma+L\delta$ to mixed errors with fractions $(\gamma,\delta)$.
This gives a direct quantitative answer to Question~2, but not an
equivalence or a claim that all possible budget conversions have been
characterized.

{
We next establish improved lower bounds for binary deletion and insertion codes with any fixed list size $L$. For deletions with $0<\delta<1/2$, our application of the Lov\'asz local lemma yields
\begin{equation}\label{eq:uniform-LLL-intro}
\cR_{2,\mathrm{del}}^{(L)}(\delta)\ge 1+\frac{\delta}{L}-\frac{L+1}{L}h_2(\delta).
\end{equation}
This recovers Levenshtein's classical bound at $L=1$ and improves the fixed-$L$ HSS union-bound benchmark (Corollary~\ref{cor:uniform-deletion}). By exploiting additional restrictions on common subsequences, we obtain a further improvement of at least
\[
-\frac{1-\delta}{L}\log_2(1-\delta^{L+1})>0,
\qquad 0<\delta<1/2,
\]
for every fixed $L$. This explicit gain follows by evaluating our general deletion bound (Proposition~\ref{thm:deletion-lower-bound}) at $(x,y)=(\delta,1)$ with run density $1/2$ (Theorem~\ref{thm:fixed-L-strict}).

For insertions, the Lov\'asz local lemma improves the
fixed-$L$ union bound of Con--Doron--Ribeiro, and our
SCS-coverage argument yields a further strict improvement
(Theorem~\ref{thm:insertion-strict-improvement}).

Finally, we establish two families of upper bounds for list decoding over arbitrary finite alphabets. We first give a finite-length, alphabet-dependent
Levenshtein-type list-decoding bound and its asymptotic form
(Theorem~\ref{thm:levenshtein-insdel-list} and
Corollary~\ref{cor:levenshtein-insdel-asymptotic-fixed-q}). Its run-counting
argument is also related to the binary deletion converse in
\cite[Section~4.2]{CDR2026}. We next prove a higher-order Elias bound
(Theorem~\ref{thm:higher-order-elias}) that recovers Yasunaga's asymptotic
bound at $L=1$ and strictly improves the HSS insertion upper bound for every
finite $L$ and $0<\gamma<q-1$. In particular, for binary codes,
\begin{equation}\label{eq:binary-zero-intro}
\gamma\ge1-2^{-L}\quad\Longrightarrow\quad
\cR_{2,\mathrm{ins}}^{(L)}(\gamma)=0.
\end{equation}
Thus, \(\gamma<1-2^{-L}\) is necessary for positive rate. For deletions, the upper bounds considered here complement one another
across different parameter ranges.
Section~\ref{sec:comparison} displays these improvements and the
remaining gaps between lower and upper bounds. All our achievability
results are existential; we do not claim efficient encoders or decoders
attaining the stated rates.
}

{
\subsection{Overview of the approach}

Our lower bounds follow the probabilistic method. We sample
candidate codewords independently and bound the probability
that a given tuple of $L+1$ codewords shares an admissible
received word. The Lov\'asz local lemma then guarantees
that, with positive probability, no such bad tuple occurs.
Sampling uniformly from the unrestricted space
$\{0,1\}^n$ and using elementary probability estimates
yields our baseline deletion bound, which recovers
Levenshtein's bound when $L=1$ and improves the
fixed-list union-bound estimate of Haeupler, Shahrasbi, and Sudan.

Our improved lower bounds combine run-restricted sampling with
sharper counting of common witnesses. We sample candidate codewords
independently and uniformly from length-$n$ binary words that
start with $0$ and have a prescribed number of runs. Varying
the run density allows us to balance the size of the sampling
ensemble against the probability of bad tuples. The use of
constrained sampling ensembles to improve coding bounds has
precedents in~\cite{marcus1992improved,goyal2024evaluating}.

The key counting idea is to exploit a necessary local property
of an extremal common witness. For deletions, choose a longest
common subsequence $\boldsymbol v$ of a bad tuple of $L+1$
binary words and use its leftmost embedding in each word.
Between consecutive matches, each gap contains only copies
of the symbol opposite to the next matched symbol. If all
$L+1$ gaps were nonempty, one could insert an additional common
symbol into $\boldsymbol v$, contradicting its maximal length.
A prefix of $\boldsymbol v$ of the required length $n-d$
inherits this restriction (Lemma~\ref{lem:local-LCS}).
Thus, excluding configurations with all gaps nonempty still
leaves at least one counted witness for every bad tuple.

For insertions, we use the analogous minimality property of
a shortest common supersequence: every position must be used
by at least one codeword embedding, since an unused position
could otherwise be deleted. We therefore exclude transitions
in which all codewords skip the same witness position.

We encode these restricted witness--tuple configurations using
generating functions and weighted transfer matrices, keeping
track of the individual word lengths and run counts. Such
methods have also been used to study coding bounds and
constrained channels~\cite{goyal2024gilbert,lenz2025multivariate}.
Positive evaluations give exponential upper bounds on the
bad-event probabilities, which the Lov\'asz local lemma
converts into achievable-rate bounds. In particular, excluding
the all-nonempty-gap configurations yields the explicit
deletion-rate gain in Theorem~\ref{thm:fixed-L-strict}; for
$L=1$, this gives the strict improvement over Levenshtein's
expression.

Our error-budget reduction is proved by contradiction.
We restore up to $t$ deleted symbols from each of $L+1$
codewords into a common received word. For each codeword,
its own restored symbols reduce deletions, while those
restored for the other $L$ words require at most $Lt$
additional insertions. The resulting word violates the
assumed list-decoding guarantee, explaining the factor $L$
without changing the code or list size.

For the upper bounds, run counting and double counting
of deletion-sphere incidences give Levenshtein-type bounds.
We extend Yasunaga's ambient-word averaging
argument~\cite{Yasunaga2024ImprovedBounds} using a higher-order
intersection lemma on embedding-position sets. Large
intersections yield long common subsequences, while small
unions, obtained by applying the lemma to complements,
yield short common supersequences. These witnesses limit
the number of codewords in each ambient word, giving rate
upper bounds by double counting.

Section~\ref{sec:preliminaries} introduces the error models and known upper bounds. Section~\ref{sec:baselines} establishes the budget reduction, sampling lemma, and baseline rates. Sections~\ref{sec:unique} and~\ref{sec:deletions} develop our deletion lower bounds, while Section~\ref{sec:insertions} treats insertions. Section~\ref{sec:converses} proves the upper bounds, and Section~\ref{sec:comparison} compares the resulting estimates.

}

\section{Preliminaries}
\label{sec:preliminaries}
Let $\Sigma_q=\{0,1,\ldots,q-1\}$ and $[n]=\{1,\ldots,n\}$.
We use bold letters for words, sequences, and vectors, and ordinary
letters for their individual symbols or coordinates. For example,
$\vc=(c_1,\ldots,c_n)$ denotes a word, while $c_j$ is its $j$th symbol;
$\vc_i$ denotes the $i$th word in a tuple. Bold capital letters, such
as $\vC_i$, denote random words or random vectors.
For words $\vu,\vvec$, write $\vu\subseq \vvec$ if $\vu$ is a subsequence of $\vvec$,
including the case $\vu=\vvec$. An embedding is a strictly increasing map
between their position sets that matches the corresponding symbols.
For a word $\vu$ and integers $a,b\ge0$, the \emph{insertion--deletion
error ball} $\cB(\vu;a,b)$ consists of the words obtainable from $\vu$
by at most $a$ insertions and at most $b$ deletions. For an exact number
$t$ of pure errors, we use the \emph{$t$-deletion sphere} and
\emph{$t$-insertion sphere}, respectively:
\[
\cD_t(\vu)=\{\vw\in\Sigma_q^{|\vu|-t}:\vw\subseq \vu\},\qquad
\cI_t(\vu)=\{\vw\in\Sigma_q^{|\vu|+t}:\vu\subseq \vw\}.
\]
We reserve $\cB$ for the budgeted mixed-error
ball. 

Insertion spheres have the same size at every center of a given length,
whereas deletion-sphere sizes depend on the center. In particular,
every length-$n$ word has the same number of length-$(n+t)$
supersequences~\cite{Levenshtein2002,Yasunaga2024ImprovedBounds}, namely
\begin{equation}\label{eq:insertion-sphere-Iq}
I_q(n,t)\eqdef |\cI_t(\vu)|
=\sum_{j=0}^{t}\binom{n+t}{j}(q-1)^j.
\end{equation}

\begin{definition}\label{def:list}
A code $\cC\subseteq\Sigma_q^n$ is \emph{$(a,b;L)$-list-decodable}
if for every received word $\vw$,
\[
\bigl|\{\vc\in\cC:\vw\in\cB(\vc;a,b)\}\bigr|\le L.
\]
Equivalently, $|\cB(\vw;b,a)\cap\cC|\le L$; the budgets are
interchanged when the ball is centered at the received word.
For real $\gamma,\delta\ge0$, the notation
\emph{$(\gamma,\delta,L)$-insdel list-decodable} means
$(\lfloor\gamma n\rfloor,\lfloor\delta n\rfloor;L)$-list-decodable.
Thus $\gamma$ always denotes insertions and $\delta$ always denotes deletions.
\end{definition}
For pure deletions, put $d=\min\{n,\lfloor\delta n\rfloor\}$.
It suffices to test received words of length $n-d$: a longer common
subsequence can be shortened.
For pure insertions it suffices to test words of length
$n+\lfloor\gamma n\rfloor$: a shorter common supersequence can be extended.

The rate of $\cC$ is $n^{-1}\log_q|\cC|$.
We write $\cR_q^{(L)}(\gamma,\delta)$ for the supremum of achievable
asymptotic rates, and abbreviate
$\cR_{q,\mathrm{del}}^{(L)}(\delta)=\cR_q^{(L)}(0,\delta)$ and
$\cR_{q,\mathrm{ins}}^{(L)}(\gamma)=\cR_q^{(L)}(\gamma,0)$. We write $[x]_+=\max\{x,0\}$; a negative lower-bound expression is
replaced by zero.

All logarithms in the binary achievability analysis have base two.
Define
\[
h_2(x)=-x\log_2x-(1-x)\log_2(1-x),\qquad
H_q(x)=\frac{h_2(x)+x\log_2(q-1)}{\log_2q},
\]
with $0\log0=0$, and let
\[
\overline H_q(x)=
\begin{cases}H_q(x),&0\le x\le(q-1)/q,\\1,&(q-1)/q\le x\le1.\end{cases}
\]

For a finite nonnegative matrix $T$, write $\lambda(T)$ for its
spectral radius, the largest modulus of its eigenvalues. If a power
of $T$ is entrywise positive, the Perron--Frobenius theorem gives a
positive, simple, strictly dominant eigenvalue $\lambda(T)$ and
strictly positive left and right eigenvectors. We call it the
\emph{Perron eigenvalue} and call a positive right eigenvector a
\emph{Perron vector}. The notation $\lambda(T)$ will also be used for
reducible nonnegative matrices, without a simplicity claim.

\subsection{Known Upper Bounds}\label{sec:benchmark-converses}
We record three known upper bounds used for comparison in
Sections~\ref{sec:converses} and~\ref{sec:comparison}. The
Yasunaga bound concerns unique decoding, whereas the two HSS bounds allow
polynomially bounded list size and hence apply, in particular, to every
fixed list size.

For equal-length words, the Levenshtein distance is the total number of
insertions and deletions needed to transform one word into the other; it
equals twice their deletion distance.
\begin{theorem}[Yasunaga's Elias-type bound]
\label{thm:yasunaga-elias}
\cite[Corollary~2]{Yasunaga2024ImprovedBounds}
Let $\mathcal C_n\subseteq\Sigma_q^n$ be a family of codes, and let
$d_n$ denote the minimum Levenshtein distance of $\mathcal C_n$.
If $\liminf_{n\to\infty} d_n/(2n)\ge\Delta$ for some
$0\le\Delta<(q-1)/q$, then
\begin{equation}\label{eq:yasunaga-asymptotic}
\limsup_n\frac1n\log_q|\cC_n|
\le\frac{1-H_q(\Delta)}{1-\Delta}.
\end{equation}
\end{theorem}

For later use, define
\begin{equation}\label{eq:Gq-definition}
G_q(p)=\frac{1-\overline H_q(p)}{1-p},\qquad 0\le p<1.
\end{equation}
Thus $G_q$ is zero for $p\ge(q-1)/q$ and is strictly decreasing on
$(0,(q-1)/q)$, since
\begin{equation}\label{eq:Gq-derivative}
G_q'(p)=\frac{\log_q(qp/(q-1))}{(1-p)^2}<0.
\end{equation}

The next two bounds allow a polynomially bounded list size.
\label{par:hss}
\begin{theorem}[HSS insertion-only bound]\label{thm:HSS-insertion}
\cite[Theorem~2]{HaeuplerShahrasbiSudan2018SyncStringsList}
Fix $q\ge2$ and $0\le\gamma<q-1$. A family of $q$-ary codes that is
list-decodable from $\lfloor\gamma n\rfloor$ insertions with a
polynomially bounded list size has asymptotic rate at most
\begin{equation}\label{eq:HSS-insertion-function}
\Uins(\gamma)=1-\log_q(1+\gamma)
-\gamma\left(\log_q\frac{1+\gamma}{\gamma}
             -\log_q\frac q{q-1}\right).
\end{equation}
The value at zero is understood by continuity.
\end{theorem}
Expanding $H_q$ gives the equivalent form
\begin{equation}\label{eq:HSS-insertion-as-Gq}
\Uins(\gamma)=G_q\left(\frac\gamma{1+\gamma}\right),
\qquad 0\le\gamma<q-1.
\end{equation}
For $q=2$ this is the function $C_{\mathrm{ins}}$ used in
Section~\ref{sec:insertions}.

\begin{theorem}[HSS deletion-only bound]\label{thm:HSS-deletion}
\cite[Theorem~3]{HaeuplerShahrasbiSudan2018SyncStringsList}
Fix $q\ge2$ and $0\le\delta<(q-1)/q$, and set
$f(x)=(1-x)(1+\log_q(1-x))$.
A family of $q$-ary codes that is list-decodable from
$\lfloor\delta n\rfloor$ deletions with a polynomially bounded list
size has asymptotic rate at most $f(\delta)$ when $\delta=d/q$ is a
grid point. More generally, if $\delta=d/q+\delta'$ with
$0\le\delta'<1/q$, then its rate is at most
\begin{equation}\label{eq:HSS-deletion-interpolation}
\Udel(\delta)=(1-q\delta')f(d/q)+q\delta'f((d+1)/q).
\end{equation}
Set $\Udel(\delta)=0$ for $\delta\ge(q-1)/q$.
\end{theorem}

\section{Error-Budget Reduction and Baseline Lower Bounds}
\label{sec:baselines}
We establish two tools used throughout the paper: a reduction that trades
insertion and deletion budgets without changing the code or its list
size, and a common-witness sampling lemma that turns counting estimates
into achievable rates. We then combine them with uniform sampling to
obtain baseline lower bounds, before exploiting run structure in the subsequent
sections. Throughout this section, $q\ge2$ and $L\ge1$ are fixed.

\subsection{Trading Insertions for Deletions}\label{sec:reduction}
\begin{proposition}[Trading insertions for deletions]\label{prop:ins2del}
Let $a,b,t$ be nonnegative integers. If a code $\cC\subseteq\Sigma_q^n$
is $(a+Lt,b;L)$-list-decodable, then it is $(a,b+t;L)$-list-decodable.
\end{proposition}
\begin{proof}
We may first delete original symbols and then insert the remaining new
symbols: canceling inserted symbols that are subsequently deleted does
not increase either error budget. Suppose, for the sake of contradiction,
that $\cC$ is not $(a,b+t;L)$-list-decodable. Then there exist $L+1$ distinct codewords $\vc_1,\vc_2,\ldots,\vc_{L+1}\in \cC$ and a sequence $\vs$ such that, for each $i\in[L+1]$, the sequence $\vs$ can be obtained from $\vc_i$ by $a_i$ insertions and $b_i$ deletions, where
\begin{equation}\label{eq:abrange}
a_i\le a
\qquad\text{and}\qquad
b_i\le b+t.
\end{equation}

    For each $1\leq i\leq L+1$, let $\vc_i'$ be the subsequence of $\vc_i$ obtained after the $b_i$ deletions, so that $\abs{\vc_i'} = n-b_i$ and $\vc_i'$ is a subsequence of $\vs$. Thus $|\vs|=n+a_i-b_i$ for every $i$. 
Let $\psi_i:[n-b_i]\to [n+a_i-b_i]$
be an embedding satisfying $\vc_i'[j]=\vs[\psi_i(j)]$ for all $j\in[n-b_i]$.

    Define $I\eqdef \set{i\in [L+1]| b_i>b }$. For each $i\in I$, take an arbitrary subsequence  $\vc_i''$ of $\vc_i$ of length $n-b$ such that $\vc_i' \subseq  \vc_i''  \subseq \vc_i$.
    Let $\phi_i: [n-b_i]\rightarrow [n-b]$
    be an embedding satisfying $\vc_i'[j]=\vc_i''[\phi_i(j)]$ for all $j \in [n-b_i]$. For each $i\notin I$, define simply $\vc_i''\eqdef \vc_i'$.

Then, for every $i\in[L+1]$, the sequence $\vc_i''$ can be obtained from $\vc_i$ using at most $b$ deletions.

We next construct a common supersequence $\vs'$ of $\vc_1'',\ldots,\vc_{L+1}''$ by inserting the symbols in $\vc_i''$ that are not in $\vc_i'$ into $\vs$. For each $i\in I$ and each $j\in [n-b] \backslash \image{\phi_i}$, let $\ell_{i,j}$ be the largest integer such that $\phi_i(\ell_{i,j})<j$. If such an integer exists, associate the symbol $\vc_i''[j]$ with the position immediately to the right of $\vs[\psi_i(\ell_{i,j})]$. Otherwise, associate $\vc_i''[j]$ with the left end of $\vs$. 
 
For each $k\in [|\vs|]$, let $I_{i,k}$ be the set of indices  of  $\vc_i''$ whose symbols are associated with the position immediately after $\vs[k]$. Similarly, let $I_{i,0}$ denote the set of indices associated with the left end of $\vs$. By construction, each nonempty $I_{i,k}$ forms a consecutive interval.

Now construct $\vs'$ as follows:
\begin{itemize}
    \item prepend $\vc_{i_1}''[I_{i_1,0}]
    \vc_{i_2}''[I_{i_2,0}]
    \cdots
    \vc_{i_M}''[I_{i_M,0}]$
    to the left of $\vs$, where
    $I=\{i_1,i_2,\ldots,i_M\}$;
    
    \item for each position $k$, insert
    $\vc_{i_1}''[I_{i_1,k}]
    \vc_{i_2}''[I_{i_2,k}]
    \cdots
    \vc_{i_M}''[I_{i_M,k}]$
    immediately after $\vs[k]$.
\end{itemize}

Let $\vs'$ denote the resulting sequence. By construction, $\vs'$ is a common supersequence of $\vc_1'',\vc_2'',\ldots,\vc_{L+1}''$. Moreover,
    \[\abs{\vs'}=\abs{\vs} + \sum_{\ell=1}^{L+1}(\abs{\vc_\ell''}-\abs{\vc_\ell'}) =\abs{\vs}+\sum_{\ell\in I} (b_\ell-b).\]

   Fix $i\in[L+1]$. Since $\vc_i''$ is a subsequence of $\vs'$, the number of insertions needed to transform $\vc_i''$ into $\vs'$ equals $\abs{\vs'}-\abs{\vc''_i}$.
    If $i\in I$, then
    \begin{align*}
\abs{\vs'}-\abs{\vc_i''}
&=n+a_i-b_i+\sum_{\ell\in I}(b_\ell-b)-(n-b)\\
&=a_i+\sum_{\ell\in I\setminus\{i\}}(b_\ell-b)\le a+Lt.
\end{align*}
    The inequality follows from \eqref{eq:abrange} and the fact that $\abs{I\backslash \set{i}}\leq L.$
    If $i\notin I$, then $\abs{\vc_i''}=\abs{\vc_i'}=n-b
    _i$
    and $\abs{I}\leq L$. It follows that
    \begin{align*}
\abs{\vs'}-\abs{\vc_i''}
&=n+a_i-b_i+\sum_{\ell\in I}(b_\ell-b)-(n-b_i)\\
&=a_i+\sum_{\ell\in I}(b_\ell-b)\le a+Lt.
\end{align*}

    Therefore,  for every $i\in[L+1]$, the sequence $\vs'$ can be obtained from $\vc_i$ using at most $b$ deletions and at most $a+Lt$ insertions. Equivalently, 
    \[\vc_i \in \cB(\vs'; b, a+Lt)\] for all $i\in[L+1]$. This contradicts the assumption that $\cC$ is $(a+Lt,b;L)$-list-decodable.
\end{proof}

For fractional budgets let $A_n=\lfloor\gamma n\rfloor$ and
$B_n=\lfloor\delta n\rfloor$. Since
$L\lfloor A_n/L\rfloor\le A_n$, the proposition implies deletion-list
decodability with the \emph{integer} budget
\begin{equation}\label{eq:effective-integer-budget}
t_n=\min\{n,B_n+\lfloor A_n/L\rfloor\}.
\end{equation}
In particular, when $\tau=\delta+\gamma/L<1$, we have $t_n/n\to\tau$.
The fractional implication \eqref{eq:reduction-intro} is understood in
this asymptotic sense; it need not hold with $t_n$ replaced by
$\lfloor\tau n\rfloor$ at every finite length. 

\begin{example}[The converse implication fails]\label{ex:reduction-not-equivalence}
Let $L=2$, $n=2$, and consider $\cC=\{00,01,11\}\subseteq\{0,1\}^2$.
We first show that $\cC$ is $(0,1;2)$-list-decodable. Indeed, the sets
of length-one subsequences are
\[
\cD_1(00)=\{0\},\qquad \cD_1(01)=\{0,1\},\qquad
\cD_1(11)=\{1\}.
\]
Their intersection is empty, so no received word obtainable by at most
one deletion corresponds to all three codewords. On the other hand,
$0011$ contains every codeword as a subsequence, and it is obtained from
each by two insertions. Hence $\cC$ is not $(2,0;2)$-list-decodable.
Thus deletion-list decodability with budget one does not imply
insertion-list decodability with budget $L\cdot1=2$.
\end{example}
The reverse implication therefore cannot transfer deletion-only
achievability to mixed errors. The forward reduction has two uses:
deletion upper bounds apply at the effective fraction $\tau=\delta+\gamma/L$,
whereas a code correcting $A_n+LB_n$ insertions is
$(A_n,B_n;L)$-list-decodable. Thus an insertion lower bound at
$\gamma+L\delta$ gives a mixed-error lower bound;
Section~\ref{sec:baseline-rates} combines this transfer with uniform sampling.

\subsection{A Common-Witness Sampling Lemma}\label{sec:sampling}
The symmetric Lov\'asz local lemma states that events of probability at
most $p$, each mutually independent of all but at most $D$ other events,
can all be avoided when $ep(D+1)\le1$. We use it in the following form.
\begin{lemma}[Common-witness sampling lemma]\label{lem:sampling}\label{thm:LLL}
Fix $L\ge1$. Let $P_n$ be a distribution on a finite set $\mathcal X_n$,
and let $\mathcal W_n(\vc)$ be a nonempty set of witnesses for each
$\vc\in\mathcal X_n$. Suppose independent samples $\vC_1,\ldots,\vC_{L+1}$
satisfy
\begin{equation}\label{eq:sampling-probability}
\Pr\left[\bigcap_{i=1}^{L+1}\mathcal W_n(\vC_i)\ne\varnothing\right]
\le 2^{-n(a-o(1))}.
\end{equation}
Then, for every $0<R<a/L$, and all sufficiently large $n$, there is a set
$\cC_n\subseteq\mathcal X_n$ of size at least
$\lfloor2^{Rn}\rfloor/L$ such that every witness belongs to at most
$L$ members of $\cC_n$.
\end{lemma}
\begin{proof}
Fix $0<R<a/L$, and put $M=\lfloor2^{Rn}\rfloor$.
Choose $\vC_1,\ldots,\vC_M$ independently according to $P_n$.
For every $(L+1)$-subset $S\subseteq[M]$, let $E_S$ be the bad event
that the samples indexed by $S$ have a common witness:
\[
E_S=\left\{\bigcap_{i\in S}\mathcal W_n(\vC_i)\ne\varnothing\right\}.
\]
By hypothesis, every such event has probability at most
$p_n=2^{-n(a-o(1))}$. The event $E_S$ is determined by the samples
with indices in $S$ and is mutually independent of all events whose
index sets are disjoint from $S$. A neighboring index set contains at
least one $i\in S$ and $L$ further indices. Hence its dependency degree
$D_n$ satisfies
\[
D_n\le (L+1)\binom{M-1}{L},\qquad
D_n+1\le 2^{n(LR+o(1))}.
\]
The gap $a-LR$ is positive and independent of $n$. Therefore
\[
ep_n(D_n+1)
\le 2^{n(LR-a+o(1))}<1
\]
for all sufficiently large $n$. The Lov\'asz local lemma gives an
outcome in which none of the events $E_S$ occurs. In this outcome,
every witness belongs to at most $L$ sample indices.

The samples form a multiset, so it remains to obtain a code consisting
of distinct words. For any sampled word $\vc$, choose one witness in
$\mathcal W_n(\vc)$, which is nonempty by assumption. If $\vc$ occurred
at least $L+1$ times, those sample indices would share this witness and
would give a bad event. Thus every word occurs at most $L$ times.
Removing repeated copies leaves at least $M/L$ distinct words. This
operation cannot increase the number of codewords associated with a
witness. The resulting set therefore has the required property, and
\[
\frac1n\log_2|\cC_n|\ge
\frac1n\log_2\!\left(\frac{\lfloor2^{Rn}\rfloor}{L}\right)
=R-o(1).
\]
\end{proof}
The lemma applies to nonuniform as well as uniform distributions. In the
pure-deletion case the witnesses are members of $\cD_d(\vc)$; in the
pure-insertion case they are members of $\cI_d(\vc)$.
Repeated sampled words are included in the bad events throughout.

\subsection{Baseline Achievable Rates}\label{sec:baseline-rates}
We first apply the sampling lemma to uniform binary deletion codes,
and then use the error-budget reduction to obtain a mixed-error bound
over a general alphabet.

\begin{corollary}[Uniform deletion baseline]\label{cor:uniform-deletion}
For $L\ge1$ and $0<\delta<1/2$,
\begin{equation}\label{eq:uniform-deletion}
\cR_{2,\mathrm{del}}^{(L)}(\delta)
\ge\left[1+\frac{\delta}{L}-\frac{L+1}{L}h_2(\delta)\right]_+.
\end{equation}
\end{corollary}
\begin{proof}
Set $d=\lfloor\delta n\rfloor$, and sample each $\vC_i$ uniformly from
$\{0,1\}^n$. For a fixed witness $\vw\in\{0,1\}^{n-d}$,
\eqref{eq:insertion-sphere-Iq} gives
\[
\bigl|\{\vc\in\{0,1\}^n:\vw\subseq \vc\}\bigr|
=I_2(n-d,d)=\sum_{j=0}^d\binom nj.
\]
Since $0<\delta<1/2$, this sum is $2^{n(h_2(\delta)+o(1))}$.
Consequently,
\[
\Pr[\vw\subseq \vC_i]
=2^{-n(1-h_2(\delta)+o(1))}.
\]
For a fixed set $S$ of $L+1$ sample indices, independence and a union
bound over the $2^{n-d}$ possible witnesses give
\begin{align*}
\Pr(E_S)
&\le\sum_{\vw\in\{0,1\}^{n-d}}
       \prod_{i\in S}\Pr[\vw\subseq \vC_i]\le2^{n(1-\delta+o(1))}
     2^{-n(L+1)(1-h_2(\delta)+o(1))}=2^{-n\{L+\delta-(L+1)h_2(\delta)-o(1)\}}.
\end{align*}
Applying Lemma~\ref{lem:sampling}, every positive rate strictly below
\[
\frac{L+\delta-(L+1)h_2(\delta)}{L}
=1+\frac\delta L-\frac{L+1}{L}h_2(\delta)
\]
is achievable. If this expression is nonpositive, a singleton code
provides the zero lower bound.
\end{proof}
For $L=1$, Corollary~\ref{cor:uniform-deletion} gives
\[
\cR_{2,\mathrm{del}}^{(1)}(\delta)
\ge\bigl[1+\delta-2h_2(\delta)\bigr]_+,
\]
which is exactly the classical Levenshtein lower bound
\eqref{eq:levenshtein-baseline}, with its nonnegative part taken.
Thus, the uniform deletion baseline extends the Levenshtein lower bound
from unique decoding to every fixed list size $L$.

This baseline is $\frac{L+1}{L}$ times the fixed-$L$ union-bound
expression $1-h_2(\delta)-(1-\delta)/(L+1)$ from the HSS random-code
analysis~\cite{HaeuplerShahrasbiSudan2018SyncStringsList}.
The local-lemma improvement and the subsequent improvement from
locally restricted LCS witnesses are separate steps.

Combining the same sampling lemma with the reduction in
Proposition~\ref{prop:ins2del} gives a coarse mixed-error baseline.
\begin{corollary}[A coarse $q$-ary local-lemma bound]
\label{prop:coarse-mixed-lll}
Let $q\ge2$, $L\ge1$, and $\gamma,\delta\ge0$, and put
$s=\gamma+L\delta$. Then
\begin{equation}\label{eq:coarse-mixed-lower}
\cR_q^{(L)}(\gamma,\delta)
\ge\left[1-\frac sL
-\frac{L+1}{L}(1+s)h_2\!\left(\frac1{1+s}\right)\log_q2\right]_+.
\end{equation}
\end{corollary}
\begin{proof}
For pure insertions with fraction $s>0$, set $d=\lfloor sn\rfloor$ and
$m=n+d$. A fixed witness $\vw\in\Sigma_q^m$ has at most $\binom mn$
length-$n$ subsequences. For independent uniform samples from $\Sigma_q^n$,
a union bound over the $q^m$ witnesses therefore gives
\begin{align*}
\Pr(E_S)
&\le q^{m-n(L+1)}\binom mn^{L+1}\le q^{n(s-L+o(1))}
  2^{n(L+1)(1+s)h_2(1/(1+s))+o(n)}.
\end{align*}
Applying Lemma~\ref{lem:sampling} and converting its binary rate to base
$q$ gives \eqref{eq:coarse-mixed-lower} for pure insertions; the lemma
also supplies the distinct-word conclusion. For mixed errors, let
$A_n=\lfloor\gamma n\rfloor$ and $B_n=\lfloor\delta n\rfloor$.
Since $A_n+LB_n\le\lfloor sn\rfloor$, Proposition~\ref{prop:ins2del}
transfers insertion correction to $(A_n,B_n;L)$-list decodability.
A nonpositive expression is replaced by zero, and $s=0$ is handled by
taking the full space.
\end{proof}
The bound in Corollary~\ref{prop:coarse-mixed-lll} holds for all $s>0$; when $s\geq q-1$, its right-hand side is zero.


\section{Improving Levenshtein's Asymptotic Lower Bound}\label{sec:unique}
We prove Theorem~\ref{thm:unique-main} by sampling candidate
codewords independently and uniformly from the set of length-$n$
binary words that start with $0$ and have a prescribed number
of runs. The bad event is that two sampled words share a
subsequence of length $n-d$, where $d$ is the deletion budget.
To bound its probability, we count witness--pair configurations
satisfying a local condition inherited from a longest common
subsequence. Every bad pair admits at least one such witness,
so an upper bound on this restricted count, obtained using
generating functions, yields a lower bound through
the sampling lemma.

\subsection{The bad-pair event and locally restricted witnesses}
\label{sec:leftmost-factors}

For $1\le r\le n$, let
\begin{equation}\label{eq:run-ensemble}
\cU(n,r)=\{\vc\in\{0,1\}^n:c_1=0,\ \runs(\vc)=r\}.
\end{equation}
A word in this set is determined by its $r-1$ run
boundaries, so
\begin{equation}\label{eq:run-ensemble-size}
|\cU(n,r)|=\binom{n-1}{r-1}
=2^{nh_2(\rho)+o(n)}
\quad\text{when }r=\lfloor\rho n\rfloor,\quad 0<\rho<1.
\end{equation}
In particular, $\cU(n,\lfloor n/2\rfloor)$ has
exponent one.
For words $\vc_1,\ldots,\vc_{L+1}$, write
$\LCS(\vc_1,\ldots,\vc_{L+1})$ for the length of
a longest common subsequence. If $0\le d<n$,
the bad event for two independent uniform samples $\vC_1$ and $\vC_2$
from $\cU(n,r)$ is
\[
E_{\{1,2\}}
=\{\LCS(\vC_1,\vC_2)\ge n-d\}.
\]
We bound its probability by counting suitable
length-$(n-d)$ common witnesses together with
the sampled pair. To this end, we introduce
locally restricted common witnesses.

For a witness $\vw$ starting with $0$, a word $\boldsymbol c$ with $c_1=0$ can contain $\vw$ in several ways. Its \emph{leftmost embedding}
matches each successive witness symbol at its earliest available
position after the preceding match. This embedding is unique and gives
\[
\vc=\underline{0}\,\vB_1\,\underline{w_2}\,
\vB_2\cdots\vB_{\ell-1}\,\underline{w_\ell}\,\vB_\ell.
\]
The initial $0$ is matched at position $1$. The gaps $\vB_j$ before
subsequent matches contain only the opposite of the next matched symbol;
otherwise that symbol would have been matched earlier. The suffix
$\vB_\ell$ is unrestricted. We use $\varepsilon$ for the empty word.

\begin{lemma}[Locally restricted common witnesses]\label{lem:local-LCS}
Let $L\ge1$, and let $\vc_1,\ldots,\vc_{L+1}$ be binary words starting
with $0$, with repetitions allowed. If
$\LCS(\vc_1,\ldots,\vc_{L+1})\ge m\ge1$, there is a common subsequence
$\vw=w_1\cdots w_m$ such that $w_1=0$ is matched at position $1$ in
every word and, before every subsequent match in the leftmost
embeddings, at least one of the $L+1$ gaps is empty. There is no condition
on the suffixes after the last match.
\end{lemma}
\begin{proof}
Choose a longest common subsequence $\vvec=v_1\cdots v_\ell$, where
$\ell\ge m$. It starts with $0$: if $v_1=1$, the initial zero in each
input precedes its first match, so $0\vvec$ would be a longer common
subsequence. Use the leftmost embeddings of $\vvec$. Before a match
of $v_j$, $j\ge2$, every gap contains only $\overline v_j=1-v_j$.
If all $L+1$ gaps were nonempty, one could match a further
$\overline v_j$ in each of them, producing
$v_1\cdots v_{j-1}\,\overline v_j\,v_j\cdots v_\ell$
of length $\ell+1$, a contradiction. Take $\vw=v_1\cdots v_m$.
The leftmost embedding of this prefix is the corresponding prefix of
the full leftmost embedding, so the stated conditions are inherited.
\end{proof}

\subsection{The generating function}

Fix a nonempty binary witness
$\vw=w_1\cdots w_\ell$ with $w_1=0$.
Assign each finite word $\vc$ starting with $0$
and containing $\vw$ the weight
\[
\operatorname{wt}_{\vw}(\vc)
=x^{|\vc|-\ell}y^{\runs(\vc)}.
\]
Here $x$ records unmatched symbols and $y$
records runs. Define the generating function
\begin{equation}\label{eq:fixed-witness-series}
F_{\vw}(x,y)=
\sum_{\substack{\vc\text{ starts with }0\\\vw\subseq\vc}}
x^{|\vc|-\ell}y^{\runs(\vc)},
\end{equation}
where the sum ranges over all finite words
satisfying these conditions. Thus, for $n=\ell+d$,
\begin{equation}\label{eq:fixed-witness-coefficient}
[x^dy^r]F_{\vw}(x,y)
=|\{\vc\in\cU(n,r):\vw\subseq\vc\}|.
\end{equation}
The notation $[x^dy^r]$ extracts the coefficient
of the indicated monomial. For example, when
$\vw=0$, the words $0$, $00$, and $01$ contribute
$y$, $xy$, and $xy^2$, respectively.

We use two elementary counting rules. Disjoint
alternatives are counted by adding their generating
functions. A unique decomposition into parts,
with additive recorded parameters, gives the product
of their generating functions. This is the
symbolic method~\cite{SedgewickFlajolet2013}.

We now translate the local restriction in
Lemma~\ref{lem:local-LCS} into generating functions.
We derive the factors for one word before combining
several words. The first matched zero creates one
run and no unmatched symbols, so $S_0(x,y)=y$.
Every later factor records only \emph{additional}
runs; a matched symbol contributes no power of $x$.

For two consecutive equal witness symbols $b,b$, the local form is
$\underline b\,(1-b)^a\,\underline b$, $a\ge0$. If $a=0$, the new
match extends the existing run, giving weight $1$. If $a\ge1$, the gap
and the return to $b$ create two new runs. Consequently, we set
\[
\Gsame(x,y)=1+\sum_{a\ge1}x^ay^2
=1+\frac{xy^2}{1-x},
\qquad
\Gsplus(x,y)=\frac{xy^2}{1-x}.
\]
The superscript $+$ retains only nonempty gaps. The earlier matched
symbol and its run have already been counted, which explains $y^2$
rather than $y^3$.

For consecutive different witness symbols $b,1-b$, the local form is
$\underline b\,b^a\,\underline{1-b}$, $a\ge0$. The gap extends the
current run, and the next match starts exactly one new run. Thus, we set
\[
\Gflip(x,y)=\sum_{a\ge0}x^ay=\frac{y}{1-x},
\qquad
\Gfplus(x,y)=\frac{xy}{1-x}.
\]
In particular, an empty flip gap contributes $y$, not $1$.

After the last match, every appended symbol is unmatched. It either
continues the current run, contributing $x$, or starts a new run,
contributing $xy$. An arbitrary suffix therefore has generating function
\[
\Ttail(x,y)=\sum_{t\ge0}[x(1+y)]^t=\frac1{1-x(1+y)}.
\]
For example, after a matched $0$, the two-symbol suffixes $00$, $01$,
$11$, and $10$ create respectively $0,1,1,2$ new runs; their total
weight is $x^2(1+y)^2$.

\begin{table}[tbp]
\centering\small
\renewcommand{\arraystretch}{1.25}
\begin{tabularx}{\textwidth}{@{}lXcc@{}}
\toprule
Part & Description & All choices & Nonempty gap\\
\midrule
Initial match & The first witness zero is matched at position $1$.
& $S_0=y$ & ---\\
Same transition & A gap of the opposite symbol, followed by the next match.
& $\Gsame=1+\dfrac{xy^2}{1-x}$ & $\Gsplus=\dfrac{xy^2}{1-x}$\\
Flip transition & Extra copies of the preceding symbol, followed by the next match.
& $\Gflip=\dfrac y{1-x}$ & $\Gfplus=\dfrac{xy}{1-x}$\\
Suffix & Arbitrary symbols after the final match.
& $\Ttail=\dfrac1{1-x(1+y)}$ & ---\\
\bottomrule
\end{tabularx}
\caption{Local factors for a witness starting with $0$. The variable $x$
counts unmatched symbols; $y$ counts initial or additional runs.
The nonempty-gap factors are used to exclude an all-nonempty choice
across the codewords.}
\label{tab:leftmost-factors}
\end{table}

\begin{lemma}[Fixed witness subsequence]
\label{lem:fixed-witness-subsequence}\label{lem:fixed-witness}
For a nonempty binary witness $\vw=w_1\cdots w_\ell$ with $w_1=0$,
\begin{equation}\label{eq:fixed-witness-gf}
F_{\vw}(x,y)=y\prod_{j=2}^{\ell}G_{w_{j-1},w_j}(x,y)\Ttail(x,y),
\end{equation}
where $G_{a,b}=\Gsame$ if $a=b$ and $G_{a,b}=\Gflip$ otherwise.
\end{lemma}
\begin{proof}
The first match, each permitted gap with its following match, and the
suffix give the unique leftmost decomposition. Their unmatched-symbol
counts and run increments add. Conversely, any choices of these parts
reconstruct a word whose indicated matches are its leftmost embedding.
Multiplying their weights and summing over the choices proves the identity.
\end{proof}

For example, consider the witness $\vw=010$. Every word starting with $0$ and containing $010$ has the unique form
\[
\vc=\underline0\,0^a\,\underline1\,1^b\,\underline0\,\vvec,
\qquad a,b\ge0,\quad \vvec\in\{0,1\}^*.
\]
Its generating function is
\[
F_{010}(x,y)=y\Gflip(x,y)^2\Ttail(x,y)
=\frac{y^3}{(1-x)^2(1-x(1+y))}
=y^3+x(3y^3+y^4)+O(x^2).
\]
The coefficient $[xy^3]=3$ counts $0010$, $0110$, and $0100$;
$[xy^4]=1$ counts $0101$. Here we are still counting a single word.
For a pair, the same-transition or flip-transition gaps must additionally
satisfy Lemma~\ref{lem:local-LCS}.

We next combine the two codewords while keeping their witness fixed.
For $m=n-d\ge1$, let $M_1(n,d;r)$ count triples
$(\vw,\vc_1,\vc_2)$ where $\vc_i\in\cU(n,r)$, $|\vw|=m$,
$w_1=0$, and the leftmost embeddings have at least one empty gap before
every subsequent match. Codewords are ordered and may be equal.
For a same transition, all pairs of gaps contribute
$\Gsame(x_1,y_1)\Gsame(x_2,y_2)$; exclude the pairs with both gaps
nonempty by subtracting $\Gsplus(x_1,y_1)\Gsplus(x_2,y_2)$.
Do the analogous subtraction for flip transitions. These differences
have nonnegative coefficients because they enumerate the remaining choices.

\begin{theorem}[The two-codeword transfer-matrix generating function]
\label{lem:pair-gf}\label{thm:pair-gf}
Set
\begin{align*}
\mathcal S_1&=\Gsame(x_1,y_1)\Gsame(x_2,y_2)
 -\Gsplus(x_1,y_1)\Gsplus(x_2,y_2),\\
\mathcal F_1&=\Gflip(x_1,y_1)\Gflip(x_2,y_2)
 -\Gfplus(x_1,y_1)\Gfplus(x_2,y_2),\\
\mathcal T_{\mathrm{del}}^{(1)}
 &=\begin{pmatrix}\mathcal S_1&\mathcal F_1\\
                  \mathcal F_1&\mathcal S_1\end{pmatrix},\qquad
\mathcal Q_1=\mathcal S_1+\mathcal F_1,\\
\mathcal A_1&=y_1y_2\Ttail(x_1,y_1)\Ttail(x_2,y_2).
\end{align*}
The rows and columns of $\mathcal T_{\mathrm{del}}^{(1)}$ are indexed
by the previous and next witness symbols, in the order $0,1$.
With $\ve_0=(1,0)^{\mathsf T}$ and $\vone=(1,1)^{\mathsf T}$,
\begin{align}\label{eq:pair-gf}
M_1(n,d;r)
& =[x_1^dx_2^dy_1^ry_2^r]\,
\mathcal A_1\ve_0^{\mathsf T}
 (\mathcal T_{\mathrm{del}}^{(1)})^{n-d-1}\vone\notag\\
&=[x_1^dx_2^dy_1^ry_2^r]\,
\mathcal A_1\mathcal Q_1^{n-d-1},
\end{align}
for $0\le d<n$ and $1\le r\le n$. Moreover,
\begin{equation}\label{eq:pair-witness-probability}
\Pr(E_{\{1,2\}})\le\frac{M_1(n,d;r)}{|\cU(n,r)|^2}.
\end{equation}
\end{theorem}
\begin{proof}
Use a separate pair $(x_i,y_i)$ for each codeword. A witness starts at
symbol $0$ and then makes $n-d-1$ same or flip transitions. At each
transition the two words use the same type, with independently chosen
gap lengths except that the all-nonempty choice is excluded. The
corresponding differences of products are $\mathcal S_1$ and
$\mathcal F_1$. Initial matches and the unrestricted suffixes give the
boundary contribution $\mathcal A_1$.

Matrix multiplication sums over the successive witness symbols.
The final vector $\vone$ allows either last symbol. Both row sums
are $\mathcal Q_1$, so
$(\mathcal T_{\mathrm{del}}^{(1)})^{n-d-1}\vone
=\mathcal Q_1^{n-d-1}\vone$, proving the second representation.
Conversely, the witness transitions, permitted gaps, and suffixes
reconstruct the two words and their leftmost embeddings uniquely.
Extracting $x_i^dy_i^r$ imposes length $n$ and $r$ runs in each word.
Lemma~\ref{lem:local-LCS} supplies at least one counted witness for
every bad pair; division by $|\cU(n,r)|^2$ proves the probability bound.
\end{proof}
The identity counts restricted witness--pair configurations, not
distinct bad pairs or all LCSs. Different witnesses for one pair are
permitted. The boundary contribution occurs once, while the matrix
is applied once per witness transition.

{
After diagonal substitution define
\begin{equation}\label{eq:pair-transfer-matrix}
\Tdel{1}(x,y)=
\begin{pmatrix}S_1&F_1\\F_1&S_1\end{pmatrix},\qquad
S_1=\Gsame^2-(\Gsplus)^2,\quad
F_1=\Gflip^2-(\Gfplus)^2.
\end{equation}
The local factors in this display are evaluated at $(x,y)$. For
$x\in(0,1)$ and $y>0$, both $S_1$ and $F_1$ are positive. The two
eigenvalues are $S_1+F_1$ and $S_1-F_1$, so
\begin{equation}\label{eq:pair-kernel}
\ldel{1}(x,y)=\lambda(\Tdel{1}(x,y))
=S_1+F_1=1+y^2\frac{1+3x}{1-x}.
\end{equation}
Thus the scalar governing the transfer-matrix powers is explicit.}

The following finite example illustrates the restricted count.
\label{ex:finite-witness-count}
For $n=4$, $d=1$, and $r=2$, the ensemble is
$\cU(4,2)=\{0001,0011,0111\}$. The restricted witnesses $000$, $001$,
and $011$ contribute respectively $1$, $3$, and $3$ ordered pairs.
For example, $001$ is allowed for $(0001,0011)$, $(0011,0001)$,
and $(0011,0011)$; for $(0001,0001)$ its final match has two nonempty
gaps, so that configuration is excluded. Hence
$[x_1x_2y_1^2y_2^2]\mathcal A_1\mathcal Q_1^2=M_1(4,1;2)=7$.

\subsection{Achievable rates and the Levenshtein improvement}
\label{sec:positive-coefficients}

We now substitute positive real values for the formal variables.
If $F(\vu)=\sum_{\vj}f_{\vj}\vu^{\vj}$ has nonnegative
coefficients and converges at $\va>0$, each weighted term is
at most the sum, so
\begin{equation}\label{eq:positive-coefficient-principle}
[\vu^{\vj}]F(\vu)\le F(\va)\va^{-\vj}.
\end{equation}
The tail series $\Ttail(x,y)$ requires $x(1+y)<1$.
For positive $x,y$, this also implies $x<1$, ensuring
convergence of the gap factors. Accordingly, let
\begin{equation}\label{eq:del-domain}
\Ddomain=\{(x,y)\in\R_{>0}^2:x(1+y)<1\}.
\end{equation}
For any fixed $(x,y)\in\Ddomain$, diagonal evaluation in
\eqref{eq:pair-gf} gives
\begin{equation}\label{eq:pair-coefficient-bound}
M_1(n,d;r)
\le A_1(x,y)\ldel{1}(x,y)^{n-d-1}x^{-2d}y^{-2r},
\qquad
A_1(x,y)=\left(\frac{y}{1-x(1+y)}\right)^2.
\end{equation}

For $r=\lfloor\rho n\rfloor$ and $d=\lfloor\delta n\rfloor$,
define
\[
\Phi_{1,\delta,\rho}(x,y)
=(1-\delta)\log_2\ldel{1}(x,y)
-2\delta\log_2x-2\rho\log_2y.
\]
Then \eqref{eq:pair-coefficient-bound} yields
\begin{equation}\label{eq:pair-coefficient-bound-asymptotic}
M_1(n,d;r)
\le 2^{n\Phi_{1,\delta,\rho}(x,y)+o(n)}
\end{equation}
for every fixed $(x,y)\in\Ddomain$.

We are now ready to prove our first result and demonstrate
a strict improvement over Levenshtein's bound.

\begin{proof}[Proof of Theorem~\ref{thm:unique-main}]
Fix $0<\rho<1$, and let $r=\lfloor\rho n\rfloor$
and $d=\lfloor\delta n\rfloor$.
At any fixed feasible point $(x,y)$,
\eqref{eq:pair-witness-probability},
\eqref{eq:pair-coefficient-bound-asymptotic}, and
$|\cU(n,r)|=2^{nh_2(\rho)+o(n)}$ give
\[
\Pr(E_{\{1,2\}})
\le
2^{-n(2h_2(\rho)-\Phi_{1,\delta,\rho}(x,y)-o(1))}.
\]
Applying Lemma~\ref{lem:sampling} with the usual
length-$(n-d)$ subsequences as witnesses and $L=1$
gives every positive rate below
$2h_2(\rho)-\Phi_{1,\delta,\rho}(x,y)$.

For fixed $x,y$, we have
\[
\sup_{0<\rho<1}
\{h_2(\rho)+\rho\log_2 y\}
=\log_2(1+y),
\]
with the maximum attained at $\rho=y/(1+y)$.
Optimizing over $\rho$ and then $(x,y)$ proves
\[
\cR_{2,\mathrm{del}}^{(1)}(\delta)
\ge \pos{\underline R_{\mathrm{del}}^{(1)}(\delta)}.
\]

For the strict comparison, $0<\delta<1/2$ makes
$(x,y)=(\delta,1)$ feasible. Direct substitution gives
\[
\ldel{1}(\delta,1)
=\frac{2(1+\delta)}{1-\delta}
=\frac{2(1-\delta^2)}{(1-\delta)^2}.
\]
Consequently,
\begin{equation}\label{eq:unique-explicit-gain}
\underline R_{\mathrm{del}}^{(1)}(\delta)
\ge
1+\delta-2h_2(\delta)
-(1-\delta)\log_2(1-\delta^2)
>R_{\mathrm{Lev}}(\delta).
\end{equation}
The last inequality follows from $0<1-\delta^2<1$.
\end{proof}

\begin{remark}[The half-run-density bound]
\label{rem:uniform-rate}
The explicit improvement in
\eqref{eq:unique-explicit-gain} uses $y=1$,
which corresponds to run density $\rho=1/2$.
Thus, the strict improvement over Levenshtein's bound
already holds without varying the run density.
Allowing other densities gives stronger rate evaluations,
as illustrated by the numerical example in
\eqref{eq:intro-numerical-improvement}.
\end{remark}

\subsection{Size of the improvement}\label{sec:unique-numerics}
At $\delta=1/10$, use the rational evaluation parameters
\begin{equation}\label{eq:unique-rational-point}
x=\frac4{27},\qquad y=\frac57,\qquad \rho=\frac5{12}.
\end{equation}
They satisfy $x(1+y)=16/63<1$, $\rho=y/(1+y)$, and
$\ldel{1}(x,y)=2102/1127$. The bound \eqref{eq:unique-main-bound} from our main theorem therefore gives the achievable rate
\begin{align}\label{eq:unique-rational-rate}
R_{\mathrm{cert}}
&=2\log_2\frac{12}{7}+\frac15\log_2\frac4{27}
-\frac9{10}\log_2\frac{2102}{1127}\notag\\
&=0.194890018483448849728\ldots.
\end{align}
In particular,
$\cR_{2,\mathrm{del}}^{(1)}(0.1)\ge0.194890$, compared with
$R_{\mathrm{Lev}}(0.1)=0.162008812821\ldots$, an increase of
approximately $20.30\%$. 
For comparison, solving the stationary equations in Theorem~\ref{thm:deletion-stationary} with $L=1$ gives an admissible candidate near
\[
(x^*,y^*,\rho^*)
=(0.147863161940,\ 0.712243774234,\ 0.415971011226),
\]
yielding a rate of approximately $0.194894751421$.

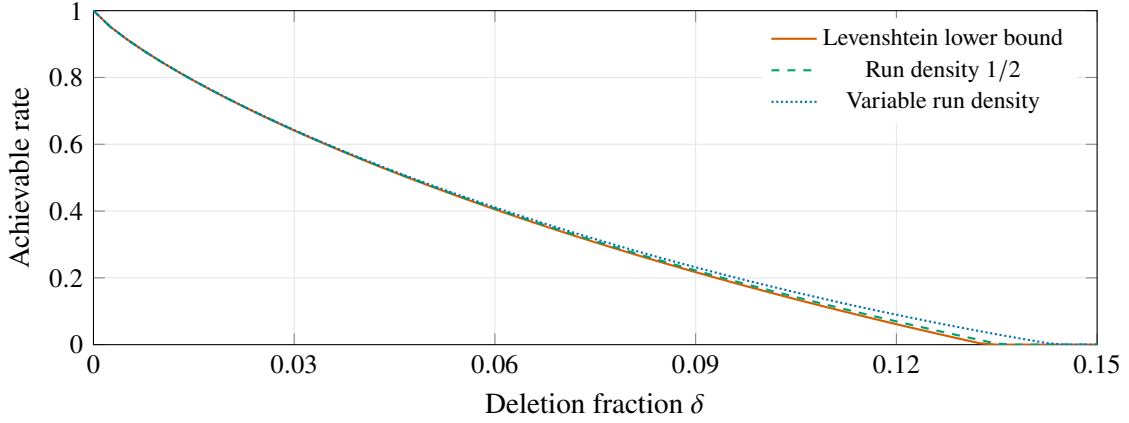
\begin{figure}[tbp]
\centering
\begin{tikzpicture}
\begin{axis}[
 width=.90\textwidth,height=6.0cm,
 xlabel={Deletion fraction $\delta$},ylabel={Achievable rate},
 xmin=0,xmax=.17,ymin=0,ymax=1,
 xtick={0,.04,.08,.12,.16},
 scaled ticks=false, tick label style={font=\small},
 xticklabel style={/pgf/number format/fixed,/pgf/number format/precision=2},
 legend style={at={(.98,.98)},anchor=north east,font=\footnotesize,draw=none},
 grid=major,major grid style={gray!20}]
\addplot[BaselineOrange,thick] coordinates {
(0,1)
(0.0025,0.952076268214)
(0.005,0.914170615332)
(0.0075,0.880057524290)
(0.01,0.848413728208)
(0.0125,0.818610787875)
(0.015,0.790278579801)
(0.0175,0.763172554660)
(0.02,0.737118914916)
(0.0225,0.711987724251)
(0.025,0.687678137006)
(0.0275,0.664109600473)
(0.03,0.641216284336)
(0.0325,0.618943380269)
(0.035,0.597244546921)
(0.0375,0.576080088617)
(0.04,0.555415621835)
(0.0425,0.535221076343)
(0.045,0.515469932128)
(0.0475,0.496138626402)
(0.05,0.477206085768)
(0.0525,0.458653352176)
(0.055,0.440463280278)
(0.0575,0.422620289908)
(0.06,0.405110161691)
(0.0625,0.387919866765)
(0.065,0.371037423798)
(0.0675,0.354451778018)
(0.07,0.338152698199)
(0.0725,0.322130688375)
(0.075,0.306376911747)
(0.0775,0.290883124741)
(0.08,0.275641619595)
(0.0825,0.260645174111)
(0.085,0.245887007522)
(0.0875,0.231360741541)
(0.09,0.217060365871)
(0.0925,0.202980207554)
(0.095,0.189114903625)
(0.0975,0.175459376660)
(0.1,0.162008812821)
(0.1025,0.148758642109)
(0.105,0.135704520537)
(0.1075,0.122842314003)
(0.11,0.110168083670)
(0.1125,0.097678072670)
(0.115,0.085368693990)
(0.1175,0.073236519418)
(0.12,0.061278269425)
(0.1225,0.049490803893)
(0.125,0.037871113600)
(0.1275,0.026416312389)
(0.13,0.015123629944)
(0.1325,0.003990405133)
(0.135,0.000000000000)
(0.1375,0.000000000000)
(0.14,0.000000000000)
(0.1425,0.000000000000)
(0.145,0.000000000000)
(0.1475,0.000000000000)
(0.15,0.000000000000)
(0.1525,0.000000000000)
(0.155,0.000000000000)
(0.1575,0.000000000000)
(0.16,0.000000000000)
(0.1625,0.000000000000)
(0.165,0.000000000000)
(0.1675,0.000000000000)
(0.17,0.000000000000)
};
\addlegendentry{Levenshtein lower bound}
\addplot[HalfTeal,thick,dashed] coordinates {
(0,1)
(0.0025,0.952089781968)
(0.005,0.914224622365)
(0.0075,0.880178927736)
(0.01,0.848629348839)
(0.0125,0.818947358021)
(0.015,0.790762737162)
(0.0175,0.763830836019)
(0.02,0.737977749745)
(0.0225,0.713073428231)
(0.025,0.689016905450)
(0.0275,0.665727501660)
(0.03,0.643139252766)
(0.0325,0.621197209838)
(0.035,0.599854884030)
(0.0375,0.579072425204)
(0.04,0.558815288361)
(0.0425,0.539053234716)
(0.045,0.519759568582)
(0.0475,0.500910544342)
(0.05,0.482484898581)
(0.0525,0.464463476022)
(0.055,0.446828926865)
(0.0575,0.429565459266)
(0.06,0.412658634944)
(0.0625,0.396095198926)
(0.065,0.379862936579)
(0.0675,0.363950552682)
(0.07,0.348347568450)
(0.0725,0.333044233310)
(0.075,0.318031448867)
(0.0775,0.303300703046)
(0.08,0.288844012761)
(0.0825,0.274653873776)
(0.085,0.260723216672)
(0.0875,0.247045368028)
(0.09,0.233614016066)
(0.0925,0.220423180138)
(0.095,0.207467183549)
(0.0975,0.194740629264)
(0.1,0.182238378142)
(0.1025,0.169955529377)
(0.105,0.157887402875)
(0.1075,0.146029523343)
(0.11,0.134377605891)
(0.1125,0.122927542973)
(0.115,0.111675392516)
(0.1175,0.100617367125)
(0.12,0.089749824223)
(0.1225,0.079069257057)
(0.125,0.068572286461)
(0.1275,0.058255653312)
(0.13,0.048116211609)
(0.1325,0.038150922114)
(0.135,0.028356846497)
(0.1375,0.018731141956)
(0.14,0.009271056243)
(0.1425,0.000000000000)
(0.145,0.000000000000)
(0.1475,0.000000000000)
(0.15,0.000000000000)
(0.1525,0.000000000000)
(0.155,0.000000000000)
(0.1575,0.000000000000)
(0.16,0.000000000000)
(0.1625,0.000000000000)
(0.165,0.000000000000)
(0.1675,0.000000000000)
(0.17,0.000000000000)
};
\addlegendentry{Run density $1/2$}
\addplot[RunBlue,thick,densely dotted] coordinates {
(0,1)
(0.0025,0.952094347233)
(0.005,0.914243114753)
(0.0075,0.880221065211)
(0.01,0.848705218003)
(0.0125,0.819067427135)
(0.015,0.790937869695)
(0.0175,0.764072304745)
(0.02,0.738297251430)
(0.0225,0.713483098934)
(0.025,0.689529336480)
(0.0275,0.666355756239)
(0.03,0.643896883430)
(0.0325,0.622098276638)
(0.035,0.600913973577)
(0.0375,0.580304670631)
(0.04,0.560236390257)
(0.0425,0.540679483110)
(0.045,0.521607866066)
(0.0475,0.502998430410)
(0.05,0.484830575293)
(0.0525,0.467085835103)
(0.055,0.449747578364)
(0.0575,0.432800761909)
(0.06,0.416231728346)
(0.0625,0.400028037827)
(0.065,0.384178327306)
(0.0675,0.368672192054)
(0.07,0.353500085382)
(0.0725,0.338653233387)
(0.075,0.324123562226)
(0.0775,0.309903635914)
(0.08,0.295986603079)
(0.0825,0.282366151378)
(0.085,0.269036468583)
(0.0875,0.255992209502)
(0.09,0.243228468107)
(0.0925,0.230740754365)
(0.095,0.218524975401)
(0.0975,0.206577420734)
(0.1,0.194894751421)
(0.1025,0.183473993077)
(0.105,0.172312532812)
(0.1075,0.161408120262)
(0.11,0.150758873044)
(0.1125,0.140363287100)
(0.115,0.130220252621)
(0.1175,0.120329076516)
(0.12,0.110689512747)
(0.1225,0.101301802360)
(0.125,0.092166725718)
(0.1275,0.083285670509)
(0.13,0.074660720549)
(0.1325,0.066294772739)
(0.135,0.058191693124)
(0.1375,0.050356528859)
(0.14,0.042795802888)
(0.1425,0.035517935853)
(0.145,0.028533873305)
(0.1475,0.021858064422)
(0.15,0.015510090433)
(0.1525,0.009517625109)
(0.155,0.003922569646)
(0.1575,0.000000000000)
(0.16,0.000000000000)
(0.1625,0.000000000000)
(0.165,0.000000000000)
(0.1675,0.000000000000)
(0.17,0.000000000000)
};
\addlegendentry{Variable run density}
\end{axis}
\end{tikzpicture}
\caption{Unique-decoding lower bounds, with negative expressions replaced
by zero. The locally restricted witness count is evaluated at density
$1/2$ and at variable density; the curves use feasible points.}
\label{fig:unique-rate-comparison}
\end{figure}

\section{Achievable Rates for Deletion List Decoding}\label{sec:deletions}
We now let $L\ge1$ be arbitrary. Lemma~\ref{lem:local-LCS} supplies
locally restricted witnesses for $L+1$ codewords. Their joint generating
function gives the rate bound and an explicit strict improvement over
the uniform local-lemma baseline. Admissible stationary points provide
further rate evaluations; fixed-density coefficient asymptotics are
proved separately in the appendix.

\subsection{Locally restricted common-witness counting}
Put $m=n-d\ge1$. Let $M_L(n,d;r)$ count configurations
$(\vw,\vc_1,\ldots,\vc_{L+1})$ in which $\vc_i\in\cU(n,r)$,
$\vw=w_1\cdots w_m$ is common to all words, $w_1=0$, and the leftmost
embeddings have at least one empty gap before every subsequent match.
The tuple is ordered and repetitions are permitted. Equivalently, if
$\mathcal W_m(\vc_1,\ldots,\vc_{L+1})$ is this family of restricted
witnesses, then
\begin{equation}\label{eq:ML-definition}
M_L(n,d;r)=
\sum_{(\vc_1,\ldots,\vc_{L+1})\in\cU(n,r)^{L+1}}
|\mathcal W_m(\vc_1,\ldots,\vc_{L+1})|.
\end{equation}
For independent uniform samples and a fixed $(L+1)$-set of indices $S$,
Lemma~\ref{lem:local-LCS} gives
\begin{equation}\label{eq:deletion-bad-probability}
\Pr(E_S)\le\frac{M_L(n,d;r)}{|\cU(n,r)|^{L+1}}.
\end{equation}
The restricted witnesses cover every bad tuple, but a tuple can still
have more than one of them. The condition on their gaps is joint:
for $L>1$ it is not a pairwise constraint.

\begin{theorem}[The locally restricted common-witness generating function]
\label{thm:list-witness-gf}\label{cor:list-witness-rational-gf}
For $\vx=(x_1,\ldots,x_{L+1})$ and $\vy=(y_1,\ldots,y_{L+1})$, define
\begin{equation}\label{eq:AL-definition}
\mathcal A_L(\vx,\vy)=\prod_{i=1}^{L+1}
\frac{y_i}{1-x_i(1+y_i)},
\end{equation}
\begin{align}\label{eq:KL-definition}
\mathcal S_L(\vx,\vy)
&=\prod_{i=1}^{L+1}\Gsame(x_i,y_i)-\prod_{i=1}^{L+1}\Gsplus(x_i,y_i),\notag\\
\mathcal F_L(\vx,\vy)
&=\prod_{i=1}^{L+1}\Gflip(x_i,y_i)-\prod_{i=1}^{L+1}\Gfplus(x_i,y_i),\\
\mathcal T_{\mathrm{del}}^{(L)}(\vx,\vy)
&=\begin{pmatrix}\mathcal S_L&\mathcal F_L\\
                 \mathcal F_L&\mathcal S_L\end{pmatrix},\qquad
\mathcal Q_L=\mathcal S_L+\mathcal F_L.\notag
\end{align}
For $0\le d<n$ and $1\le r\le n$,
\begin{equation}\label{eq:deletion-gf-power}
M_L(n,d;r)=
\left[\prod_{i=1}^{L+1} x_i^dy_i^r\right]
\mathcal A_L(\vx,\vy)\mathcal Q_L(\vx,\vy)^{n-d-1}.
\end{equation}
The series $\mathcal A_L$ and the matrix entries have nonnegative
coefficients. {
Equivalently, the scalar power in this identity is
$\ve_0^{\mathsf T}(\mathcal T_{\mathrm{del}}^{(L)})^{n-d-1}\vone$.}
\end{theorem}

\begin{proof}
Apply the proof of Theorem~\ref{thm:pair-gf} to $L+1$ codewords, using
one pair of variables for each word. The first matches and unrestricted
suffixes give \eqref{eq:AL-definition}. For each of the $n-d-1$
shared transitions, remove the choices in which all $L+1$ gaps are
nonempty, separately for the same and flip types. This gives
\eqref{eq:KL-definition} and proves its coefficient nonnegativity.
The reconstruction by leftmost embeddings is again bijective, and
coefficient extraction imposes the individual lengths and run counts.
\end{proof}
{
The boundary contribution handles the initial matches and suffixes;
the transfer matrix is applied once per subsequent witness symbol.
Its two row sums equal $\mathcal Q_L$. If $z$ records witness
length, the formal identity
$\sum_{\ell\ge1}z^\ell\mathcal A_L\mathcal Q_L^{\ell-1}
=z\mathcal A_L/(1-z\mathcal Q_L)$ also gives
\begin{equation}\label{eq:deletion-gf-rational}
M_L(n,d;r)=
\left[z^{n-d}\prod_{i=1}^{L+1} x_i^dy_i^r\right]
\frac{z\mathcal A_L(\vx,\vy)}{1-z\mathcal Q_L(\vx,\vy)}.
\end{equation}
We use the large-power form \eqref{eq:deletion-gf-power} below.
On the diagonal, set $\Tdel{L}(x,y)
=\mathcal T_{\mathrm{del}}^{(L)}(x,\ldots,x,y,\ldots,y)$.
Its two positive transition weights are
$S_L^{\mathrm{del}}=\Gsame^{L+1}-(\Gsplus)^{L+1}$ and
$F_L^{\mathrm{del}}=\Gflip^{L+1}-(\Gfplus)^{L+1}$.
Consequently its Perron eigenvalue is their sum:
\begin{align}\label{eq:KL-symmetric}
\ldel{L}(x,y)
={}&\left(1+\frac{xy^2}{1-x}\right)^{L+1}
-\left(\frac{xy^2}{1-x}\right)^{L+1} 
+\left(\frac y{1-x}\right)^{L+1}
-\left(\frac{xy}{1-x}\right)^{L+1}.
\end{align}}

\subsection{Achievable rates and strict improvement}
\label{sec:deletion-evaluation}\label{sec:deletion-certificates}The boundary contribution is
$A_L(x,y)=(y/(1-x(1+y)))^{L+1}$, so the common positive-evaluation
domain remains $\Ddomain$ in \eqref{eq:del-domain}. 
For $L\ge1$, $0<\delta<1$, $0<\rho<1$, and $(x,y)\in\Ddomain$, set
\begin{align}
\Phi_{L,\delta,\rho}(x,y)
&=(1-\delta)\log_2\ldel{L}(x,y)
-(L+1)\delta\log_2x-(L+1)\rho\log_2y,
\label{eq:general-phi}\\
\widehat\beta_L(\delta,\rho)
&=\inf_{(x,y)\in\Ddomain}\Phi_{L,\delta,\rho}(x,y).
\label{eq:beta-envelope}
\end{align}
The variables $x,y$ are evaluation parameters, whereas $\rho$ is the
run density of the sampled codewords. The quantity $\widehat\beta_L$
bounds the exponent of the locally restricted incidence count.

\begin{proposition}[Direct deletion rate bound]
\label{thm:deletion-lower-bound}\label{prop:deletion-certificates}
Fix $L\ge1$ and $0<\delta<1$. For every $\rho\in(0,1)$ and
$(x,y)\in\Ddomain$, the nonnegative part of
\begin{equation}\label{eq:del-point-certificate}
\frac{(L+1)h_2(\rho)-\Phi_{L,\delta,\rho}(x,y)}{L}
\end{equation}
is a lower bound on $\cR_{2,\mathrm{del}}^{(L)}(\delta)$.
Consequently, defining
\begin{equation}\label{eq:deletion-evaluation-bound}
\RdelEval(\delta)=\sup_{0<\rho<1}
\frac{(L+1)h_2(\rho)-\widehat\beta_L(\delta,\rho)}{L},
\end{equation}
we have $\cR_{2,\mathrm{del}}^{(L)}(\delta)\ge\pos{\RdelEval(\delta)}$.
\end{proposition}
\begin{proof}
Fix $\rho$ and $(x,y)\in\Ddomain$, independently of $n$, and set
$d_n=\lfloor\delta n\rfloor$, $r_n=\lfloor\rho n\rfloor$.
By \eqref{eq:positive-coefficient-principle} and
\eqref{eq:deletion-gf-power},
\begin{equation}\label{eq:deletion-positive-bound}
M_L(n,d_n;r_n)\le
A_L(x,y)\ldel{L}(x,y)^{n-d_n-1}
x^{-(L+1)d_n}y^{-(L+1)r_n}.
\end{equation}
The positive boundary contribution is finite and independent of $n$. Its normalized
logarithm vanishes, giving
\[
\limsup_{n\to\infty}\frac1n\log_2M_L(n,d_n;r_n)
\le\Phi_{L,\delta,\rho}(x,y).
\]
Together with \eqref{eq:run-ensemble-size} and
\eqref{eq:deletion-bad-probability}, this yields
\[
\Pr(E_S)\le
2^{-n\{(L+1)h_2(\rho)-\Phi_{L,\delta,\rho}(x,y)-o(1)\}}.
\]
Apply Lemma~\ref{lem:sampling} using all length-$(n-d_n)$ subsequences
as the individual witness sets. The restricted incidence count was used
only to bound this ordinary bad event. The lemma supplies every
positive rate strictly below \eqref{eq:del-point-certificate}, including
the distinct-word conclusion. A singleton code supplies zero.
Take suprema over fixed evaluation points and then over $\rho$;
attainment of either optimization is unnecessary.
\end{proof}

\begin{remark}[Eliminating the design density]\label{rem:rho-optimization}
For fixed $x,y$, the density-dependent part of
\eqref{eq:del-point-certificate} is proportional to
$h_2(\rho)+\rho\log_2y$. Its derivative is
$\log_2((1-\rho)/\rho)+\log_2y$, and its second derivative is
negative. The unique maximum is $\log_2(1+y)$ at $\rho=y/(1+y)$.
Since all three parameters are chosen by the designer,
\begin{equation}\label{eq:del-two-variable}
\RdelEval(\delta)=\frac1L
\sup_{(x,y)\in\Ddomain}
\left\{(L+1)\log_2(1+y)
-(1-\delta)\log_2\ldel{L}(x,y)+(L+1)\delta\log_2x\right\}.
\end{equation}
Thus a numerical search can use two variables. Every fixed feasible
point already gives a valid lower bound, whether or not it is optimal.
\end{remark}

A single evaluation proves strict improvement over the uniform local-lemma bound.

At density $\rho=1/2$, put
\begin{equation}\label{eq:half-rate}
\Rhalf(\delta)=\frac{L+1-\widehat\beta_L(\delta,1/2)}{L}.
\end{equation}
For $L=1$, this is the half-run-density bound discussed
in Remark~\ref{rem:uniform-rate}. The next theorem quantifies the
strict gain by one explicit point rather than by solving a critical equation.

\begin{theorem}[Strict fixed-list improvement]\label{thm:fixed-L-strict}
For every finite integer $L\ge1$ and every $0<\delta<1/2$,
\begin{align}\label{eq:fixed-L-strict}
\RdelEval(\delta)\ge\Rhalf(\delta)
&\ge 1+\frac\delta L-\frac{L+1}{L}h_2(\delta)
-\frac{1-\delta}{L}\log_2(1-\delta^{L+1})\notag\\
&>1+\frac\delta L-\frac{L+1}{L}h_2(\delta).
\end{align}
The nonnegative part of the explicit middle expression is an
achievable-rate lower bound.
\end{theorem}
\begin{proof}
The point $(x,y)=(\delta,1)$ belongs to $\Ddomain$, and
\[
\ldel{L}(\delta,1)=\frac{2(1-\delta^{L+1})}{(1-\delta)^{L+1}}.
\]
Substituting this point and $\rho=1/2$ into
\eqref{eq:del-point-certificate} gives the middle expression in
\eqref{eq:fixed-L-strict}. Its last term is strictly positive because
$0<\delta^{L+1}<1$. Optimization over all densities includes $\rho=1/2$.
\end{proof}
The local lemma improves the uniform union-bound baseline as in
Corollary~\ref{cor:uniform-deletion}. The locally restricted witness
count gives the additional explicit gain in \eqref{eq:fixed-L-strict}
even in the full-exponent half-run ensemble. The strict algebraic
comparison yields a strict gain after taking nonnegative parts wherever the uniform
baseline is nonnegative; it does not alone assert positive rate where
that baseline is negative.

\subsection{Lower bounds from stationary points}\label{sec:deletion-critical}
The preceding estimate is valid at every feasible point. To obtain
candidate evaluations by solving equations, differentiate the rate
expression with respect to $\ln x$, $\ln y$, and the design density
$\rho$. The first two equations below come from
$\Phi_{L,\delta,\rho}$; the third is the entropy maximization in
Remark~\ref{rem:rho-optimization}.

\begin{theorem}[Deletion bound from an admissible stationary point]
\label{thm:deletion-stationary}
Fix an integer $L\ge1$ and $0<\delta<1$. Let
$(x^*,y^*,\rho^*)$ be any solution of
\begin{align}
x\partial_x\ln\ldel{L}(x,y)
 &=\frac{(L+1)\delta}{1-\delta},
 \label{eq:deletion-critical-x}\\
y\partial_y\ln\ldel{L}(x,y)
 &=\frac{(L+1)\rho}{1-\delta},
 \label{eq:deletion-critical-y}\\
\rho&=\frac{y}{1+y},
 \label{eq:deletion-critical-rho}
\end{align}
satisfying
\begin{equation}\label{eq:deletion-critical-admissibility}
x>0,\qquad y>0,\qquad x(1+y)<1.
\end{equation}
Then
\begin{equation}\label{eq:deletion-stationary-bound}
\cR_{2,\mathrm{del}}^{(L)}(\delta)
\ge\left[
\frac{(L+1)h_2(\rho^*)-\Phi_{L,\delta,\rho^*}(x^*,y^*)}{L}
\right]_+.
\end{equation}
\end{theorem}
\begin{proof}
The admissibility conditions give $(x^*,y^*)\in\Ddomain$ and
$0<\rho^*<1$. Apply Proposition~\ref{thm:deletion-lower-bound}
at this triple.
\end{proof}
Every admissible solution gives a valid lower bound; neither
uniqueness nor global optimality is required. The theorem does not
assert existence for every pair $(L,\delta)$, and an evaluated
expression need not be positive. When several solutions are available,
their evaluated lower bounds may be compared. All logarithmic
derivatives in these equations use the natural logarithm.

At a selected solution, its raw value can also be written as
\[
\frac1L\left\{(L+1)\log_2(1+y^*)
 +(L+1)\delta\log_2x^*
 -(1-\delta)\log_2\ldel{L}(x^*,y^*)\right\}.
\]
This is an identity for the same evaluation point, not an assertion
that the point attains the supremum in \eqref{eq:del-two-variable}.
The latter expression retains all feasible, possibly nonstationary,
evaluations and in particular justifies the explicit point used in
Theorem~\ref{thm:fixed-L-strict}.

If $\rho$ is fixed in advance, the first two equations have at most
one admissible solution and it minimizes $\Phi_{L,\delta,\rho}$ at
that density. Theorem~\ref{thm:deletion-beta} in
Appendix~\ref{app:deletion-critical} proves this separate statement
and identifies the associated coefficient exponent. Its uniqueness
claim fixes $\rho$; it is not a uniqueness claim for the joint system
\eqref{eq:deletion-critical-x}--\eqref{eq:deletion-critical-rho}.

\section{Achievable Rates for Insertion List Decoding}\label{sec:insertions}
We again sample uniformly from words with a fixed number of runs. For
insertions, a shortest common supersequence supplies a different local
restriction: every witness position is used by at least one codeword.
We impose this restriction in a finite-state generating function that
records each word's length, run count, and termination. Positive
coefficient evaluation proves the rate bound; a spectral reduction then
makes it computable using a matrix of size $\lfloor(L+1)/2\rfloor+1$.
Throughout this section, $L\ge1$ is fixed and $0<\gamma<1$.

\subsection{Covered supersequences and finite-state counting}
\label{sec:insertion-witnesses}
Use the ensemble $\cU(n,r)$ from \eqref{eq:run-ensemble}. A uniform
sample is obtained by choosing $r-1$ of the $n-1$ intersymbol boundaries
uniformly, starting with $0$, and flipping precisely at those boundaries.
Different candidates are sampled independently, with replacement.
For $d=\lfloor\gamma n\rfloor$, a tuple of $L+1$ candidates is bad if it
has a common supersequence of length at most $n+d$.

\begin{lemma}[Coverage by a shortest common supersequence]
\label{lem:scs-coverage}
Let $\vc_1,\ldots,\vc_{L+1}$ be nonempty finite words, with repetitions
allowed, and let $\vw$ be a shortest common supersequence. For any
choice of embeddings of all the words into $\vw$, every position of
$\vw$ is used by at least one embedding. In particular, this holds
for their leftmost embeddings.
\end{lemma}
\begin{proof}
If a position were unused by every chosen embedding, deleting it and
shifting later embedding positions would preserve all the words as
subsequences. This contradicts the minimal length of $\vw$.
\end{proof}

Call a configuration $(\vw,\vc_1,\ldots,\vc_{L+1})$ \emph{covered} if all
the words embed in $\vw$ and their leftmost embeddings cover every
position. Let $W_{L+1}(n,r;\ell)$ count covered configurations with
$|\vw|=\ell$ and $\vc_i\in\cU(n,r)$ for every $i$. The tuple is
ordered, and equal codewords are allowed. If $q_{n,d,r,L+1}$ is the
bad-tuple probability under this uniform sampling, then
\begin{equation}\label{eq:covered-insertion-probability}
q_{n,d,r,L+1}\le
\frac{\displaystyle\sum_{\ell=n}^{n+d}W_{L+1}(n,r;\ell)}
{\binom{n-1}{r-1}^{L+1}}.
\end{equation}
Indeed, every bad tuple has an actual shortest witness of some length
$n\le\ell\le n+d$, and Lemma~\ref{lem:scs-coverage} makes it a counted
configuration. A covered witness need not be shortest, and one tuple
may have several covered witnesses; this only enlarges the numerator.
We must sum over the possible lengths, since extending a shortest
witness to length $n+d$ need not preserve coverage. Also $w_1=0$ in
every covered configuration: all input words start with $0$, so an
initial $1$ would be unused. No additional factor $2$ or $2^\ell$
is needed in \eqref{eq:covered-insertion-probability}.

\label{sec:insertion-generating-function}
We encode a covered configuration by finite states and monomial weights.
For word $i$, use the states $0,1,\dagger$. The first two indicate its
next required symbol, and $\dagger$ means that the finite word has
ended; it is not an additional alphabet symbol. The variable $x_i$
marks the word's length, $y_i$ marks its number of runs, and $z$ marks
the witness length. Thus $x_i$ here counts matched word symbols, not
the unmatched symbols counted in Sections~\ref{sec:unique}
and~\ref{sec:deletions}. The initial state is $0$, with initial
weight $y_i$ for the first run.

Scan one witness symbol $a$. If the current requirement equals $a$,
consume that word symbol and choose its continuation: the same next
symbol contributes $x_i$, the opposite next symbol contributes
$x_i y_i$, and termination contributes $x_i$. If the requirement is
different from $a$, skip the witness position and keep the same
requirement, with weight $1$. An ended word also skips with weight $1$.
These are combinatorial choices, not transition probabilities.
With rows indexing old states and columns new states in the order
$0,1,\dagger$, the matrices are
\begin{equation}\label{eq:finite-word-matrices}
\mathsf G_{0,i}=\begin{pmatrix}
x_i&x_i y_i&x_i\\0&1&0\\0&0&1
\end{pmatrix},\qquad
\mathsf G_{1,i}=\begin{pmatrix}
1&0&0\\x_i y_i&x_i&x_i\\0&0&1
\end{pmatrix}.
\end{equation}
A matching self-loop still consumes a symbol. The new-run factor is
charged when the next, different symbol is selected; an accepting path
must eventually consume it. Including the initial $y_i$, a completed
word of length $n_i$ with $r_i$ runs therefore has weight
$x_i^{n_i}y_i^{r_i}$. The terminal choice adds no run and occurs once.

The individual words share the same witness symbol at every step.
For fixed $a$, the choices for the $L+1$ words are combined by a
Kronecker product: its $(\vb,\vc)$ entry is
$\prod_i(\mathsf G_{a,i})_{b_i,c_i}$, for
$\vb,\vc\in\{0,1,\dagger\}^{L+1}$. The products are formed before
summing over the shared witness symbol. Set
\begin{equation}\label{eq:finite-cover-matrices}
\begin{split}
\mathsf D_0&=\operatorname{diag}(0,1,1),\qquad
\mathsf D_1=\operatorname{diag}(1,0,1),\\
\mathsf B_a(\vx,\vy)&=
\bigotimes_{i=1}^{L+1}\mathsf G_{a,i}-\mathsf D_a^{\otimes {L+1}},\qquad
\mathsf B=\mathsf B_0+\mathsf B_1.
\end{split}
\end{equation}
If no coordinate requires $a$, the tensor-product row has only its
unit diagonal entry, representing a position skipped by every word.
The subtraction removes exactly that entry. Otherwise it removes
nothing. Hence all entries have nonnegative integer coefficients, and
every retained transition consumes at least one symbol. This includes
states containing ended words: for example, $(1,\dagger,1)$ must be
suppressed when scanning $0$. The all-ended state has a zero row in
$\mathsf B$.

The resulting matrix gives the following formal generating function.
Let $\ve_{\vzero}$ and $\ve_{\vdagger}$ be standard basis vectors for
the joint states $\vzero=(0,\ldots,0)$ and
$\vdagger=(\dagger,\ldots,\dagger)$. Define the formal series
\begin{equation}\label{eq:finite-word-resolvent}
F_{L+1}(z;\vx,\vy)=
\left(\prod_{i=1}^{L+1} y_i\right)
\ve_{\vzero}^{\mathsf T}(I-z\mathsf B(\vx,\vy))^{-1}
\ve_{\vdagger},\qquad
(I-z\mathsf B)^{-1}=\sum_{\ell\ge0}z^\ell\mathsf B^\ell.
\end{equation}

\begin{theorem}[The covered-supersequence generating function]
\label{thm:insertion-gf}
For $L\ge1$, $n\ge1$, $1\le r\le n$, and $\ell\ge0$,
\begin{equation}\label{eq:covered-insertion-coefficient}
W_{L+1}(n,r;\ell)=
\left[z^\ell\prod_{i=1}^{L+1} x_i^n y_i^r\right]
F_{L+1}(z;\vx,\vy).
\end{equation}
More generally, extracting $\prod_i x_i^{n_i}y_i^{r_i}$ counts the
configurations with the separately prescribed lengths and run counts.
\end{theorem}
\begin{proof}
Expand $\mathsf B^\ell=(\mathsf B_0+\mathsf B_1)^\ell$ as an ordered
sum of products. Selecting
$\mathsf B_{w_1}\cdots\mathsf B_{w_\ell}$ specifies the witness
$\vw=w_1\cdots w_\ell$. Matrix multiplication sums over the intermediate
joint states, while each tensor entry enumerates individual word
continuations.

A covered configuration determines a unique path. Before a witness
position, record the next unmatched symbol of each word, or $\dagger$
if it is exhausted. The leftmost rule determines whether it matches or
skips. After a match, the actual next symbol or the end of that word
determines the transition in \eqref{eq:finite-word-matrices}. Coverage
ensures that no transition is removed in
\eqref{eq:finite-cover-matrices}, and all words eventually end.

Conversely, an accepting, witness-labeled path starts each word at $0$.
The continuation choices specify its successive symbols, and its unique
transition to $\dagger$ specifies its end. An active word matches
exactly when its requirement equals the witness symbol, so the
embeddings are leftmost. Every retained joint transition contains a
match. The terminal vector excludes a path with an unconsumed final
requirement. These two constructions are inverse.

Each consumed symbol contributes $x_i$, each change to the next symbol
contributes $y_i$, and the initial factor counts the first run. The
resulting weight is $z^\ell\prod_i x_i^{n_i}y_i^{r_i}$.
Extracting its coefficient proves the claim. Different witness labels
retain their multiplicity under addition; the argument does not require
$\mathsf B_0$ and $\mathsf B_1$ to commute.
\end{proof}

For a concrete example, take two words with three symbols and two runs.
For $\vc_1=001$, $\vc_2=011$, and $\vw=0011$, the accepting path is
\[
00\xrightarrow[\,x_1x_2y_2\,]{\,0\,}01
\xrightarrow[\,x_1y_1\,]{\,0\,}11
\xrightarrow[\,x_1x_2\,]{\,1\,}\dagger1
\xrightarrow[\,x_2\,]{\,1\,}\dagger\dagger.
\]
Including $z^4$ and the initial $y_1y_2$ gives
$z^4x_1^3x_2^3y_1^2y_2^2$. For the whole layer
$\cU(3,2)=\{001,011\}$, one obtains
\[
[x_1^3x_2^3y_1^2y_2^2]F_2=2z^3+4z^4+2z^5.
\]
The two equal-word pairs contribute $2z^3$. For each ordering of the
distinct pair, the covered witnesses are $0011$, $0101$, and $01101$.
The last is not shortest, illustrating the intentional witness
multiplicity in \eqref{eq:covered-insertion-probability}.

\subsection{Achievable rates by positive evaluation}
\label{sec:insertion-coding}
Write $\mathsf B(x,y)$ and $F_{L+1}(z;x,y)$ for the diagonal substitutions
$x_i=x$, $y_i=y$, and let
\begin{equation}\label{eq:finite-matrix-convergence}
\mu_{L+1}(x,y)=\lambda(\mathsf B(x,y)),\qquad
x>0,\quad y>0,\quad 0<z<1,\quad z\mu_{L+1}(x,y)<1.
\end{equation}
The spectral-radius convention is that of Section~\ref{sec:preliminaries}.
By the transfer-matrix method~\cite[Sec.~V.6]{FlajoletSedgewick2009},
the condition $z\mu_{L+1}(x,y)<1$ ensures convergence of
the series in \eqref{eq:finite-word-resolvent}.
Moreover, $F_{L+1}(z;x,y)>0$, since the one-symbol
witness $0$ contributes $z(xy)^{L+1}$.
No additional restriction on $x$ is needed,
as the entries of $\mathsf B(x,y)$ are polynomials.

Nonnegativity and \eqref{eq:covered-insertion-coefficient} give
\[
F_{L+1}(z;x,y)\ge x^{(L+1)n}y^{(L+1)r}\sum_{\ell=n}^{n+d}W_{L+1}(n,r;\ell)z^\ell
\ge x^{(L+1)n}y^{(L+1)r}z^{n+d}\sum_{\ell=n}^{n+d}W_{L+1}(n,r;\ell).
\]
Consequently,
\begin{equation}\label{eq:insertion-finite-evaluation}
q_{n,d,r,L+1}\le
\frac{F_{L+1}(z;x,y)}{\binom{n-1}{r-1}^{L+1} x^{(L+1)n}y^{(L+1)r}z^{n+d}}.
\end{equation}
The individual length and run constraints were imposed before diagonal
evaluation; they have not been replaced by total constraints in the
exact identity.

\begin{theorem}[A fixed-run insertion rate by positive evaluation]
\label{thm:insertion-positive-rate}
Fix $0<\rho<1$ and parameters satisfying
\eqref{eq:finite-matrix-convergence}, independently of $n$. Define
\begin{equation}\label{eq:insertion-point-certificate}
E_{L,\gamma}(\rho;x,y,z)=
\frac{(L+1)h_2(\rho)+(L+1)\log_2x+(L+1)\rho\log_2y+(1+\gamma)\log_2z}{L}.
\end{equation}
Every positive rate strictly below this expression is achievable by
codes contained in $\cU(n,\lfloor\rho n\rfloor)$ for all sufficiently
large block lengths, up to an arbitrarily small rate loss.
\end{theorem}
\begin{proof}
Put $r_n=\lfloor\rho n\rfloor$ and $d_n=\lfloor\gamma n\rfloor$.
The value $F_{L+1}(z;x,y)$ is a fixed finite number, and
$\log_2\binom{n-1}{r_n-1}=nh_2(\rho)+O(\log n)$. Thus
\eqref{eq:insertion-finite-evaluation} gives
\[
q_{n,d_n,r_n,L+1}\le
2^{-nL E_{L,\gamma}(\rho;x,y,z)+O(\log n)}.
\]
Apply Lemma~\ref{lem:sampling} to independent uniform samples from
$\cU(n,r_n)$. Its deduplication step produces a set of distinct
codewords with the stated list-size guarantee. Nonpositive expressions
are replaced by zero.
\end{proof}

The design density and the witness-length evaluation variable can be eliminated analytically.
For fixed $y>0$,
\[
\sup_{0<\rho<1}\{h_2(\rho)+\rho\log_2 y\}=\log_2(1+y),
\qquad \rho=\frac{y}{1+y}.
\]
For fixed $x,y$, the supremum of $\log_2z$ under
\eqref{eq:finite-matrix-convergence} is
$-\log_2\max\{1,\mu_{L+1}(x,y)\}$. Hence the theorem implies
\begin{equation}\label{eq:insertion-resolvent-rate}
\cR_{2,\mathrm{ins}}^{(L)}(\gamma)\ge
\pos{\frac1L\sup_{x,y>0}
\left\{(L+1)\log_2\bigl(x(1+y)\bigr)
-(1+\gamma)\log_2\max\{1,\mu_{L+1}(x,y)\}\right\}}.
\end{equation}
Only a supremum is taken at the boundary in $z$. For a strictly smaller
rate, choose a fixed $z$ inside the convergence domain; the proof never
evaluates the resolvent where $z\mu_{L+1}=1$.

\subsection{Reduced transfer matrices and spectral evaluation}
\label{sec:insertion-spectral}
Direct evaluation of the bound in \eqref{eq:insertion-resolvent-rate} involves computing the Perron eigenvalue $\mu_{L+1}(x,y)$ of a matrix of order $3^{L+1}$. In this subsection, by exploiting the terminal-state structure and symmetry, we reduce this to computing the Perron eigenvalue of a matrix of order $\lfloor (L+1)/2\rfloor+1$, without changing the resulting rate bound.

First consider the states in which the remaining words have not ended.
Put $t=x(1+y)$ and define
\begin{equation}\label{eq:insertion-active-matrices}
\mathsf H_0=\begin{pmatrix}x&xy\\0&1\end{pmatrix},\qquad
\mathsf H_1=\begin{pmatrix}1&0\\xy&x\end{pmatrix}.
\end{equation}
For $j\ge1$, let $\mathsf E_{\vb,\vb}$ have a single unit entry at the
indicated diagonal position, and set
\begin{equation}\label{eq:insertion-active-kernel}
\mathsf H_a^{[j]}=\mathsf H_a^{\otimes j}
-\mathsf E_{(1-a)^j,(1-a)^j},\qquad
\mathsf H^{[j]}=\mathsf H_0^{[j]}+\mathsf H_1^{[j]},\qquad
\Lambda_j(x,y)=\lambda(\mathsf H^{[j]}).
\end{equation}
These nonnegative matrices allow only nonterminal continuations.
Their squares are entrywise positive: every state can reach $0^j$
in one step, by matching and changing all current ones to zeros
(or matching zeros and keeping them if none are ones). From $0^j$,
scanning $0$ reaches any active state with positive weight.

The Perron--Frobenius theorem therefore supplies strictly positive left
and right eigenvectors at $\Lambda_j$. We call a positive right
vector $\vvec$ with $\mathsf H^{[j]}\vvec=\Lambda_j\vvec$ a
\emph{Perron vector}; $v_{\vb}$ is its scalar coordinate at state
$\vb$. For these matrices, $A\boldsymbol f\ge c\boldsymbol f$ with
$\boldsymbol f>0$ implies $\lambda(A)\ge c$, by multiplication with a
positive left Perron vector. In particular, the spectral radius is
at least the smallest row sum. A state with $h$ ones has row sum
\begin{equation}\label{eq:insertion-active-row-sums}
\begin{cases}
t^j,&h=0\text{ or }h=j,\\
t^h+t^{j-h},&0<h<j.
\end{cases}
\end{equation}
Indeed, scanning $0$ consumes $j-h$ symbols and scanning $1$ consumes
$h$; each consumed symbol has total continuation weight $x+xy=t$.
The all-skip choice is removed. Thus $\Lambda_j>1$ when $t>1$.

\begin{lemma}[Removing terminal states from the spectral calculation]
\label{lem:insertion-terminal-reduction}
For every $x,y>0$,
\begin{equation}\label{eq:insertion-block-spectrum}
\mu_{L+1}(x,y)=\max_{1\le j\le {L+1}}\Lambda_j(x,y).
\end{equation}
If $x(1+y)\ge1$, then $\Lambda_{j+1}(x,y)\ge\Lambda_j(x,y)$.
Consequently, $\mu_{L+1}(x,y)=\Lambda_{L+1}(x,y)>1$ when $x(1+y)>1$.
\end{lemma}
\begin{proof}
Group the states of $\mathsf B$ by the actual set $S\subseteq[{L+1}]$ of
nonterminal coordinates, in decreasing order of $|S|$. A transition
can only replace $S$ by a subset, so the matrix is block upper
triangular. For fixed nonempty $S$, its diagonal block forbids new
terminations; it is exactly $\mathsf H^{[|S|]}$. Ended coordinates
contribute $1$, and the empty set contributes a zero block. This proves
\eqref{eq:insertion-block-spectrum}.

For monotonicity, take a positive Perron vector $\vvec$ of
$\mathsf H^{[j]}$ and lift it to $j+1$ coordinates by
$\widetilde v_{(\vb,c)}=v_{\vb}$ for $c\in\{0,1\}$. Let
$s_a(c)=t$ if $c=a$ and $s_a(c)=1$ otherwise. The definitions give
\[
(\mathsf H_a^{[j+1]}\widetilde{\vvec})_{(\vb,c)}
=s_a(c)(\mathsf H_a^{[j]}\vvec)_{\vb}
+t\mathbf1_{\{\vb=(1-a)^j,\ c=a\}}v_{\vb}.
\]
When one of the first $j$ words matches, the extra word multiplies the
continuation weight by $s_a(c)$. If none matches, the only newly
allowed case is $c=a$, when the extra word alone covers the position.
Since $s_a(c)\ge1$, summing over $a$ gives
$\mathsf H^{[j+1]}\widetilde{\vvec}\ge\Lambda_j\widetilde{\vvec}$.
The positive-eigenvector comparison proves monotonicity. Finally use
\eqref{eq:insertion-active-row-sums}.
\end{proof}

The active matrix is invariant under permutations of the words. Group
states by their number of ones to obtain an $(L+2)\times(L+2)$ matrix
$\widehat{\mathsf H}_{L+1}$, with entries for $0\le j,h\le {L+1}$ given by
\begin{equation}\label{eq:insertion-quotient}
(\widehat{\mathsf H}_{L+1})_{j,h}
=\mathbf1_{\{j<{L+1},\ j\le h\}}
\binom{{L+1}-j}{h-j}x^{{L+1}-j}y^{h-j}
+\mathbf1_{\{j>0,\ h\le j\}}
\binom{j}{h}x^j y^{j-h}.
\end{equation}
A term is absent outside its indicated range. Scanning $0$ matches
${L+1}-j$ zeros, of which $h-j$ must select ones next; scanning $1$
matches $j$ ones, of which $h$ remain ones next. This gives the two
terms. The endpoint conditions exclude all-skip transitions.

For $0\le a\le\lfloor (L+1)/2\rfloor$, let
$C_a=\{a,{L+1}-a\}$ as a set, so that a middle class is a singleton.
Define the final quotient by
\begin{equation}\label{eq:insertion-reduced-matrix}
(\Tins(x,y))_{a,b}
=\sum_{h\in C_b}(\widehat{\mathsf H}_{L+1}(x,y))_{a,h},
\qquad 0\le a,b\le\lfloor (L+1)/2\rfloor,
\end{equation}
and put $\lins(x,y)=\lambda(\Tins(x,y))$.

\begin{proposition}[Spectral preservation under symmetry reduction]
\label{prop:insertion-symmetry}
The matrix $\widehat{\mathsf H}_{L+1}$ has spectral radius $\Lambda_{L+1}$.
Moreover,
$(\widehat{\mathsf H}_{L+1})_{j,h}=(\widehat{\mathsf H}_{L+1})_{{L+1}-j,{L+1}-h}$,
so grouping $j$ and ${L+1}-j$ reduces the calculation further to a matrix
of size $\lfloor (L+1)/2\rfloor+1$. In particular,
$\lins(x,y)=\Lambda_{L+1}(x,y)$.
\end{proposition}
\begin{proof}
For any full state with $j$ ones, the sum of its transition weights
to all states with $h$ ones is \eqref{eq:insertion-quotient}. If $P$
lifts a vector indexed by $j$ to one constant on each class, then
$\mathsf H^{[{L+1}]}P=P\widehat{\mathsf H}_{L+1}$. The quotient has a positive
square since the full active matrix does. Lifting its positive Perron
vector gives a positive eigenvector of $\mathsf H^{[{L+1}]}$, so the
eigenvalue equals $\Lambda_{L+1}$. Global bit complementation gives the
second identity; the same argument applies to the second quotient.
\end{proof}
The first two cases are the $2\times2$ matrices
\begin{align}\label{eq:insertion-small-matrices}
\mathsf T_{\mathrm{ins}}^{(1)}(x,y)
 &=\begin{pmatrix}x^2(1+y^2)&2x^2y\\2xy&2x\end{pmatrix},\notag\\
\mathsf T_{\mathrm{ins}}^{(2)}(x,y)
 &=\begin{pmatrix}x^3(1+y^3)&3x^3y(1+y)\\
                   xy+x^2y^2&x+x^2(1+2y)\end{pmatrix}.
\end{align}
More generally, $\lins$ is the largest positive root of
$\det(\lambda I-\Tins(x,y))=0$, whose degree is
$\lfloor (L+1)/2\rfloor+1$. For $L=1,2$ this is a quadratic equation.

In the positive-rate evaluation region, these reductions preserve the
required spectral radius, not the entire finite-word resolvent. The full matrix is reducible because ended words cannot
restart. We use its block spectrum and convergence, rather than claim
an exact coefficient asymptotic from its dominant eigenvalue.

\begin{theorem}[Run-restricted SCS-coverage insertion bound]
\label{thm:insertion-lower-bound}
For $L\ge1$ and $0<\gamma<1$, define
\begin{equation}\label{eq:insertion-list-lower-bound}
\Rirun(\gamma)=\frac1L
\sup_{\substack{x,y>0\\x(1+y)>1}}
\left\{(L+1)\log_2\bigl(x(1+y)\bigr)
-(1+\gamma)\log_2\lins(x,y)\right\}.
\end{equation}
Then $\cR_{2,\mathrm{ins}}^{(L)}(\gamma)\ge\pos{\Rirun(\gamma)}$.
At each fixed feasible $(x,y)$, every smaller positive rate than its
displayed evaluation is attainable in layers with density
$\rho=y/(1+y)$.
\end{theorem}
\begin{proof}
If $x(1+y)\le1$, the expression in braces in
\eqref{eq:insertion-resolvent-rate} is nonpositive. Discarding these
points does not change the positive part of the supremum. At the
remaining points, Lemma~\ref{lem:insertion-terminal-reduction} gives
$\max\{1,\mu_{L+1}\}=\Lambda_{L+1}=\lins$. For a strictly smaller rate at a fixed
$(x,y)$, choose $\rho=y/(1+y)$ and a fixed $z<1/\lins$ sufficiently
close to that boundary, then apply
Theorem~\ref{thm:insertion-positive-rate}.
\end{proof}

At an interior maximizer of the raw expression in
\eqref{eq:insertion-list-lower-bound}, the necessary stationary
conditions are
\begin{equation}\label{eq:insertion-stationary-equations}
x\partial_x\ln\lins=\frac{L+1}{1+\gamma},\qquad
 y\partial_y\ln\lins=\frac{L+1}{1+\gamma}\frac{y}{1+y}.
\end{equation}
For positive left and right Perron vectors $\vu,\vvec$ of $\Tins$,
its simple eigenvalue has derivatives
\begin{equation}\label{eq:insertion-perron-derivative}
\partial_x\lins
=\frac{\vu^{\mathsf T}(\partial_x\Tins)\vvec}
       {\vu^{\mathsf T}\vvec},
\end{equation}
and similarly for $y$. Differentiate the eigenvector equation and
multiply by $\vu^{\mathsf T}$ to obtain this identity. These equations
are useful for locating candidate parameters; they do not certify a
unique or global maximizer. Every admissible point already supplies
a lower bound through Theorem~\ref{thm:insertion-lower-bound}.

\subsection{Explicit improvement and the unique-decoding case}
\label{sec:insertion-comparison}

Recall the binary insertion capacity
$C_{\mathrm{ins}}(\gamma)
=(1+\gamma)\left[
1-h_2\!\left(\frac{\gamma}{1+\gamma}\right)
\right]$.
For a prescribed list size $L$, the union-bound argument of
Con, Doron, and Ribeiro~\cite[Eq.~(14)]{CDR2026}
gives the achievable-rate expression
\[
C_{\mathrm{ins}}(\gamma)-\frac{1+\gamma}{L+1}.
\]
Using the Lov\'asz local lemma, we improve this bound by the factor $(L+1)/L$, obtaining
\begin{equation}\label{eq:improved-CDR}
B_{\mathrm{ins},L}(\gamma)
=\frac{(L+1)C_{\mathrm{ins}}(\gamma)-(1+\gamma)}{L}.
\end{equation}
We refer to~\eqref{eq:improved-CDR} as the {\em improved CDR bound}.
Nevertheless, the next result shows that our SCS-coverage insertion bound (Theorem~\ref{thm:insertion-lower-bound}) gives a strict improvement.

\begin{theorem}[Strict fixed-list improvement of the improved CDR bound]
\label{thm:insertion-strict-improvement}
For every fixed $L\ge1$ and $0<\gamma<1$,
\begin{equation}\label{eq:insertion-strict-comparison}
\Rirun(\gamma)>B_{\mathrm{ins},L}(\gamma).
\end{equation}
The strict comparison already holds at the explicit feasible point
\begin{equation}\label{eq:insertion-reference-point}
x_0=\frac{1+\gamma}{2\gamma},\qquad
y_0=\frac{1-\gamma}{1+\gamma},\qquad
x_0(1+y_0)=\frac1\gamma>1.
\end{equation}
If $B_{\mathrm{ins},L}(\gamma)\ge0$, the rate lower bound is
strictly improved as well.
\end{theorem}
\begin{proof}
At \eqref{eq:insertion-reference-point}, $x_0=1+x_0y_0$, so each column
of either matrix in \eqref{eq:insertion-active-matrices} sums to $x_0$.
The auxiliary matrix
$\mathsf A=\mathsf H_0^{\otimes {L+1}}+\mathsf H_1^{\otimes {L+1}}$ therefore
has constant column sum $\lambda_0=2x_0^{L+1}$. For a positive right
Perron vector $\vvec$ of $\mathsf H^{[{L+1}]}$,
\[
\lins\vone^{\mathsf T}\vvec
=\vone^{\mathsf T}\mathsf H^{[{L+1}]}\vvec
=\lambda_0\vone^{\mathsf T}\vvec-v_{0^{L+1}}-v_{1^{L+1}}
<\lambda_0\vone^{\mathsf T}\vvec.
\]
Hence $\lins(x_0,y_0)<\lambda_0$. Direct expansion of binary entropy
gives
\[
\frac{(L+1)\log_2(1/\gamma)-(1+\gamma)\log_2(2x_0^{L+1})}{L}
=B_{\mathrm{ins},L}(\gamma).
\]
Replacing $2x_0^{L+1}$ by the strictly smaller $\lins$ proves the
claim. The gain at this point is
$(1+\gamma)L^{-1}\log_2(2x_0^{L+1}/\lins)>0$.
\end{proof}

For $L=1$, unique-insertion and unique-deletion correction
are equivalent. The deletion lower bound from
Theorem~\ref{thm:unique-main} is stronger numerically. 
We therefore use the deletion bound for unique-insertion
correction in Section~\ref{sec:comparison}.

\section{Upper Bounds}
\label{sec:converses}\label{sec:upper-bounds}
We now prove Singleton-, Levenshtein-, and higher-order Elias-type
converses, using the benchmark bounds from Section~\ref{sec:benchmark-converses} for comparison. Throughout this section, $q\ge2$ and $L\ge1$ are fixed.
By Proposition~\ref{prop:ins2del} in Section~\ref{sec:reduction},
deletion converses transfer to mixed errors at the effective fraction
$\tau=\delta+\gamma/L$, with the exact integer budget $t_n$ in
\eqref{eq:effective-integer-budget}. The insertion-based Elias estimate
also uses a direct common-supersequence argument.

\subsection{Singleton- and Levenshtein-type Bounds}\label{sec:lev-upper}
We begin with the prefix-counting consequence of the reduction.
\begin{corollary}[Singleton-type bound]\label{thm:singleton-insdel-list}
\label{cor:singleton-insdel-asymptotic}
Every $(\gamma,\delta,L)$-insdel list-decodable code satisfies
\begin{equation}\label{eq:singleton-insdel-finite}
|\cC_n|\le Lq^{n-t_n},
\end{equation}
with $t_n$ as in \eqref{eq:effective-integer-budget}. Consequently,
\begin{equation}\label{eq:singleton-insdel-asymptotic}
\limsup_{n\to\infty}\frac1n\log_q|\cC_n|
\le\max\{1-\delta-\gamma/L,0\}.
\end{equation}
\end{corollary}
\begin{proof}
No prefix of length $n-t_n$ can occur in more than $L$ codewords, since
such codewords would have a common subsequence of that length.
There are $q^{n-t_n}$ possible prefixes. If $\tau\ge1$, then $n-t_n=O(1)$,
which proves the zero-rate conclusion, including the rounding-sensitive
case $\tau=1$.
\end{proof}

We next generalize Levenshtein's run-counting converse to list decoding.
The following finite-length statement is formulated for pure deletions;
Proposition~\ref{prop:ins2del} applies it to mixed errors with the exact
budget $t_n$ in \eqref{eq:effective-integer-budget}.

\begin{theorem}[Levenshtein-type list-decoding upper bound]
\label{thm:levenshtein-insdel-list}
Suppose $\cC\subseteq\Sigma_q^n$ is $(0,t;L)$-list-decodable, with
$0\le t<n$. For every integer $r$ with $t\le r+1\le n$,
\begin{equation}\label{eq:levenshtein-insdel-list-finite}
|\cC|\le
\frac{Lq^{n-t}}{\displaystyle\sum_{i=0}^t\binom{r+1-t}{i}}
+q\sum_{i=0}^{r-1}\binom{n-1}{i}(q-1)^i.
\end{equation}
In particular, this applies to mixed-error codes with $t=t_n$ from
\eqref{eq:effective-integer-budget} whenever $t_n<n$.
\end{theorem}
\begin{proof}
Since $\cC$ is $(0,t;L)$-list-decodable, every
$\vy\in\Sigma_q^{n-t}$ is a subsequence of at most $L$ codewords in
$\cC$.

Split the code according to the number of runs:
\[
\cC_{\le r}\eqdef\{\vc\in\cC:\runs(\vc)\le r\},
\qquad
\cC_{\ge r+1}\eqdef\{\vc\in\cC:\runs(\vc)\ge r+1\}.
\]
A word with exactly $j+1$ runs is determined by its first symbol, its $j$
run boundaries, and the new symbol at every boundary.  Therefore
\begin{equation}\label{eq:low-run-count}
|\cC_{\le r}|
\le q\sum_{j=0}^{r-1}\binom{n-1}{j}(q-1)^j.
\end{equation}

For the high-run part, Levenshtein's deletion-shadow bound~\cite{Levenshtein2002} gives, for every
$\vx\in\Sigma_q^n$ with $\runs(\vx)\ge t$,
\begin{equation}\label{eq:lev-shadow-bound}
|\cD_t(\vx)|
\ge \sum_{i=0}^{t}\binom{\runs(\vx)-t}{i}.
\end{equation}
Consequently, every $\vc\in\cC_{\ge r+1}$ satisfies
\[
|\cD_t(\vc)|
\ge S_{r,t},
\qquad
S_{r,t}\eqdef\sum_{i=0}^{t}\binom{r+1-t}{i}.
\]
Double-counting the pairs $(\vc,\vy)$ with
$\vc\in\cC_{\ge r+1}$ and $\vy\in\cD_t(\vc)$ gives
\begin{equation}\label{eq:high-run-count}
|\cC_{\ge r+1}|S_{r,t}\le Lq^{n-t}.
\end{equation}
Combining \eqref{eq:low-run-count} and \eqref{eq:high-run-count} proves
\eqref{eq:levenshtein-insdel-list-finite}.
\end{proof}

\begin{corollary}[Asymptotic fixed-alphabet form]
\label{cor:levenshtein-insdel-asymptotic-fixed-q}
For $0<\tau<1$ and $\tau\le\rho\le1$, let
\begin{equation}\label{eq:Psi-definition}
\Psi(\rho,\tau)=
\begin{cases}
(\rho-\tau)h_2\!\left(\dfrac\tau{\rho-\tau}\right),&\rho\ge3\tau,\\[1mm]
\rho-\tau,&\tau\le\rho\le3\tau.
\end{cases}
\end{equation}
Define
\begin{equation}\label{eq:singleton-levenshtein-combined-function}
\mathsf U_q^{\mathrm{SL}}(\tau)=
\min\left\{1-\tau,
\inf_{\rho\in[\tau,1]}
\max\{1-\tau-\Psi(\rho,\tau)\log_q2,\overline H_q(\rho)\}\right\}.
\end{equation}
Every $(\gamma,\delta,L)$-insdel list-decodable family with
$\tau=\delta+\gamma/L\in(0,1)$ satisfies
\begin{equation}\label{eq:fixed-q-asymptotic-combined}
\limsup_n\frac1n\log_q|\cC_n|\le\mathsf U_q^{\mathrm{SL}}(\tau).
\end{equation}
At $\tau=0$ set $\mathsf U_q^{\mathrm{SL}}(0)=1$, and at $\tau\ge1$
set it equal to zero.
\end{corollary}
\begin{proof}
Take $t=t_n=\tau n+O(1)$ from \eqref{eq:effective-integer-budget} and choose $r=\rho n+O(1)$ with $t\le r+1\le n$ in
Theorem~\ref{thm:levenshtein-insdel-list}.  The second summand in
\eqref{eq:levenshtein-insdel-list-finite} has $q$-ary exponential rate
$\overline H_q(\rho)+o(1)$.  For the denominator in the first summand, put
$N=r+1-t=(\rho-\tau)n+o(n)$.  If $\rho\ge3\tau$, then
$t/N\le1/2+o(1)$ and
\[
\sum_{i=0}^{t}\binom{N}{i}
=2^{n((\rho-\tau)h_2(\tau/(\rho-\tau))+o(1))}.
\]
If $\tau\le\rho\le3\tau$, then the same sum has exponent
$N+o(n)$, namely
\[
\sum_{i=0}^{t}\binom{N}{i}=2^{n(\rho-\tau+o(1))}.
\]
Thus the first summand has $q$-ary exponential rate
$1-\tau-\Psi(\rho,\tau)\log_q2+o(1)$.  Taking the larger of the two
exponents and then minimizing over $\rho$ proves
the variational expression in \eqref{eq:singleton-levenshtein-combined-function}; combining it with
Corollary~\ref{cor:singleton-insdel-asymptotic} gives
\eqref{eq:fixed-q-asymptotic-combined}.
\end{proof}
The two expressions inside the maximum are respectively nonincreasing
and nondecreasing in $\rho$. When a crossing exists, the corresponding
one-dimensional minimization can be performed using
$\overline H_q(\rho)=1-\tau-\Psi(\rho,\tau)\log_q2$.
The term ``Levenshtein-type'' here refers to a \emph{converse}; it is
distinct from the achievable-rate curve improved in
Theorem~\ref{thm:unique-main}.

\subsection{Higher-order Elias-type Bounds}\label{sec:elias}

We extend the ambient-word averaging method underlying the Yasunaga
bound recalled in Section~\ref{sec:benchmark-converses}. The key new
ingredient is a higher-order intersection estimate; after deriving the
bound, we compare its insertion term with the HSS benchmark from the
same section.

Let
\[
S_L(p)=p+p^2+\cdots+p^L.
\]
Our argument follows the ambient-word averaging step of Yasunaga's
Elias-type bound~\cite{Yasunaga2024ImprovedBounds}, but replaces its
pairwise multiplicity estimate with the following higher-order lemma.

\begin{lemma}[Higher-order intersection averaging]
\label{lem:higher-order-intersection}
Let $M\ge k\ge2$ and let $A_1,\ldots,A_M\subseteq[N]$ satisfy
$|A_i|\ge\alpha N$ for $0\le\alpha\le1$. There are $k$ distinct indices
$i_1,\ldots,i_k$ such that
\begin{equation}\label{eq:higher-order-intersection}
\left|\bigcap_{j=1}^k A_{i_j}\right|
\ge\left(\alpha^k-\frac{\binom k2}{M}\right)N.
\end{equation}
\end{lemma}
\begin{proof} We assume that 
$\alpha^k>\frac{\binom{k}{2}}{M}$, since otherwise the desired conclusion follows immediately. For each $u\in\Int{N}$, let
\[
d_u\eqdef |\{i\in\Int{M}:u\in A_i\}|.
\]
Choose $J_1,\ldots,J_k$ independently and uniformly from $\Int{M}$, and
put
\[
X\eqdef\left|\bigcap_{j=1}^{k}A_{J_j}\right|.
\]
For a fixed $u$, the probability that $u$ belongs to all the selected sets
is $(d_u/M)^k$.  Hence
\begin{align}
\mathbb E[X]
&=\sum_{u=1}^{N}\left(\frac{d_u}{M}\right)^k \overset{(*)}{\geq}
N\left(\frac{1}{N}\sum_{u=1}^{N}\frac{d_u}{M}\right)^k =N\left(\frac{1}{MN}\sum_{i=1}^{M}|A_i|\right)^k
\geq \alpha^kN,\nonumber
\end{align}
where $(*)$ follows from the convexity of
$x\mapsto x^k$.

Let $\mathcal D$ be the event that $J_1,\ldots,J_k$ are pairwise distinct.
By the union bound,
\[
\Pr(\mathcal D^c)
\leq \frac{\binom{k}{2}}{M}.
\]
Since $\alpha^k>\frac{\binom{k}{2}}{M}$, we have $\Pr(\mathcal D)>0$.

Let $\mathbf 1_{\mathcal D}$ denote the indicator of the event
$\mathcal D$: it equals $1$ when $\mathcal D$ occurs and $0$ otherwise.
Thus, $\mathbb E[X\mathbf 1_{\mathcal D}]$ is the contribution to
$\mathbb E[X]$ from outcomes in $\mathcal D$.
Since $0\leq X\leq N$, 
\[X\mathbf 1_{\mathcal D^c}\leq N \mathbf 1_{\mathcal D^c}.
\]
It follows that 
\[\mathbb E[X\mathbf 1_{\mathcal D^c} ]\leq  N  \mathbb E[\mathbf 1_{\mathcal D^c}]=N\Pr[\mathcal D^c]\leq \frac{\binom{k}{2}}{M}N.\]
Hence, 
\begin{align*}
\mathbb E[X\mathbf 1_{\mathcal D}]
=\mathbb E[X]-\mathbb E[X\mathbf 1_{\mathcal D^c}]\geq
\left(\alpha^k-\frac{\binom{k}{2}}{M}\right)N.
\end{align*}
Moreover, since $\Pr(\mathcal D)>0$, the definition of conditional
expectation gives
\[
\mathbb E[X\mathbf 1_{\mathcal D}]
=
\Pr(\mathcal D)\mathbb E[X\mid\mathcal D].
\]
It follows that
\[
\mathbb E[X\mid\mathcal D]
=\frac{\mathbb E[X\mathbf 1_{\mathcal D}]}{\Pr(\mathcal D)}
\geq
\left(\alpha^k-\frac{\binom{k}{2}}{M}\right)N.
\]
Therefore some ordered $k$-tuple of distinct indices satisfies
\eqref{eq:higher-order-intersection}.
\end{proof}

\begin{theorem}[Higher-order Elias-type upper bound]
\label{thm:higher-order-elias}
Let $\cC_n\subseteq\Sigma_q^n$ be $(\gamma,\delta,L)$-insdel
list-decodable, and put $R_n=n^{-1}\log_q|\cC_n|$ and
$\tau=\delta+\gamma/L$. If $\tau\ge1$, then $\limsup_nR_n=0$.
If $\tau<1$, define
\begin{equation}\label{eq:pdel-definition}
p_{\mathrm{del}}=1-(1-\tau)^{1/L},
\qquad
p_{\mathrm{ins}}=
\sup\{p\in[0,(q-1)/q]:S_L(p)\le\gamma\}.
\end{equation}
Then
\begin{equation}\label{eq:higher-order-elias-rate}
\limsup_{n\to\infty}R_n
\le\min\{G_q(p_{\mathrm{del}}),G_q(p_{\mathrm{ins}})\}.
\end{equation}
\end{theorem}
\begin{proof}
Suppose first that $\tau\ge1$. By
Corollary~\ref{cor:singleton-insdel-asymptotic}, the transferred integer
budget satisfies $n-t_n=O(1)$, and $|\cC_n|\le Lq^{n-t_n}$.
Thus $\limsup_n R_n=0$, including the rounding-sensitive case $\tau=1$.

Assume henceforth that $\tau<1$.  Fix
$0\leq p<(q-1)/q$, and choose integers $N_n\geq n$ such that
$n/{N_n}\longrightarrow1-p$.

Set $s_n\eqdef N_n-n$.  By \eqref{eq:insertion-sphere-Iq}, every fixed
$\vc\in\Sigma_q^n$ has exactly $I_q(n,s_n)$ supersequences of length
$N_n$, and the standard type estimate gives
\begin{equation}\label{eq:ambient-entropy}
I_q(n,s_n)
=q^{N_n(H_q(p)+o(1))}.
\end{equation}
Consequently, if every $\vy\in\Sigma_q^{N_n}$ contains at most a constant
$K$ codewords of $\cC_n$ as subsequences, then double counting the pairs
$(\vc,\vy)$ gives
\[
|\cC_n|I_q(n,s_n)\leq Kq^{N_n},
\]
and therefore
\begin{equation}\label{eq:ambient-averaging}
\limsup_{n\to\infty}R_n\leq G_q(p).
\end{equation}
We establish the required constant-multiplicity property in two ways.
In both arguments, write $\alpha_n=n/N_n$ and $c_L=\binom{L+1}{2}$.

We first prove the deletion-based estimate.
If $p_{\mathrm{del}}=0$, the deletion term is the trivial upper bound one. Otherwise, assume that $p<p_{\mathrm{del}}$, or equivalently,
\begin{equation}\label{eq:deletion-strict-gap}
(1-p)^L>1-\tau.
\end{equation}
Because $\alpha_n\to1-p$ and $t_n/n\to\tau$, the strict inequality
\eqref{eq:deletion-strict-gap} allows us to choose a constant
$M\geq L+1$ and then take $n$ sufficiently large so that
\begin{equation}\label{eq:deletion-gap-with-M}
\alpha_n^L-\frac{c_L}{M\alpha_n}>1-\frac{t_n}{n}.
\end{equation}
Suppose that some $\vy\in\Sigma_q^{N_n}$ contains $M$ distinct codewords
$\vc_1,\ldots,\vc_M\in\cC_n$.  For each $i$, fix an embedding of
$\vc_i$ into $\vy$, and let $A_i\subseteq\Int{N_n}$ be its set of
embedding positions.  Thus $|A_i|=n=\alpha_nN_n$.  Applying
Lemma~\ref{lem:higher-order-intersection} with $k=L+1$ gives distinct
indices $i_1,\ldots,i_{L+1}$ such that
\begin{align*}
\left|\bigcap_{j=1}^{L+1}A_{i_j}\right|
\geq
\left(\alpha_n^{L+1}-\frac{c_L}{M}\right)N_n=
\left(\alpha_n^L-\frac{c_L}{M\alpha_n}\right)n>n-t_n.
\end{align*}
The symbols of $\vy$ on this intersection form a common subsequence of
$\vc_{i_1},\ldots,\vc_{i_{L+1}}$.  Its length is greater than $n-t_n$,
contradicting $(0,t_n;L)$-list decodability supplied by the reduction.
Hence every ambient word contains at most
$M-1$ codewords.  By \eqref{eq:ambient-averaging},
$\limsup R_n\leq G_q(p)$ for every $p<p_{\mathrm{del}}$ with
$p<(q-1)/q$.  Letting $p$ approach
$\min\{p_{\mathrm{del}},(q-1)/q\}$ and using the definition of
$\overline H_q$ yields
\[
\limsup_{n\to\infty}R_n\leq G_q(p_{\mathrm{del}}).
\]

We next prove the insertion-based estimate.
If $\gamma=0$, the insertion term is one. Otherwise, assume that $p<p_{\mathrm{ins}}$ and $S_L(p)<\gamma$.  Suppose that an
ambient word $\vy\in\Sigma_q^{N_n}$ contains $M$ distinct codewords, and
fix embedding sets $A_1,\ldots,A_M$ for them.  Set
\[
B_i\eqdef\Int{N_n}\setminus A_i,
\qquad
\beta_n\eqdef\frac{s_n}{N_n}=1-\alpha_n.
\]
Then $|B_i|=\beta_nN_n$ and $\beta_n\to p$.  Since
\begin{equation}\label{eq:insertion-strict-gap}
\frac{1-p^{L+1}}{1-p}
=1+S_L(p)<1+\gamma,
\end{equation}
we may first choose a constant $M\geq L+1$ and then take $n$ sufficiently
large so that
\begin{equation}\label{eq:insertion-gap-with-M}
\frac{1-\beta_n^{L+1}+c_L/M}{\alpha_n}<1+\gamma.
\end{equation}
Apply Lemma~\ref{lem:higher-order-intersection} to
$B_1,\ldots,B_M$.  For some distinct $i_1,\ldots,i_{L+1}$,
\begin{align*}
\left|\bigcup_{j=1}^{L+1}A_{i_j}\right|
=N_n-\left|\bigcap_{j=1}^{L+1}B_{i_j}\right|\leq
\left(1-\beta_n^{L+1}+\frac{c_L}{M}\right)N_n<(1+\gamma)n.
\end{align*}
Restricting $\vy$ to this union produces a common supersequence of the
$L+1$ codewords.  Its integer length is at most
$n+\lfloor\gamma n\rfloor$, contradicting
$(\gamma,\delta,L)$-list decodability.  Thus every ambient word again
contains at most a constant number of codewords, and
\eqref{eq:ambient-averaging} holds.  Letting $p$ approach
$p_{\mathrm{ins}}$ from below gives
\[
\limsup_{n\to\infty}R_n\leq G_q(p_{\mathrm{ins}}).
\]
Combining the deletion-based and insertion-based estimates proves
\eqref{eq:higher-order-elias-rate}.
\end{proof}

\begin{corollary}[Explicit specializations of the higher-order Elias bound]
\label{cor:higher-order-elias-specializations}
Assume $\tau<1$ in Theorem~\ref{thm:higher-order-elias}.
\begin{enumerate}
\item The deletion term is
\begin{equation}\label{eq:explicit-deletion-elias-term}
G_q(p_{\mathrm{del}})
=\frac{1-\overline H_q\!\left(1-(1-\tau)^{1/L}\right)}
       {(1-\tau)^{1/L}}.
\end{equation}
It vanishes whenever $\tau\ge1-q^{-L}$.
\item If
$0\le\gamma<S_L((q-1)/q)$ and $p_L(\gamma)$ denotes the unique root in
$[0,(q-1)/q)$ of
$p+p^2+\cdots+p^L=\gamma$,
then, for pure insertions ($\delta=0$), Theorem~\ref{thm:higher-order-elias} gives
\begin{equation}\label{eq:pure-insertion-higher-order-elias}
\limsup_{n\to\infty}R_n\le G_q(p_L(\gamma)).
\end{equation}
\item The insertion term vanishes when
\begin{equation}\label{eq:insertion-zero-threshold}
\gamma\ge S_L\left(\frac{q-1}{q}\right)
=(q-1)\left[1-\left(\frac{q-1}{q}\right)^L\right].
\end{equation}
\end{enumerate}
\end{corollary}
\begin{proof}
The first identity follows by substituting
$p_{\mathrm{del}}=1-(1-\tau)^{1/L}$ into \eqref{eq:Gq-definition}.
Since $G_q$ vanishes for $p\ge(q-1)/q$, its deletion term is zero when
$p_{\mathrm{del}}\ge(q-1)/q$, equivalently when
$\tau\ge1-q^{-L}$. For the insertion term, $S_L$ is strictly increasing
on $[0,1)$, so below the endpoint threshold the supremum in
\eqref{eq:pdel-definition} is the unique root $p_L(\gamma)$. At and
above the endpoint threshold it equals $(q-1)/q$, where $G_q$ vanishes.
\end{proof}

\begin{remark}[Recovery of Yasunaga at $L=1$]
\label{rem:yasunaga-recovery}
Let $L=1$ and set $\tau=\delta+\gamma$. If $\tau\ge1$, the higher-order
bound gives zero rate. If $\tau<1$, then
$p_{\mathrm{del}}=\tau$
and 
$p_{\mathrm{ins}}=\min\left\{\gamma,(q-1)/{q}\right\}$.

Since $\tau\ge\gamma$ and $G_q$ is nonincreasing,
\[
\min\{G_q(p_{\mathrm{del}}),G_q(p_{\mathrm{ins}})\}=G_q(\tau).
\]
Thus
\begin{equation}\label{eq:L1-yasunaga-recovery}
\limsup_{n\to\infty}R_n\le G_q(\tau).
\end{equation}
In the nontrivial range $0\le\tau<(q-1)/q$, this becomes
\begin{equation}\label{eq:L1-yasunaga-nontrivial-range}
\limsup_{n\to\infty}R_n
\le\frac{1-H_q(\delta+\gamma)}{1-\delta-\gamma},
\end{equation}
which is the asymptotic Yasunaga bound in
Theorem~\ref{thm:yasunaga-elias} with $\Delta=\delta+\gamma$.
\end{remark}

We now compare Theorem~\ref{thm:higher-order-elias} with the HSS
insertion-only benchmark recalled in Theorem~\ref{thm:HSS-insertion}.
\begin{proposition}[Strict fixed-list improvement of the HSS insertion converse]
\label{prop:hoe-improves-hss-insertion}
For $q\ge2$, a finite $L\ge1$, and $0<\gamma<q-1$,
\begin{equation}\label{eq:hoe-strictly-improves-hss}
G_q(p_{\mathrm{ins}})<\Uins(\gamma).
\end{equation}
\end{proposition}
\begin{proof}
Put
\[
a\eqdef\frac{\gamma}{1+\gamma}.
\]
A direct calculation shows that
\begin{equation}\label{eq:HSS-Gq-in-proof}
\mathsf U_{q,\mathrm{ins}}^{\mathrm{HSS}}(\gamma)=G_q(a).
\end{equation}
Since $0<a<(q-1)/q$ and
\[
S_L(a)
=a+a^2+\cdots+a^L
<\frac{a}{1-a}
=\gamma,
\]
the definition of $p_{\mathrm{ins}}$ gives $p_{\mathrm{ins}}>a$.
Equation~\eqref{eq:Gq-derivative} and the fact that $G_q((q-1)/q)=0$
therefore imply
\[
G_q(p_{\mathrm{ins}})<G_q(a),
\]
which proves \eqref{eq:hoe-strictly-improves-hss}.
\end{proof}
For fixed $0<\gamma<q-1$, the solution of $S_L(p)=\gamma$ satisfies,
for all sufficiently large $L$,
\begin{equation}\label{eq:fixed-L-to-HSS-limit}
p_L(\gamma)\downarrow\frac\gamma{1+\gamma},\qquad
G_q(p_L(\gamma))\uparrow\Uins(\gamma).
\end{equation}
This is a limit of fixed-$L$ bounds, not an application of a fixed-$L$
theorem to a sequence $L=L_n$. For genuinely growing lists,
double counting at length $N_n=n+\lfloor\gamma n\rfloor$ gives directly
\begin{equation}\label{eq:growing-list-HSS-recovery}
|\cC_n|I_q(n,\lfloor\gamma n\rfloor)\le L_nq^{N_n},\qquad
R_n\le\Uins(\gamma)+\frac{\log_qL_n}{n}+o(1).
\end{equation}
Thus the first-order HSS bound also holds for every subexponential list
size. The higher-order deletion term, on the other hand, tends to one
as $L\to\infty$ at fixed positive deletion fraction.

\section{Comparison of Upper and Lower Bounds}\label{sec:comparison}

We compare upper and lower bounds for binary codes with fixed list
sizes $L=1,2,3$. The lower curves use feasible evaluations of our
rate bounds, while the upper curves follow from the stated converses.
Negative rate expressions are replaced by zero.

\begin{figure}[!tbp]
\centering
\begin{tikzpicture}
\begin{axis}[
 width=.85\textwidth,height=5.8cm,
 xlabel={Deletion fraction $\delta$},ylabel={Gain over $R_{\mathrm{Lev}}$},
 xmin=0,xmax=.134,ymin=0,ymax=.07,
 xtick={0,.02,.04,.06,.08,.10,.12},
 scaled ticks=false,tick label style={font=\small},
 xticklabel style={/pgf/number format/fixed,/pgf/number format/precision=2},
 yticklabel style={/pgf/number format/fixed,/pgf/number format/precision=3},
 legend style={at={(.02,.98)},anchor=north west,font=\footnotesize,draw=none},
 grid=major,major grid style={gray!20}]
\addplot[HalfTeal,thick,dashed] coordinates {
(0,0)
(0.0025,0.000013513754)
(0.005,0.000054007033)
(0.0075,0.000121403446)
(0.01,0.000215620631)
(0.0125,0.000336570146)
(0.015,0.000484157361)
(0.0175,0.000658281358)
(0.02,0.000858834828)
(0.0225,0.001085703980)
(0.025,0.001338768444)
(0.0275,0.001617901186)
(0.03,0.001922968429)
(0.0325,0.002253829569)
(0.035,0.002610337108)
(0.0375,0.002992336587)
(0.04,0.003399666526)
(0.0425,0.003832158372)
(0.045,0.004289636453)
(0.0475,0.004771917939)
(0.05,0.005278812813)
(0.0525,0.005810123846)
(0.055,0.006365646587)
(0.0575,0.006945169357)
(0.06,0.007548473252)
(0.0625,0.008175332160)
(0.065,0.008825512780)
(0.0675,0.009498774663)
(0.07,0.010194870251)
(0.0725,0.010913544934)
(0.075,0.011654537120)
(0.0775,0.012417578304)
(0.08,0.013202393166)
(0.0825,0.014008699664)
(0.085,0.014836209149)
(0.0875,0.015684626487)
(0.09,0.016553650194)
(0.0925,0.017442972584)
(0.095,0.018352279923)
(0.0975,0.019281252604)
(0.1,0.020229565320)
(0.1025,0.021196887267)
(0.105,0.022182882338)
(0.1075,0.023187209339)
(0.11,0.024209522220)
(0.1125,0.025249470302)
(0.115,0.026306698526)
(0.1175,0.027380847706)
(0.12,0.028471554798)
(0.1225,0.029578453164)
(0.125,0.030701172860)
(0.1275,0.031839340923)
(0.13,0.032992581665)
(0.1325,0.034160516981)
};
\addlegendentry{Run density $1/2$}
\addplot[RunBlue,thick] coordinates {
(0,0)
(0.0025,0.000018079019)
(0.005,0.000072499420)
(0.0075,0.000163540921)
(0.01,0.000291489795)
(0.0125,0.000456639260)
(0.015,0.000659289893)
(0.0175,0.000899750084)
(0.02,0.001178336514)
(0.0225,0.001495374682)
(0.025,0.001851199473)
(0.0275,0.002246155765)
(0.03,0.002680599093)
(0.0325,0.003154896369)
(0.035,0.003669426655)
(0.0375,0.004224582014)
(0.04,0.004820768422)
(0.0425,0.005458406766)
(0.045,0.006137933938)
(0.0475,0.006859804007)
(0.05,0.007624489525)
(0.0525,0.008432482927)
(0.055,0.009284298086)
(0.0575,0.010180472000)
(0.06,0.011121566655)
(0.0625,0.012108171061)
(0.065,0.013140903507)
(0.0675,0.014220414036)
(0.07,0.015347387182)
(0.0725,0.016522545012)
(0.075,0.017746650479)
(0.0775,0.019020511172)
(0.08,0.020344983483)
(0.0825,0.021720977266)
(0.085,0.023149461060)
(0.0875,0.024631467961)
(0.09,0.026168102235)
(0.0925,0.027760546811)
(0.095,0.029410071776)
(0.0975,0.031118044074)
(0.1,0.032885938599)
(0.1025,0.034715350968)
(0.105,0.036608012275)
(0.1075,0.038565806258)
(0.11,0.040590789373)
(0.1125,0.042685214430)
(0.115,0.044851558631)
(0.1175,0.047092557097)
(0.12,0.049411243322)
(0.1225,0.051810998467)
(0.125,0.054295612118)
(0.1275,0.056869358120)
(0.13,0.059537090604)
(0.1325,0.062304367606)
};
\addlegendentry{Variable run density}
\end{axis}
\end{tikzpicture}
\caption{Additive unique-decoding gains over the Levenshtein expression
where that expression is positive, separating density $1/2$ from variable
density.}
\label{fig:unique-gain}
\end{figure}
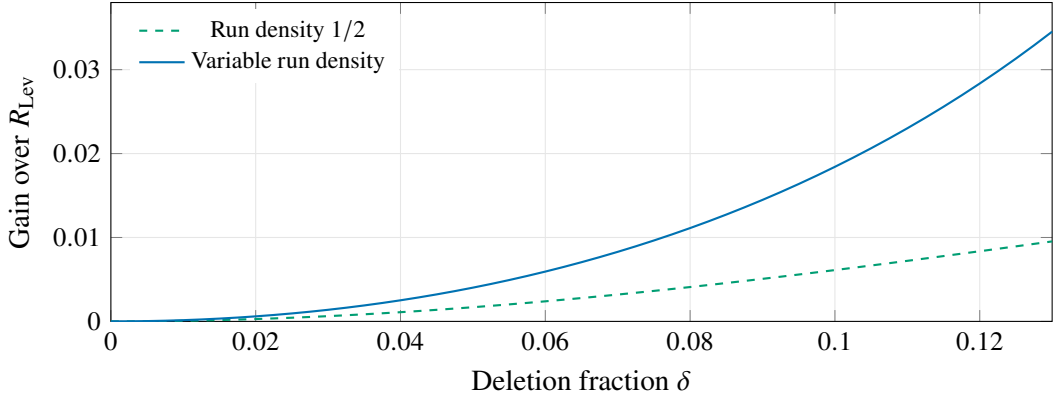

\begin{table}[!tbp]
\centering\small
\begin{tabular}{@{}rrrrrr@{}}
\toprule
$L$ & $\delta$ & HSS union bound & Uniform LLL & Density $1/2$ & Variable density\\
\midrule
2 & 0.05 & 0.396936 & 0.595405 & 0.596367 & 0.596955 \\
2 & 0.10 & 0.231004 & 0.346507 & 0.350480 & 0.353584 \\
2 & 0.20 & 0.011405 & 0.017108 & 0.032191 & 0.055408 \\
3 & 0.05 & 0.476103 & 0.634804 & 0.635403 & 0.635666 \\
3 & 0.10 & 0.306004 & 0.408006 & 0.410362 & 0.411768 \\
3 & 0.20 & 0.078072 & 0.104096 & 0.112478 & 0.122873 \\
\bottomrule
\end{tabular}
\caption{Deletion lower bounds for fixed list sizes}
\label{tab:fixed-list-rates}
\end{table}

\subsection{Binary deletion codes}\label{sec:comparison-deletion}

By Theorem~\ref{thm:fixed-L-strict}, our deletion lower bound
$\pos{\RdelEval(\delta)}$ improves the uniform-LLL bound
for $0<\delta<1/2$.
For $L=1$, Figures~\ref{fig:unique-rate-comparison}
and~\ref{fig:unique-gain} illustrate the improvement over
Levenshtein's bound and the additional gain from optimizing
the run density.
For $L=2,3$, Table~\ref{tab:fixed-list-rates} compares the
HSS union bound~\cite{HaeuplerShahrasbiSudan2018SyncStringsList},
the uniform-LLL bound (Corollary~\ref{cor:uniform-deletion}),
and our locally restricted bounds at fixed and variable run densities
(Theorem~\ref{thm:fixed-L-strict} and
Proposition~\ref{thm:deletion-lower-bound}).

Combining the converses in Section~\ref{sec:converses}, we obtain
\[
\cR_{2,\mathrm{del}}^{(L)}(\delta)
\le \min\left\{
\mathsf U_2^{\mathrm{SL}}(\delta),
\mathsf U_{2,\mathrm{del}}^{\mathrm{HSS}}(\delta),
G_2\!\left(1-(1-\delta)^{1/L}\right)
\right\}.
\]

Figures~\ref{fig:comparison-deletion-L1}--\ref{fig:comparison-deletion-L3}
compare the three deletion converses.
For $L=1$, Yasunaga's bound is strongest except at small
deletion fractions. For $L=2$, the higher-order Elias bound
is strongest over an intermediate range of deletion fractions.
For $L=3$, the Singleton--Levenshtein and HSS bounds
are stronger throughout the displayed range.

Substantial gaps remain. At $\delta=0.1$, the smallest displayed upper
bounds for $L=1,2,3$ are approximately $0.590005$, $0.746335$, and
$0.772908$, compared with lower evaluations $0.194890$, $0.353584$, and
$0.411768$. 

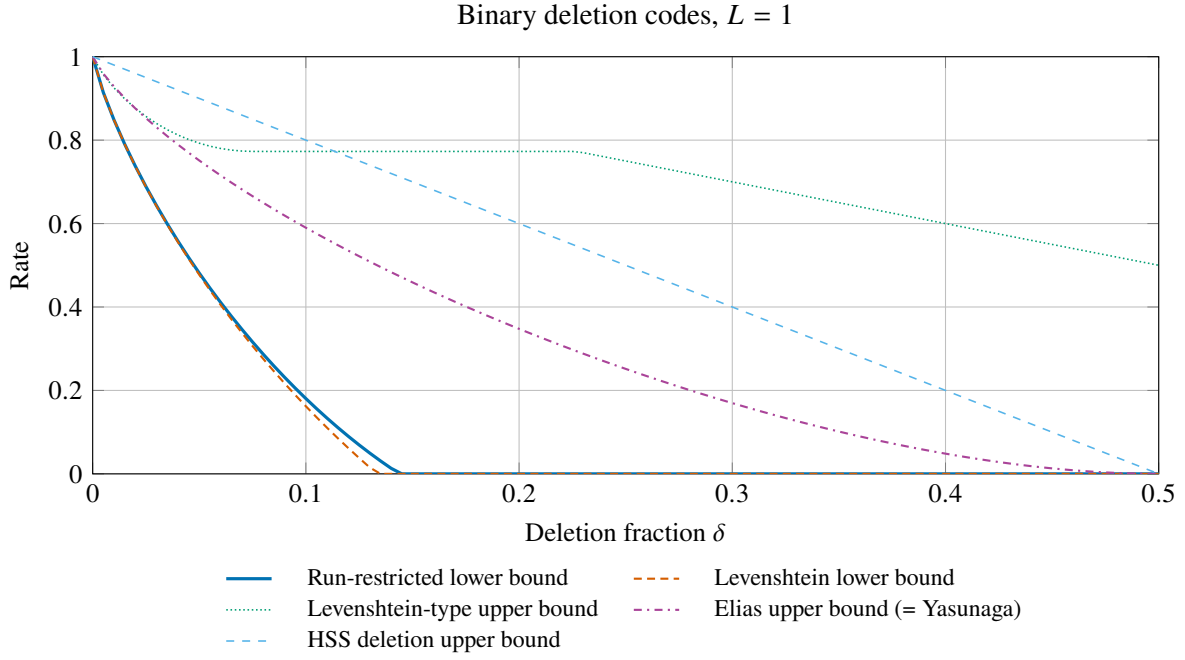
\begin{figure}[tbp]
\centering
\begin{tikzpicture}
\begin{axis}[comparison axis,xmax=.5,xtick={0,.1,.2,.3,.4,.5},xlabel={Deletion fraction $\delta$},title={Binary deletion codes, $L=1$}]
\addplot[RunBlue,very thick] coordinates {
(0,1.000000000000)
(0.005,0.914243114753)
(0.01,0.848705218003)
(0.015,0.790937869695)
(0.02,0.738297251430)
(0.025,0.689529336480)
(0.03,0.643896883430)
(0.035,0.600913973577)
(0.04,0.560236390257)
(0.045,0.521607866066)
(0.05,0.484830575293)
(0.055,0.449747578364)
(0.06,0.416231728346)
(0.065,0.384178327306)
(0.07,0.353500085382)
(0.075,0.324123562226)
(0.08,0.295986603079)
(0.085,0.269036468583)
(0.09,0.243228468107)
(0.095,0.218524975401)
(0.1,0.194894751421)
(0.105,0.172312532812)
(0.11,0.150758873044)
(0.115,0.130220252621)
(0.12,0.110689512747)
(0.125,0.092166725718)
(0.13,0.074660720549)
(0.135,0.058191693124)
(0.14,0.042795802888)
(0.145,0.028533873305)
(0.15,0.015510090433)
(0.155,0.003922569646)
(0.16,0.000000000000)
(0.165,0.000000000000)
(0.17,0.000000000000)
(0.175,0.000000000000)
(0.18,0.000000000000)
(0.185,0.000000000000)
(0.19,0.000000000000)
(0.195,0.000000000000)
(0.2,0.000000000000)
(0.205,0.000000000000)
(0.21,0.000000000000)
(0.215,0.000000000000)
(0.22,0.000000000000)
(0.225,0.000000000000)
(0.23,0.000000000000)
(0.235,0.000000000000)
(0.24,0.000000000000)
(0.245,0.000000000000)
(0.25,0.000000000000)
(0.255,0.000000000000)
(0.26,0.000000000000)
(0.265,0.000000000000)
(0.27,0.000000000000)
(0.275,0.000000000000)
(0.28,0.000000000000)
(0.285,0.000000000000)
(0.29,0.000000000000)
(0.295,0.000000000000)
(0.3,0.000000000000)
(0.305,0.000000000000)
(0.31,0.000000000000)
(0.315,0.000000000000)
(0.32,0.000000000000)
(0.325,0.000000000000)
(0.33,0.000000000000)
(0.335,0.000000000000)
(0.34,0.000000000000)
(0.345,0.000000000000)
(0.35,0.000000000000)
(0.355,0.000000000000)
(0.36,0.000000000000)
(0.365,0.000000000000)
(0.37,0.000000000000)
(0.375,0.000000000000)
(0.38,0.000000000000)
(0.385,0.000000000000)
(0.39,0.000000000000)
(0.395,0.000000000000)
(0.4,0.000000000000)
(0.405,0.000000000000)
(0.41,0.000000000000)
(0.415,0.000000000000)
(0.42,0.000000000000)
(0.425,0.000000000000)
(0.43,0.000000000000)
(0.435,0.000000000000)
(0.44,0.000000000000)
(0.445,0.000000000000)
(0.45,0.000000000000)
(0.455,0.000000000000)
(0.46,0.000000000000)
(0.465,0.000000000000)
(0.47,0.000000000000)
(0.475,0.000000000000)
(0.48,0.000000000000)
(0.485,0.000000000000)
(0.49,0.000000000000)
(0.495,0.000000000000)
(0.5,0.000000000000)
};
\addlegendentry{Locally restricted lower bound}
\addplot[BaselineOrange,thick,densely dashed] coordinates {
(0.00000,1.0000000000)
(0.00500,0.9141706153)
(0.01000,0.8484137282)
(0.01500,0.7902785798)
(0.02000,0.7371189149)
(0.02500,0.6876781370)
(0.03000,0.6412162843)
(0.03500,0.5972445469)
(0.04000,0.5554156218)
(0.04500,0.5154699321)
(0.05000,0.4772060858)
(0.05500,0.4404632803)
(0.06000,0.4051101617)
(0.06500,0.3710374238)
(0.07000,0.3381526982)
(0.07500,0.3063769117)
(0.08000,0.2756416196)
(0.08500,0.2458870075)
(0.09000,0.2170603659)
(0.09500,0.1891149036)
(0.10000,0.1620088128)
(0.10500,0.1357045205)
(0.11000,0.1101680837)
(0.11500,0.0853686940)
(0.12000,0.0612782694)
(0.12500,0.0378711136)
(0.13000,0.0151236299)
(0.13500,0.0000000000)
(0.14000,0.0000000000)
(0.14500,0.0000000000)
(0.15000,0.0000000000)
(0.15500,0.0000000000)
(0.16000,0.0000000000)
(0.16500,0.0000000000)
(0.17000,0.0000000000)
(0.17500,0.0000000000)
(0.18000,0.0000000000)
(0.18500,0.0000000000)
(0.19000,0.0000000000)
(0.19500,0.0000000000)
(0.20000,0.0000000000)
(0.20500,0.0000000000)
(0.21000,0.0000000000)
(0.21500,0.0000000000)
(0.22000,0.0000000000)
(0.22500,0.0000000000)
(0.23000,0.0000000000)
(0.23500,0.0000000000)
(0.24000,0.0000000000)
(0.24500,0.0000000000)
(0.25000,0.0000000000)
(0.25500,0.0000000000)
(0.26000,0.0000000000)
(0.26500,0.0000000000)
(0.27000,0.0000000000)
(0.27500,0.0000000000)
(0.28000,0.0000000000)
(0.28500,0.0000000000)
(0.29000,0.0000000000)
(0.29500,0.0000000000)
(0.30000,0.0000000000)
(0.30500,0.0000000000)
(0.31000,0.0000000000)
(0.31500,0.0000000000)
(0.32000,0.0000000000)
(0.32500,0.0000000000)
(0.33000,0.0000000000)
(0.33500,0.0000000000)
(0.34000,0.0000000000)
(0.34500,0.0000000000)
(0.35000,0.0000000000)
(0.35500,0.0000000000)
(0.36000,0.0000000000)
(0.36500,0.0000000000)
(0.37000,0.0000000000)
(0.37500,0.0000000000)
(0.38000,0.0000000000)
(0.38500,0.0000000000)
(0.39000,0.0000000000)
(0.39500,0.0000000000)
(0.40000,0.0000000000)
(0.40500,0.0000000000)
(0.41000,0.0000000000)
(0.41500,0.0000000000)
(0.42000,0.0000000000)
(0.42500,0.0000000000)
(0.43000,0.0000000000)
(0.43500,0.0000000000)
(0.44000,0.0000000000)
(0.44500,0.0000000000)
(0.45000,0.0000000000)
(0.45500,0.0000000000)
(0.46000,0.0000000000)
(0.46500,0.0000000000)
(0.47000,0.0000000000)
(0.47500,0.0000000000)
(0.48000,0.0000000000)
(0.48500,0.0000000000)
(0.49000,0.0000000000)
(0.49500,0.0000000000)
(0.50000,0.0000000000)
};
\addlegendentry{Levenshtein lower bound}
\addplot[LevGreen,semithick,densely dotted] coordinates {
(0.00000,1.0000000000)
(0.00500,0.9567266063)
(0.01000,0.9253263663)
(0.01500,0.8990672912)
(0.02000,0.8764746603)
(0.02500,0.8568464276)
(0.03000,0.8397717279)
(0.03500,0.8249819033)
(0.04000,0.8122878408)
(0.04500,0.8015484908)
(0.05000,0.7926529821)
(0.05500,0.7855094682)
(0.06000,0.7800376490)
(0.06500,0.7761634826)
(0.07000,0.7738153328)
(0.07500,0.7729211582)
(0.08000,0.7729078048)
(0.08500,0.7729078048)
(0.09000,0.7729078048)
(0.09500,0.7729078048)
(0.10000,0.7729078048)
(0.10500,0.7729078048)
(0.11000,0.7729078048)
(0.11500,0.7729078048)
(0.12000,0.7729078048)
(0.12500,0.7729078048)
(0.13000,0.7729078048)
(0.13500,0.7729078048)
(0.14000,0.7729078048)
(0.14500,0.7729078048)
(0.15000,0.7729078048)
(0.15500,0.7729078048)
(0.16000,0.7729078048)
(0.16500,0.7729078048)
(0.17000,0.7729078048)
(0.17500,0.7729078048)
(0.18000,0.7729078048)
(0.18500,0.7729078048)
(0.19000,0.7729078048)
(0.19500,0.7729078048)
(0.20000,0.7729078048)
(0.20500,0.7729078048)
(0.21000,0.7729078048)
(0.21500,0.7729078048)
(0.22000,0.7729078048)
(0.22500,0.7729078048)
(0.23000,0.7700000000)
(0.23500,0.7650000000)
(0.24000,0.7600000000)
(0.24500,0.7550000000)
(0.25000,0.7500000000)
(0.25500,0.7450000000)
(0.26000,0.7400000000)
(0.26500,0.7350000000)
(0.27000,0.7300000000)
(0.27500,0.7250000000)
(0.28000,0.7200000000)
(0.28500,0.7150000000)
(0.29000,0.7100000000)
(0.29500,0.7050000000)
(0.30000,0.7000000000)
(0.30500,0.6950000000)
(0.31000,0.6900000000)
(0.31500,0.6850000000)
(0.32000,0.6800000000)
(0.32500,0.6750000000)
(0.33000,0.6700000000)
(0.33500,0.6650000000)
(0.34000,0.6600000000)
(0.34500,0.6550000000)
(0.35000,0.6500000000)
(0.35500,0.6450000000)
(0.36000,0.6400000000)
(0.36500,0.6350000000)
(0.37000,0.6300000000)
(0.37500,0.6250000000)
(0.38000,0.6200000000)
(0.38500,0.6150000000)
(0.39000,0.6100000000)
(0.39500,0.6050000000)
(0.40000,0.6000000000)
(0.40500,0.5950000000)
(0.41000,0.5900000000)
(0.41500,0.5850000000)
(0.42000,0.5800000000)
(0.42500,0.5750000000)
(0.43000,0.5700000000)
(0.43500,0.5650000000)
(0.44000,0.5600000000)
(0.44500,0.5550000000)
(0.45000,0.5500000000)
(0.45500,0.5450000000)
(0.46000,0.5400000000)
(0.46500,0.5350000000)
(0.47000,0.5300000000)
(0.47500,0.5250000000)
(0.48000,0.5200000000)
(0.48500,0.5150000000)
(0.49000,0.5100000000)
(0.49500,0.5050000000)
(0.50000,0.5000000000)
};
\addlegendentry{Levenshtein-type upper bound}
\addplot[EliasPurple,thick,dashdotted] coordinates {
(0.00000,1.0000000000)
(0.00500,0.9593822188)
(0.01000,0.9284917819)
(0.01500,0.9011566395)
(0.02000,0.8760810790)
(0.02500,0.8526554549)
(0.03000,0.8305238579)
(0.03500,0.8094531331)
(0.04000,0.7892789697)
(0.04500,0.7698795456)
(0.05000,0.7511610978)
(0.05500,0.7330493546)
(0.06000,0.7154841286)
(0.06500,0.6984157347)
(0.07000,0.6818025259)
(0.07500,0.6656091415)
(0.08000,0.6498052280)
(0.08500,0.6343644850)
(0.09000,0.6192639373)
(0.09500,0.6044833722)
(0.10000,0.5900048960)
(0.10500,0.5758125813)
(0.11000,0.5618921818)
(0.11500,0.5482309006)
(0.12000,0.5348171985)
(0.12500,0.5216406363)
(0.13000,0.5086917413)
(0.13500,0.4959618959)
(0.14000,0.4834432423)
(0.14500,0.4711286020)
(0.15000,0.4590114062)
(0.15500,0.4470856355)
(0.16000,0.4353457683)
(0.16500,0.4237867355)
(0.17000,0.4124038810)
(0.17500,0.4011929272)
(0.18000,0.3901499442)
(0.18500,0.3792713234)
(0.19000,0.3685537532)
(0.19500,0.3579941985)
(0.20000,0.3475898814)
(0.20500,0.3373382649)
(0.21000,0.3272370381)
(0.21500,0.3172841024)
(0.22000,0.3074775603)
(0.22500,0.2978157045)
(0.23000,0.2882970084)
(0.23500,0.2789201177)
(0.24000,0.2696838429)
(0.24500,0.2605871526)
(0.25000,0.2516291674)
(0.25500,0.2428091546)
(0.26000,0.2341265237)
(0.26500,0.2255808220)
(0.27000,0.2171717312)
(0.27500,0.2088990644)
(0.28000,0.2007627631)
(0.28500,0.1927628952)
(0.29000,0.1848996530)
(0.29500,0.1771733521)
(0.30000,0.1695844297)
(0.30500,0.1621334442)
(0.31000,0.1548210748)
(0.31500,0.1476481211)
(0.32000,0.1406155033)
(0.32500,0.1337242629)
(0.33000,0.1269755630)
(0.33500,0.1203706897)
(0.34000,0.1139110531)
(0.34500,0.1075981890)
(0.35000,0.1014337610)
(0.35500,0.0954195619)
(0.36000,0.0895575168)
(0.36500,0.0838496855)
(0.37000,0.0782982656)
(0.37500,0.0729055953)
(0.38000,0.0676741577)
(0.38500,0.0626065840)
(0.39000,0.0577056582)
(0.39500,0.0529743212)
(0.40000,0.0484156759)
(0.40500,0.0440329925)
(0.41000,0.0398297140)
(0.41500,0.0358094619)
(0.42000,0.0319760430)
(0.42500,0.0283334563)
(0.43000,0.0248858997)
(0.43500,0.0216377782)
(0.44000,0.0185937121)
(0.44500,0.0157585456)
(0.45000,0.0131373564)
(0.45500,0.0107354652)
(0.46000,0.0085584466)
(0.46500,0.0066121404)
(0.47000,0.0049026631)
(0.47500,0.0034364209)
(0.48000,0.0022201231)
(0.48500,0.0012607965)
(0.49000,0.0005658005)
(0.49500,0.0001428435)
(0.50000,0.0000000000)
};
\addlegendentry{Elias upper bound (= Yasunaga)}
\addplot[HSSCyan,semithick,dashed] coordinates {
(0.00000,1.0000000000)
(0.00500,0.9900000000)
(0.01000,0.9800000000)
(0.01500,0.9700000000)
(0.02000,0.9600000000)
(0.02500,0.9500000000)
(0.03000,0.9400000000)
(0.03500,0.9300000000)
(0.04000,0.9200000000)
(0.04500,0.9100000000)
(0.05000,0.9000000000)
(0.05500,0.8900000000)
(0.06000,0.8800000000)
(0.06500,0.8700000000)
(0.07000,0.8600000000)
(0.07500,0.8500000000)
(0.08000,0.8400000000)
(0.08500,0.8300000000)
(0.09000,0.8200000000)
(0.09500,0.8100000000)
(0.10000,0.8000000000)
(0.10500,0.7900000000)
(0.11000,0.7800000000)
(0.11500,0.7700000000)
(0.12000,0.7600000000)
(0.12500,0.7500000000)
(0.13000,0.7400000000)
(0.13500,0.7300000000)
(0.14000,0.7200000000)
(0.14500,0.7100000000)
(0.15000,0.7000000000)
(0.15500,0.6900000000)
(0.16000,0.6800000000)
(0.16500,0.6700000000)
(0.17000,0.6600000000)
(0.17500,0.6500000000)
(0.18000,0.6400000000)
(0.18500,0.6300000000)
(0.19000,0.6200000000)
(0.19500,0.6100000000)
(0.20000,0.6000000000)
(0.20500,0.5900000000)
(0.21000,0.5800000000)
(0.21500,0.5700000000)
(0.22000,0.5600000000)
(0.22500,0.5500000000)
(0.23000,0.5400000000)
(0.23500,0.5300000000)
(0.24000,0.5200000000)
(0.24500,0.5100000000)
(0.25000,0.5000000000)
(0.25500,0.4900000000)
(0.26000,0.4800000000)
(0.26500,0.4700000000)
(0.27000,0.4600000000)
(0.27500,0.4500000000)
(0.28000,0.4400000000)
(0.28500,0.4300000000)
(0.29000,0.4200000000)
(0.29500,0.4100000000)
(0.30000,0.4000000000)
(0.30500,0.3900000000)
(0.31000,0.3800000000)
(0.31500,0.3700000000)
(0.32000,0.3600000000)
(0.32500,0.3500000000)
(0.33000,0.3400000000)
(0.33500,0.3300000000)
(0.34000,0.3200000000)
(0.34500,0.3100000000)
(0.35000,0.3000000000)
(0.35500,0.2900000000)
(0.36000,0.2800000000)
(0.36500,0.2700000000)
(0.37000,0.2600000000)
(0.37500,0.2500000000)
(0.38000,0.2400000000)
(0.38500,0.2300000000)
(0.39000,0.2200000000)
(0.39500,0.2100000000)
(0.40000,0.2000000000)
(0.40500,0.1900000000)
(0.41000,0.1800000000)
(0.41500,0.1700000000)
(0.42000,0.1600000000)
(0.42500,0.1500000000)
(0.43000,0.1400000000)
(0.43500,0.1300000000)
(0.44000,0.1200000000)
(0.44500,0.1100000000)
(0.45000,0.1000000000)
(0.45500,0.0900000000)
(0.46000,0.0800000000)
(0.46500,0.0700000000)
(0.47000,0.0600000000)
(0.47500,0.0500000000)
(0.48000,0.0400000000)
(0.48500,0.0300000000)
(0.49000,0.0200000000)
(0.49500,0.0100000000)
(0.50000,0.0000000000)
};
\addlegendentry{HSS deletion upper bound}
\end{axis}
\end{tikzpicture}
\caption{Binary deletion bounds for $L=1$. The higher-order Elias upper
curve is Yasunaga's bound; the locally restricted lower curve uses
\eqref{eq:del-point-certificate}.}
\label{fig:comparison-deletion-L1}
\end{figure}

\begin{figure}[tbp]
\centering
\begin{tikzpicture}
\begin{axis}[comparison axis,xmax=.5,xtick={0,.1,.2,.3,.4,.5},xlabel={Deletion fraction $\delta$},title={Binary deletion codes, $L=2$}]
\addplot[RunBlue,very thick] coordinates {
(0,1.000000000000)
(0.005,0.934391678880)
(0.01,0.883865938417)
(0.015,0.839085879002)
(0.02,0.798067994059)
(0.025,0.759871038198)
(0.03,0.723941261373)
(0.035,0.689913298762)
(0.04,0.657527939621)
(0.045,0.626591794049)
(0.05,0.596955140007)
(0.055,0.568498729904)
(0.06,0.541125438236)
(0.065,0.514754714819)
(0.07,0.489318756567)
(0.075,0.464759780314)
(0.08,0.441028027926)
(0.085,0.418080274052)
(0.09,0.395878688386)
(0.095,0.374389953967)
(0.1,0.353584574331)
(0.105,0.333436322638)
(0.11,0.313921799429)
(0.115,0.295020074923)
(0.12,0.276712398177)
(0.125,0.258981960066)
(0.13,0.241813700371)
(0.135,0.225194151804)
(0.14,0.209111315707)
(0.145,0.193554565749)
(0.15,0.178514577275)
(0.155,0.163983281173)
(0.16,0.149953842349)
(0.165,0.136420664281)
(0.17,0.123379422723)
(0.175,0.110827133886)
(0.18,0.098762265435)
(0.185,0.087184903328)
(0.19,0.076096994732)
(0.195,0.065502699203)
(0.2,0.055408900958)
(0.205,0.045825972797)
(0.21,0.036768956022)
(0.215,0.028259477253)
(0.22,0.020329093254)
(0.225,0.013025768610)
(0.23,0.006428653521)
(0.235,0.000693850309)
(0.24,0.000000000000)
(0.245,0.000000000000)
(0.25,0.000000000000)
(0.255,0.000000000000)
(0.26,0.000000000000)
(0.265,0.000000000000)
(0.27,0.000000000000)
(0.275,0.000000000000)
(0.28,0.000000000000)
(0.285,0.000000000000)
(0.29,0.000000000000)
(0.295,0.000000000000)
(0.3,0.000000000000)
(0.305,0.000000000000)
(0.31,0.000000000000)
(0.315,0.000000000000)
(0.32,0.000000000000)
(0.325,0.000000000000)
(0.33,0.000000000000)
(0.335,0.000000000000)
(0.34,0.000000000000)
(0.345,0.000000000000)
(0.35,0.000000000000)
(0.355,0.000000000000)
(0.36,0.000000000000)
(0.365,0.000000000000)
(0.37,0.000000000000)
(0.375,0.000000000000)
(0.38,0.000000000000)
(0.385,0.000000000000)
(0.39,0.000000000000)
(0.395,0.000000000000)
(0.4,0.000000000000)
(0.405,0.000000000000)
(0.41,0.000000000000)
(0.415,0.000000000000)
(0.42,0.000000000000)
(0.425,0.000000000000)
(0.43,0.000000000000)
(0.435,0.000000000000)
(0.44,0.000000000000)
(0.445,0.000000000000)
(0.45,0.000000000000)
(0.455,0.000000000000)
(0.46,0.000000000000)
(0.465,0.000000000000)
(0.47,0.000000000000)
(0.475,0.000000000000)
(0.48,0.000000000000)
(0.485,0.000000000000)
(0.49,0.000000000000)
(0.495,0.000000000000)
(0.5,0.000000000000)
};
\addlegendentry{Locally restricted lower bound}
\addplot[BaselineOrange,thick,densely dashed] coordinates {
(0.00000,1.0000000000)
(0.00500,0.9343779615)
(0.01000,0.8838102962)
(0.01500,0.8389589349)
(0.02000,0.7978391862)
(0.02500,0.7595086028)
(0.03000,0.7234122133)
(0.03500,0.6891834102)
(0.04000,0.6565617164)
(0.04500,0.6253524491)
(0.05000,0.5954045643)
(0.05500,0.5665974602)
(0.06000,0.5388326213)
(0.06500,0.5120280678)
(0.07000,0.4861145236)
(0.07500,0.4610326838)
(0.08000,0.4367312147)
(0.08500,0.4131652556)
(0.09000,0.3902952744)
(0.09500,0.3680861777)
(0.10000,0.3465066096)
(0.10500,0.3255283904)
(0.11000,0.3051260628)
(0.11500,0.2852765205)
(0.12000,0.2659587021)
(0.12500,0.2471533352)
(0.13000,0.2288427225)
(0.13500,0.2110105599)
(0.14000,0.1936417825)
(0.14500,0.1767224321)
(0.15000,0.1602395429)
(0.15500,0.1441810430)
(0.16000,0.1285356680)
(0.16500,0.1132928862)
(0.17000,0.0984428319)
(0.17500,0.0839762474)
(0.18000,0.0698844314)
(0.18500,0.0561591928)
(0.19000,0.0427928102)
(0.19500,0.0297779947)
(0.20000,0.0171078577)
(0.20500,0.0047758809)
(0.21000,0.0000000000)
(0.21500,0.0000000000)
(0.22000,0.0000000000)
(0.22500,0.0000000000)
(0.23000,0.0000000000)
(0.23500,0.0000000000)
(0.24000,0.0000000000)
(0.24500,0.0000000000)
(0.25000,0.0000000000)
(0.25500,0.0000000000)
(0.26000,0.0000000000)
(0.26500,0.0000000000)
(0.27000,0.0000000000)
(0.27500,0.0000000000)
(0.28000,0.0000000000)
(0.28500,0.0000000000)
(0.29000,0.0000000000)
(0.29500,0.0000000000)
(0.30000,0.0000000000)
(0.30500,0.0000000000)
(0.31000,0.0000000000)
(0.31500,0.0000000000)
(0.32000,0.0000000000)
(0.32500,0.0000000000)
(0.33000,0.0000000000)
(0.33500,0.0000000000)
(0.34000,0.0000000000)
(0.34500,0.0000000000)
(0.35000,0.0000000000)
(0.35500,0.0000000000)
(0.36000,0.0000000000)
(0.36500,0.0000000000)
(0.37000,0.0000000000)
(0.37500,0.0000000000)
(0.38000,0.0000000000)
(0.38500,0.0000000000)
(0.39000,0.0000000000)
(0.39500,0.0000000000)
(0.40000,0.0000000000)
(0.40500,0.0000000000)
(0.41000,0.0000000000)
(0.41500,0.0000000000)
(0.42000,0.0000000000)
(0.42500,0.0000000000)
(0.43000,0.0000000000)
(0.43500,0.0000000000)
(0.44000,0.0000000000)
(0.44500,0.0000000000)
(0.45000,0.0000000000)
(0.45500,0.0000000000)
(0.46000,0.0000000000)
(0.46500,0.0000000000)
(0.47000,0.0000000000)
(0.47500,0.0000000000)
(0.48000,0.0000000000)
(0.48500,0.0000000000)
(0.49000,0.0000000000)
(0.49500,0.0000000000)
(0.50000,0.0000000000)
};
\addlegendentry{Uniform-LLL lower bound}
\addplot[UnionGray,thick,dotted] coordinates {
(0.00000,0.6666666667)
(0.00500,0.6229186410)
(0.01000,0.5892068641)
(0.01500,0.5593059566)
(0.02000,0.5318927908)
(0.02500,0.5063390685)
(0.03000,0.4822748088)
(0.03500,0.4594556068)
(0.04000,0.4377078109)
(0.04500,0.4169016327)
(0.05000,0.3969363762)
(0.05500,0.3777316401)
(0.06000,0.3592217475)
(0.06500,0.3413520452)
(0.07000,0.3240763491)
(0.07500,0.3073551225)
(0.08000,0.2911541431)
(0.08500,0.2754435038)
(0.09000,0.2601968496)
(0.09500,0.2453907851)
(0.10000,0.2310044064)
(0.10500,0.2170189269)
(0.11000,0.2034173752)
(0.11500,0.1901843470)
(0.12000,0.1773058014)
(0.12500,0.1647688901)
(0.13000,0.1525618150)
(0.13500,0.1406737066)
(0.14000,0.1290945217)
(0.14500,0.1178149547)
(0.15000,0.1068263620)
(0.15500,0.0961206953)
(0.16000,0.0856904454)
(0.16500,0.0755285908)
(0.17000,0.0656285546)
(0.17500,0.0559841649)
(0.18000,0.0465896209)
(0.18500,0.0374394619)
(0.19000,0.0285285401)
(0.19500,0.0198519965)
(0.20000,0.0114052384)
(0.20500,0.0031839206)
(0.21000,0.0000000000)
(0.21500,0.0000000000)
(0.22000,0.0000000000)
(0.22500,0.0000000000)
(0.23000,0.0000000000)
(0.23500,0.0000000000)
(0.24000,0.0000000000)
(0.24500,0.0000000000)
(0.25000,0.0000000000)
(0.25500,0.0000000000)
(0.26000,0.0000000000)
(0.26500,0.0000000000)
(0.27000,0.0000000000)
(0.27500,0.0000000000)
(0.28000,0.0000000000)
(0.28500,0.0000000000)
(0.29000,0.0000000000)
(0.29500,0.0000000000)
(0.30000,0.0000000000)
(0.30500,0.0000000000)
(0.31000,0.0000000000)
(0.31500,0.0000000000)
(0.32000,0.0000000000)
(0.32500,0.0000000000)
(0.33000,0.0000000000)
(0.33500,0.0000000000)
(0.34000,0.0000000000)
(0.34500,0.0000000000)
(0.35000,0.0000000000)
(0.35500,0.0000000000)
(0.36000,0.0000000000)
(0.36500,0.0000000000)
(0.37000,0.0000000000)
(0.37500,0.0000000000)
(0.38000,0.0000000000)
(0.38500,0.0000000000)
(0.39000,0.0000000000)
(0.39500,0.0000000000)
(0.40000,0.0000000000)
(0.40500,0.0000000000)
(0.41000,0.0000000000)
(0.41500,0.0000000000)
(0.42000,0.0000000000)
(0.42500,0.0000000000)
(0.43000,0.0000000000)
(0.43500,0.0000000000)
(0.44000,0.0000000000)
(0.44500,0.0000000000)
(0.45000,0.0000000000)
(0.45500,0.0000000000)
(0.46000,0.0000000000)
(0.46500,0.0000000000)
(0.47000,0.0000000000)
(0.47500,0.0000000000)
(0.48000,0.0000000000)
(0.48500,0.0000000000)
(0.49000,0.0000000000)
(0.49500,0.0000000000)
(0.50000,0.0000000000)
};
\addlegendentry{HSS union-bound lower bound}
\addplot[LevGreen,semithick,densely dotted] coordinates {
(0.00000,1.0000000000)
(0.00500,0.9567266063)
(0.01000,0.9253263663)
(0.01500,0.8990672912)
(0.02000,0.8764746603)
(0.02500,0.8568464276)
(0.03000,0.8397717279)
(0.03500,0.8249819033)
(0.04000,0.8122878408)
(0.04500,0.8015484908)
(0.05000,0.7926529821)
(0.05500,0.7855094682)
(0.06000,0.7800376490)
(0.06500,0.7761634826)
(0.07000,0.7738153328)
(0.07500,0.7729211582)
(0.08000,0.7729078048)
(0.08500,0.7729078048)
(0.09000,0.7729078048)
(0.09500,0.7729078048)
(0.10000,0.7729078048)
(0.10500,0.7729078048)
(0.11000,0.7729078048)
(0.11500,0.7729078048)
(0.12000,0.7729078048)
(0.12500,0.7729078048)
(0.13000,0.7729078048)
(0.13500,0.7729078048)
(0.14000,0.7729078048)
(0.14500,0.7729078048)
(0.15000,0.7729078048)
(0.15500,0.7729078048)
(0.16000,0.7729078048)
(0.16500,0.7729078048)
(0.17000,0.7729078048)
(0.17500,0.7729078048)
(0.18000,0.7729078048)
(0.18500,0.7729078048)
(0.19000,0.7729078048)
(0.19500,0.7729078048)
(0.20000,0.7729078048)
(0.20500,0.7729078048)
(0.21000,0.7729078048)
(0.21500,0.7729078048)
(0.22000,0.7729078048)
(0.22500,0.7729078048)
(0.23000,0.7700000000)
(0.23500,0.7650000000)
(0.24000,0.7600000000)
(0.24500,0.7550000000)
(0.25000,0.7500000000)
(0.25500,0.7450000000)
(0.26000,0.7400000000)
(0.26500,0.7350000000)
(0.27000,0.7300000000)
(0.27500,0.7250000000)
(0.28000,0.7200000000)
(0.28500,0.7150000000)
(0.29000,0.7100000000)
(0.29500,0.7050000000)
(0.30000,0.7000000000)
(0.30500,0.6950000000)
(0.31000,0.6900000000)
(0.31500,0.6850000000)
(0.32000,0.6800000000)
(0.32500,0.6750000000)
(0.33000,0.6700000000)
(0.33500,0.6650000000)
(0.34000,0.6600000000)
(0.34500,0.6550000000)
(0.35000,0.6500000000)
(0.35500,0.6450000000)
(0.36000,0.6400000000)
(0.36500,0.6350000000)
(0.37000,0.6300000000)
(0.37500,0.6250000000)
(0.38000,0.6200000000)
(0.38500,0.6150000000)
(0.39000,0.6100000000)
(0.39500,0.6050000000)
(0.40000,0.6000000000)
(0.40500,0.5950000000)
(0.41000,0.5900000000)
(0.41500,0.5850000000)
(0.42000,0.5800000000)
(0.42500,0.5750000000)
(0.43000,0.5700000000)
(0.43500,0.5650000000)
(0.44000,0.5600000000)
(0.44500,0.5550000000)
(0.45000,0.5500000000)
(0.45500,0.5450000000)
(0.46000,0.5400000000)
(0.46500,0.5350000000)
(0.47000,0.5300000000)
(0.47500,0.5250000000)
(0.48000,0.5200000000)
(0.48500,0.5150000000)
(0.49000,0.5100000000)
(0.49500,0.5050000000)
(0.50000,0.5000000000)
};
\addlegendentry{Levenshtein-type upper bound}
\addplot[EliasPurple,thick,dashdotted] coordinates {
(0.00000,1.0000000000)
(0.00500,0.9772071478)
(0.01000,0.9592979337)
(0.01500,0.9431796886)
(0.02000,0.9282011163)
(0.02500,0.9140550274)
(0.03000,0.9005616568)
(0.03500,0.8876029564)
(0.04000,0.8750955271)
(0.04500,0.8629773765)
(0.05000,0.8512006615)
(0.05500,0.8397273770)
(0.06000,0.8285266322)
(0.06500,0.8175728453)
(0.07000,0.8068445002)
(0.07500,0.7963232621)
(0.08000,0.7859933318)
(0.08500,0.7758409625)
(0.09000,0.7658540919)
(0.09500,0.7560220573)
(0.10000,0.7463353706)
(0.10500,0.7367855391)
(0.11000,0.7273649211)
(0.11500,0.7180666078)
(0.12000,0.7088843262)
(0.12500,0.6998123578)
(0.13000,0.6908454715)
(0.13500,0.6819788661)
(0.14000,0.6732081218)
(0.14500,0.6645291586)
(0.15000,0.6559382008)
(0.15500,0.6474317462)
(0.16000,0.6390065388)
(0.16500,0.6306595462)
(0.17000,0.6223879383)
(0.17500,0.6141890698)
(0.18000,0.6060604642)
(0.18500,0.5979997996)
(0.19000,0.5900048960)
(0.19500,0.5820737046)
(0.20000,0.5742042973)
(0.20500,0.5663948580)
(0.21000,0.5586436746)
(0.21500,0.5509491314)
(0.22000,0.5433097029)
(0.22500,0.5357239478)
(0.23000,0.5281905036)
(0.23500,0.5207080816)
(0.24000,0.5132754626)
(0.24500,0.5058914930)
(0.25000,0.4985550809)
(0.25500,0.4912651926)
(0.26000,0.4840208502)
(0.26500,0.4768211279)
(0.27000,0.4696651502)
(0.27500,0.4625520892)
(0.28000,0.4554811625)
(0.28500,0.4484516311)
(0.29000,0.4414627982)
(0.29500,0.4345140067)
(0.30000,0.4276046383)
(0.30500,0.4207341122)
(0.31000,0.4139018832)
(0.31500,0.4071074415)
(0.32000,0.4003503109)
(0.32500,0.3936300482)
(0.33000,0.3869462427)
(0.33500,0.3802985147)
(0.34000,0.3736865155)
(0.34500,0.3671099268)
(0.35000,0.3605684596)
(0.35500,0.3540618546)
(0.36000,0.3475898814)
(0.36500,0.3411523381)
(0.37000,0.3347490514)
(0.37500,0.3283798764)
(0.38000,0.3220446962)
(0.38500,0.3157434225)
(0.39000,0.3094759949)
(0.39500,0.3032423814)
(0.40000,0.2970425788)
(0.40500,0.2908766123)
(0.41000,0.2847445363)
(0.41500,0.2786464346)
(0.42000,0.2725824206)
(0.42500,0.2665526381)
(0.43000,0.2605572616)
(0.43500,0.2545964971)
(0.44000,0.2486705824)
(0.44500,0.2427797885)
(0.45000,0.2369244196)
(0.45500,0.2311048147)
(0.46000,0.2253213481)
(0.46500,0.2195744309)
(0.47000,0.2138645116)
(0.47500,0.2081920778)
(0.48000,0.2025576574)
(0.48500,0.1969618198)
(0.49000,0.1914051778)
(0.49500,0.1858883890)
(0.50000,0.1804121577)
};
\addlegendentry{Higher-order Elias upper bound}
\addplot[HSSCyan,semithick,dashed] coordinates {
(0.00000,1.0000000000)
(0.00500,0.9900000000)
(0.01000,0.9800000000)
(0.01500,0.9700000000)
(0.02000,0.9600000000)
(0.02500,0.9500000000)
(0.03000,0.9400000000)
(0.03500,0.9300000000)
(0.04000,0.9200000000)
(0.04500,0.9100000000)
(0.05000,0.9000000000)
(0.05500,0.8900000000)
(0.06000,0.8800000000)
(0.06500,0.8700000000)
(0.07000,0.8600000000)
(0.07500,0.8500000000)
(0.08000,0.8400000000)
(0.08500,0.8300000000)
(0.09000,0.8200000000)
(0.09500,0.8100000000)
(0.10000,0.8000000000)
(0.10500,0.7900000000)
(0.11000,0.7800000000)
(0.11500,0.7700000000)
(0.12000,0.7600000000)
(0.12500,0.7500000000)
(0.13000,0.7400000000)
(0.13500,0.7300000000)
(0.14000,0.7200000000)
(0.14500,0.7100000000)
(0.15000,0.7000000000)
(0.15500,0.6900000000)
(0.16000,0.6800000000)
(0.16500,0.6700000000)
(0.17000,0.6600000000)
(0.17500,0.6500000000)
(0.18000,0.6400000000)
(0.18500,0.6300000000)
(0.19000,0.6200000000)
(0.19500,0.6100000000)
(0.20000,0.6000000000)
(0.20500,0.5900000000)
(0.21000,0.5800000000)
(0.21500,0.5700000000)
(0.22000,0.5600000000)
(0.22500,0.5500000000)
(0.23000,0.5400000000)
(0.23500,0.5300000000)
(0.24000,0.5200000000)
(0.24500,0.5100000000)
(0.25000,0.5000000000)
(0.25500,0.4900000000)
(0.26000,0.4800000000)
(0.26500,0.4700000000)
(0.27000,0.4600000000)
(0.27500,0.4500000000)
(0.28000,0.4400000000)
(0.28500,0.4300000000)
(0.29000,0.4200000000)
(0.29500,0.4100000000)
(0.30000,0.4000000000)
(0.30500,0.3900000000)
(0.31000,0.3800000000)
(0.31500,0.3700000000)
(0.32000,0.3600000000)
(0.32500,0.3500000000)
(0.33000,0.3400000000)
(0.33500,0.3300000000)
(0.34000,0.3200000000)
(0.34500,0.3100000000)
(0.35000,0.3000000000)
(0.35500,0.2900000000)
(0.36000,0.2800000000)
(0.36500,0.2700000000)
(0.37000,0.2600000000)
(0.37500,0.2500000000)
(0.38000,0.2400000000)
(0.38500,0.2300000000)
(0.39000,0.2200000000)
(0.39500,0.2100000000)
(0.40000,0.2000000000)
(0.40500,0.1900000000)
(0.41000,0.1800000000)
(0.41500,0.1700000000)
(0.42000,0.1600000000)
(0.42500,0.1500000000)
(0.43000,0.1400000000)
(0.43500,0.1300000000)
(0.44000,0.1200000000)
(0.44500,0.1100000000)
(0.45000,0.1000000000)
(0.45500,0.0900000000)
(0.46000,0.0800000000)
(0.46500,0.0700000000)
(0.47000,0.0600000000)
(0.47500,0.0500000000)
(0.48000,0.0400000000)
(0.48500,0.0300000000)
(0.49000,0.0200000000)
(0.49500,0.0100000000)
(0.50000,0.0000000000)
};
\addlegendentry{HSS deletion upper bound}
\end{axis}
\end{tikzpicture}
\caption{Binary deletion bounds for $L=2$, showing both the improvement
in achievable rate and the complementary roles of the three converses.}
\label{fig:comparison-deletion-L2}
\end{figure}
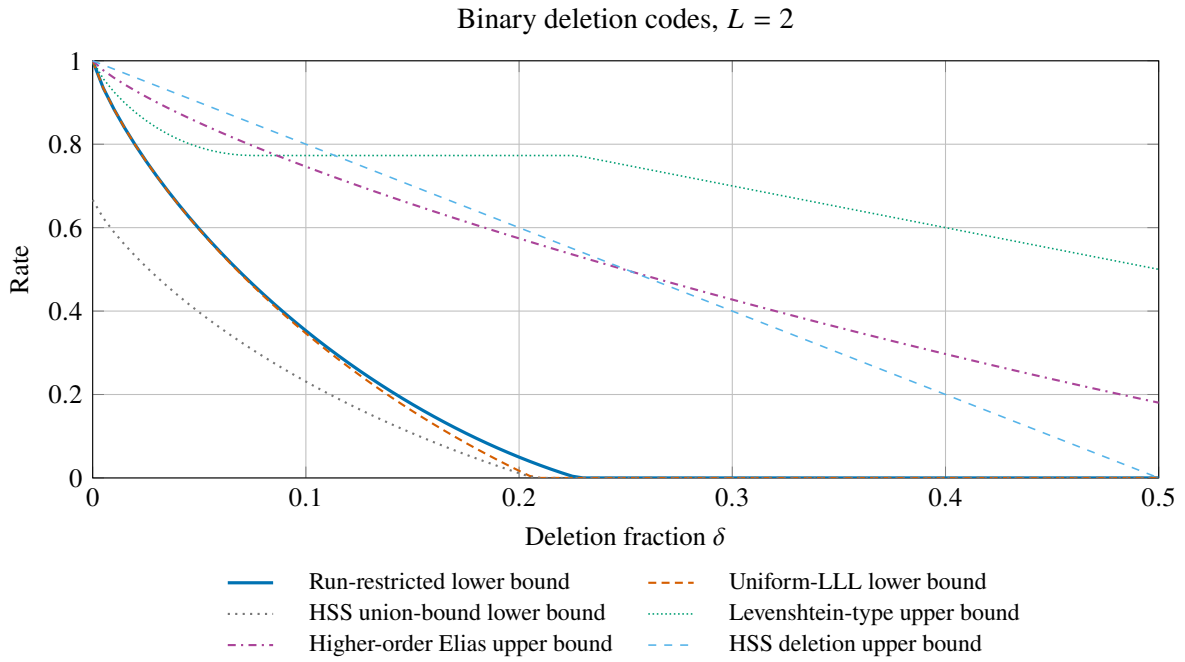

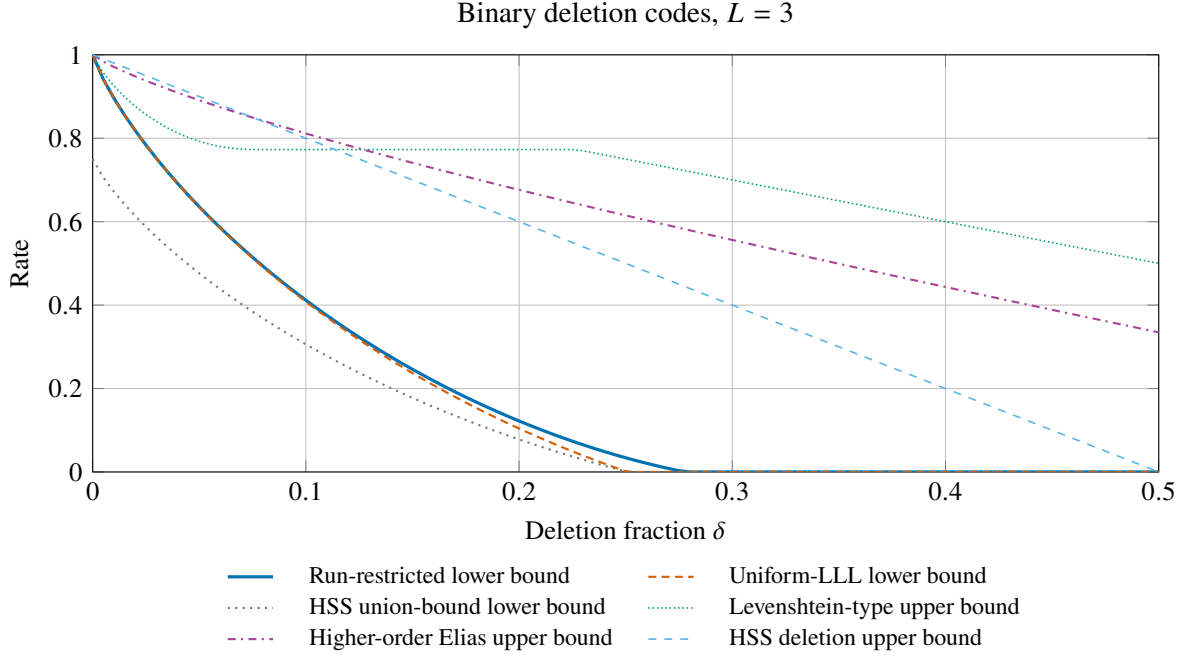
\begin{figure}[tbp]
\centering
\begin{tikzpicture}
\begin{axis}[comparison axis,xmax=.5,xtick={0,.1,.2,.3,.4,.5},xlabel={Deletion fraction $\delta$},title={Binary deletion codes, $L=3$}]
\addplot[RunBlue,very thick] coordinates {
(0,1.000000000000)
(0.005,0.941121812697)
(0.01,0.895641651333)
(0.015,0.855259357478)
(0.02,0.818211126774)
(0.025,0.783659612675)
(0.03,0.751111910985)
(0.035,0.720242571641)
(0.04,0.690820502777)
(0.045,0.662673117263)
(0.05,0.635666639314)
(0.055,0.609694374979)
(0.06,0.584669285569)
(0.065,0.560519054715)
(0.07,0.537182682724)
(0.075,0.514608059297)
(0.08,0.492750186785)
(0.085,0.471569849759)
(0.09,0.451032599145)
(0.095,0.431107963305)
(0.1,0.411768826198)
(0.105,0.392990930820)
(0.11,0.374752478106)
(0.115,0.357033799638)
(0.12,0.339817088190)
(0.125,0.323086174165)
(0.13,0.306826338869)
(0.135,0.291024157682)
(0.14,0.275667367773)
(0.145,0.260744756191)
(0.15,0.246246065045)
(0.155,0.232161911238)
(0.16,0.218483718741)
(0.165,0.205203661863)
(0.17,0.192314618340)
(0.175,0.179810131410)
(0.18,0.167684380302)
(0.185,0.155932158887)
(0.19,0.144548862495)
(0.195,0.133530483226)
(0.2,0.122873614463)
(0.205,0.112575465731)
(0.21,0.102633889656)
(0.215,0.093047423592)
(0.22,0.083815349633)
(0.225,0.074937778429)
(0.23,0.066415764777)
(0.235,0.058251466995)
(0.24,0.050448368599)
(0.245,0.043011591933)
(0.25,0.035948353147)
(0.255,0.029268645386)
(0.26,0.022986313199)
(0.265,0.017120851480)
(0.27,0.011700693684)
(0.275,0.006770060789)
(0.28,0.002406718261)
(0.285,0.000000000000)
(0.29,0.000000000000)
(0.295,0.000000000000)
(0.3,0.000000000000)
(0.305,0.000000000000)
(0.31,0.000000000000)
(0.315,0.000000000000)
(0.32,0.000000000000)
(0.325,0.000000000000)
(0.33,0.000000000000)
(0.335,0.000000000000)
(0.34,0.000000000000)
(0.345,0.000000000000)
(0.35,0.000000000000)
(0.355,0.000000000000)
(0.36,0.000000000000)
(0.365,0.000000000000)
(0.37,0.000000000000)
(0.375,0.000000000000)
(0.38,0.000000000000)
(0.385,0.000000000000)
(0.39,0.000000000000)
(0.395,0.000000000000)
(0.4,0.000000000000)
(0.405,0.000000000000)
(0.41,0.000000000000)
(0.415,0.000000000000)
(0.42,0.000000000000)
(0.425,0.000000000000)
(0.43,0.000000000000)
(0.435,0.000000000000)
(0.44,0.000000000000)
(0.445,0.000000000000)
(0.45,0.000000000000)
(0.455,0.000000000000)
(0.46,0.000000000000)
(0.465,0.000000000000)
(0.47,0.000000000000)
(0.475,0.000000000000)
(0.48,0.000000000000)
(0.485,0.000000000000)
(0.49,0.000000000000)
(0.495,0.000000000000)
(0.5,0.000000000000)
};
\addlegendentry{Locally restricted lower bound}
\addplot[BaselineOrange,thick,densely dashed] coordinates {
(0.00000,1.0000000000)
(0.00500,0.9411137436)
(0.01000,0.8956091521)
(0.01500,0.8551857199)
(0.02000,0.8180792766)
(0.02500,0.7834520913)
(0.03000,0.7508108562)
(0.03500,0.7198296979)
(0.04000,0.6902770812)
(0.04500,0.6619799548)
(0.05000,0.6348040572)
(0.05500,0.6086421869)
(0.06000,0.5834067745)
(0.06500,0.5590249492)
(0.07000,0.5354351321)
(0.07500,0.5125846078)
(0.08000,0.4904277464)
(0.08500,0.4689246717)
(0.09000,0.4480402439)
(0.09500,0.4277432691)
(0.10000,0.4080058752)
(0.10500,0.3888030137)
(0.11000,0.3701120558)
(0.11500,0.3519124627)
(0.12000,0.3341855130)
(0.12500,0.3169140757)
(0.13000,0.3000824200)
(0.13500,0.2836760532)
(0.14000,0.2676815845)
(0.14500,0.2520866063)
(0.15000,0.2368795937)
(0.15500,0.2220498160)
(0.16000,0.2075872605)
(0.16500,0.1934825655)
(0.17000,0.1797269617)
(0.17500,0.1663122199)
(0.18000,0.1532306057)
(0.18500,0.1404748381)
(0.19000,0.1280380535)
(0.19500,0.1159137730)
(0.20000,0.1040958735)
(0.20500,0.0925785608)
(0.21000,0.0813563468)
(0.21500,0.0704240272)
(0.22000,0.0597766627)
(0.22500,0.0494095613)
(0.23000,0.0393182619)
(0.23500,0.0294985200)
(0.24000,0.0199462942)
(0.24500,0.0106577336)
(0.25000,0.0016291674)
(0.25500,0.0000000000)
(0.26000,0.0000000000)
(0.26500,0.0000000000)
(0.27000,0.0000000000)
(0.27500,0.0000000000)
(0.28000,0.0000000000)
(0.28500,0.0000000000)
(0.29000,0.0000000000)
(0.29500,0.0000000000)
(0.30000,0.0000000000)
(0.30500,0.0000000000)
(0.31000,0.0000000000)
(0.31500,0.0000000000)
(0.32000,0.0000000000)
(0.32500,0.0000000000)
(0.33000,0.0000000000)
(0.33500,0.0000000000)
(0.34000,0.0000000000)
(0.34500,0.0000000000)
(0.35000,0.0000000000)
(0.35500,0.0000000000)
(0.36000,0.0000000000)
(0.36500,0.0000000000)
(0.37000,0.0000000000)
(0.37500,0.0000000000)
(0.38000,0.0000000000)
(0.38500,0.0000000000)
(0.39000,0.0000000000)
(0.39500,0.0000000000)
(0.40000,0.0000000000)
(0.40500,0.0000000000)
(0.41000,0.0000000000)
(0.41500,0.0000000000)
(0.42000,0.0000000000)
(0.42500,0.0000000000)
(0.43000,0.0000000000)
(0.43500,0.0000000000)
(0.44000,0.0000000000)
(0.44500,0.0000000000)
(0.45000,0.0000000000)
(0.45500,0.0000000000)
(0.46000,0.0000000000)
(0.46500,0.0000000000)
(0.47000,0.0000000000)
(0.47500,0.0000000000)
(0.48000,0.0000000000)
(0.48500,0.0000000000)
(0.49000,0.0000000000)
(0.49500,0.0000000000)
(0.50000,0.0000000000)
};
\addlegendentry{Uniform-LLL lower bound}
\addplot[UnionGray,thick,dotted] coordinates {
(0.00000,0.7500000000)
(0.00500,0.7058353077)
(0.01000,0.6717068641)
(0.01500,0.6413892899)
(0.02000,0.6135594575)
(0.02500,0.5875890685)
(0.03000,0.5631081422)
(0.03500,0.5398722735)
(0.04000,0.5177078109)
(0.04500,0.4964849661)
(0.05000,0.4761030429)
(0.05500,0.4564816401)
(0.06000,0.4375550808)
(0.06500,0.4192687119)
(0.07000,0.4015763491)
(0.07500,0.3844384559)
(0.08000,0.3678208098)
(0.08500,0.3516935038)
(0.09000,0.3360301829)
(0.09500,0.3208074518)
(0.10000,0.3060044064)
(0.10500,0.2916022603)
(0.11000,0.2775840418)
(0.11500,0.2639343470)
(0.12000,0.2506391347)
(0.12500,0.2376855568)
(0.13000,0.2250618150)
(0.13500,0.2127570399)
(0.14000,0.2007611884)
(0.14500,0.1890649547)
(0.15000,0.1776596953)
(0.15500,0.1665373620)
(0.16000,0.1556904454)
(0.16500,0.1451119241)
(0.17000,0.1347952213)
(0.17500,0.1247341649)
(0.18000,0.1149229543)
(0.18500,0.1053561286)
(0.19000,0.0960285401)
(0.19500,0.0869353298)
(0.20000,0.0780719051)
(0.20500,0.0694339206)
(0.21000,0.0610172601)
(0.21500,0.0528180204)
(0.22000,0.0448324970)
(0.22500,0.0370571710)
(0.23000,0.0294886965)
(0.23500,0.0221238900)
(0.24000,0.0149597206)
(0.24500,0.0079933002)
(0.25000,0.0012218755)
(0.25500,0.0000000000)
(0.26000,0.0000000000)
(0.26500,0.0000000000)
(0.27000,0.0000000000)
(0.27500,0.0000000000)
(0.28000,0.0000000000)
(0.28500,0.0000000000)
(0.29000,0.0000000000)
(0.29500,0.0000000000)
(0.30000,0.0000000000)
(0.30500,0.0000000000)
(0.31000,0.0000000000)
(0.31500,0.0000000000)
(0.32000,0.0000000000)
(0.32500,0.0000000000)
(0.33000,0.0000000000)
(0.33500,0.0000000000)
(0.34000,0.0000000000)
(0.34500,0.0000000000)
(0.35000,0.0000000000)
(0.35500,0.0000000000)
(0.36000,0.0000000000)
(0.36500,0.0000000000)
(0.37000,0.0000000000)
(0.37500,0.0000000000)
(0.38000,0.0000000000)
(0.38500,0.0000000000)
(0.39000,0.0000000000)
(0.39500,0.0000000000)
(0.40000,0.0000000000)
(0.40500,0.0000000000)
(0.41000,0.0000000000)
(0.41500,0.0000000000)
(0.42000,0.0000000000)
(0.42500,0.0000000000)
(0.43000,0.0000000000)
(0.43500,0.0000000000)
(0.44000,0.0000000000)
(0.44500,0.0000000000)
(0.45000,0.0000000000)
(0.45500,0.0000000000)
(0.46000,0.0000000000)
(0.46500,0.0000000000)
(0.47000,0.0000000000)
(0.47500,0.0000000000)
(0.48000,0.0000000000)
(0.48500,0.0000000000)
(0.49000,0.0000000000)
(0.49500,0.0000000000)
(0.50000,0.0000000000)
};
\addlegendentry{HSS union-bound lower bound}
\addplot[LevGreen,semithick,densely dotted] coordinates {
(0.00000,1.0000000000)
(0.00500,0.9567266063)
(0.01000,0.9253263663)
(0.01500,0.8990672912)
(0.02000,0.8764746603)
(0.02500,0.8568464276)
(0.03000,0.8397717279)
(0.03500,0.8249819033)
(0.04000,0.8122878408)
(0.04500,0.8015484908)
(0.05000,0.7926529821)
(0.05500,0.7855094682)
(0.06000,0.7800376490)
(0.06500,0.7761634826)
(0.07000,0.7738153328)
(0.07500,0.7729211582)
(0.08000,0.7729078048)
(0.08500,0.7729078048)
(0.09000,0.7729078048)
(0.09500,0.7729078048)
(0.10000,0.7729078048)
(0.10500,0.7729078048)
(0.11000,0.7729078048)
(0.11500,0.7729078048)
(0.12000,0.7729078048)
(0.12500,0.7729078048)
(0.13000,0.7729078048)
(0.13500,0.7729078048)
(0.14000,0.7729078048)
(0.14500,0.7729078048)
(0.15000,0.7729078048)
(0.15500,0.7729078048)
(0.16000,0.7729078048)
(0.16500,0.7729078048)
(0.17000,0.7729078048)
(0.17500,0.7729078048)
(0.18000,0.7729078048)
(0.18500,0.7729078048)
(0.19000,0.7729078048)
(0.19500,0.7729078048)
(0.20000,0.7729078048)
(0.20500,0.7729078048)
(0.21000,0.7729078048)
(0.21500,0.7729078048)
(0.22000,0.7729078048)
(0.22500,0.7729078048)
(0.23000,0.7700000000)
(0.23500,0.7650000000)
(0.24000,0.7600000000)
(0.24500,0.7550000000)
(0.25000,0.7500000000)
(0.25500,0.7450000000)
(0.26000,0.7400000000)
(0.26500,0.7350000000)
(0.27000,0.7300000000)
(0.27500,0.7250000000)
(0.28000,0.7200000000)
(0.28500,0.7150000000)
(0.29000,0.7100000000)
(0.29500,0.7050000000)
(0.30000,0.7000000000)
(0.30500,0.6950000000)
(0.31000,0.6900000000)
(0.31500,0.6850000000)
(0.32000,0.6800000000)
(0.32500,0.6750000000)
(0.33000,0.6700000000)
(0.33500,0.6650000000)
(0.34000,0.6600000000)
(0.34500,0.6550000000)
(0.35000,0.6500000000)
(0.35500,0.6450000000)
(0.36000,0.6400000000)
(0.36500,0.6350000000)
(0.37000,0.6300000000)
(0.37500,0.6250000000)
(0.38000,0.6200000000)
(0.38500,0.6150000000)
(0.39000,0.6100000000)
(0.39500,0.6050000000)
(0.40000,0.6000000000)
(0.40500,0.5950000000)
(0.41000,0.5900000000)
(0.41500,0.5850000000)
(0.42000,0.5800000000)
(0.42500,0.5750000000)
(0.43000,0.5700000000)
(0.43500,0.5650000000)
(0.44000,0.5600000000)
(0.44500,0.5550000000)
(0.45000,0.5500000000)
(0.45500,0.5450000000)
(0.46000,0.5400000000)
(0.46500,0.5350000000)
(0.47000,0.5300000000)
(0.47500,0.5250000000)
(0.48000,0.5200000000)
(0.48500,0.5150000000)
(0.49000,0.5100000000)
(0.49500,0.5050000000)
(0.50000,0.5000000000)
};
\addlegendentry{Levenshtein-type upper bound}
\addplot[EliasPurple,thick,dashdotted] coordinates {
(0.00000,1.0000000000)
(0.00500,0.9838329137)
(0.01000,0.9709250482)
(0.01500,0.9592131248)
(0.02000,0.9482621615)
(0.02500,0.9378668214)
(0.03000,0.9279068419)
(0.03500,0.9183031816)
(0.04000,0.9089999497)
(0.04500,0.8999555630)
(0.05000,0.8911378987)
(0.05500,0.8825214135)
(0.06000,0.8740853232)
(0.06500,0.8658123966)
(0.07000,0.8576881233)
(0.07500,0.8497001229)
(0.08000,0.8418377125)
(0.08500,0.8340915837)
(0.09000,0.8264535564)
(0.09500,0.8189163876)
(0.10000,0.8114736210)
(0.10500,0.8041194672)
(0.11000,0.7968487064)
(0.11500,0.7896566096)
(0.12000,0.7825388735)
(0.12500,0.7754915657)
(0.13000,0.7685110798)
(0.13500,0.7615940970)
(0.14000,0.7547375531)
(0.14500,0.7479386111)
(0.15000,0.7411946368)
(0.15500,0.7345031782)
(0.16000,0.7278619475)
(0.16500,0.7212688054)
(0.17000,0.7147217471)
(0.17500,0.7082188902)
(0.18000,0.7017584640)
(0.18500,0.6953387999)
(0.19000,0.6889583232)
(0.19500,0.6826155449)
(0.20000,0.6763090555)
(0.20500,0.6700375186)
(0.21000,0.6637996657)
(0.21500,0.6575942907)
(0.22000,0.6514202460)
(0.22500,0.6452764382)
(0.23000,0.6391618243)
(0.23500,0.6330754084)
(0.24000,0.6270162389)
(0.24500,0.6209834052)
(0.25000,0.6149760355)
(0.25500,0.6089932943)
(0.26000,0.6030343804)
(0.26500,0.5970985245)
(0.27000,0.5911849880)
(0.27500,0.5852930606)
(0.28000,0.5794220593)
(0.28500,0.5735713268)
(0.29000,0.5677402299)
(0.29500,0.5619281586)
(0.30000,0.5561345248)
(0.30500,0.5503587614)
(0.31000,0.5446003211)
(0.31500,0.5388586754)
(0.32000,0.5331333141)
(0.32500,0.5274237444)
(0.33000,0.5217294900)
(0.33500,0.5160500905)
(0.34000,0.5103851008)
(0.34500,0.5047340906)
(0.35000,0.4990966436)
(0.35500,0.4934723573)
(0.36000,0.4878608424)
(0.36500,0.4822617223)
(0.37000,0.4766746328)
(0.37500,0.4710992218)
(0.38000,0.4655351487)
(0.38500,0.4599820845)
(0.39000,0.4544397111)
(0.39500,0.4489077214)
(0.40000,0.4433858189)
(0.40500,0.4378737176)
(0.41000,0.4323711417)
(0.41500,0.4268778255)
(0.42000,0.4213935133)
(0.42500,0.4159179596)
(0.43000,0.4104509283)
(0.43500,0.4049921933)
(0.44000,0.3995415383)
(0.44500,0.3940987566)
(0.45000,0.3886636514)
(0.45500,0.3832360356)
(0.46000,0.3778157321)
(0.46500,0.3724025735)
(0.47000,0.3669964028)
(0.47500,0.3615970730)
(0.48000,0.3562044475)
(0.48500,0.3508184004)
(0.49000,0.3454388165)
(0.49500,0.3400655917)
(0.50000,0.3346986333)
};
\addlegendentry{Higher-order Elias upper bound}
\addplot[HSSCyan,semithick,dashed] coordinates {
(0.00000,1.0000000000)
(0.00500,0.9900000000)
(0.01000,0.9800000000)
(0.01500,0.9700000000)
(0.02000,0.9600000000)
(0.02500,0.9500000000)
(0.03000,0.9400000000)
(0.03500,0.9300000000)
(0.04000,0.9200000000)
(0.04500,0.9100000000)
(0.05000,0.9000000000)
(0.05500,0.8900000000)
(0.06000,0.8800000000)
(0.06500,0.8700000000)
(0.07000,0.8600000000)
(0.07500,0.8500000000)
(0.08000,0.8400000000)
(0.08500,0.8300000000)
(0.09000,0.8200000000)
(0.09500,0.8100000000)
(0.10000,0.8000000000)
(0.10500,0.7900000000)
(0.11000,0.7800000000)
(0.11500,0.7700000000)
(0.12000,0.7600000000)
(0.12500,0.7500000000)
(0.13000,0.7400000000)
(0.13500,0.7300000000)
(0.14000,0.7200000000)
(0.14500,0.7100000000)
(0.15000,0.7000000000)
(0.15500,0.6900000000)
(0.16000,0.6800000000)
(0.16500,0.6700000000)
(0.17000,0.6600000000)
(0.17500,0.6500000000)
(0.18000,0.6400000000)
(0.18500,0.6300000000)
(0.19000,0.6200000000)
(0.19500,0.6100000000)
(0.20000,0.6000000000)
(0.20500,0.5900000000)
(0.21000,0.5800000000)
(0.21500,0.5700000000)
(0.22000,0.5600000000)
(0.22500,0.5500000000)
(0.23000,0.5400000000)
(0.23500,0.5300000000)
(0.24000,0.5200000000)
(0.24500,0.5100000000)
(0.25000,0.5000000000)
(0.25500,0.4900000000)
(0.26000,0.4800000000)
(0.26500,0.4700000000)
(0.27000,0.4600000000)
(0.27500,0.4500000000)
(0.28000,0.4400000000)
(0.28500,0.4300000000)
(0.29000,0.4200000000)
(0.29500,0.4100000000)
(0.30000,0.4000000000)
(0.30500,0.3900000000)
(0.31000,0.3800000000)
(0.31500,0.3700000000)
(0.32000,0.3600000000)
(0.32500,0.3500000000)
(0.33000,0.3400000000)
(0.33500,0.3300000000)
(0.34000,0.3200000000)
(0.34500,0.3100000000)
(0.35000,0.3000000000)
(0.35500,0.2900000000)
(0.36000,0.2800000000)
(0.36500,0.2700000000)
(0.37000,0.2600000000)
(0.37500,0.2500000000)
(0.38000,0.2400000000)
(0.38500,0.2300000000)
(0.39000,0.2200000000)
(0.39500,0.2100000000)
(0.40000,0.2000000000)
(0.40500,0.1900000000)
(0.41000,0.1800000000)
(0.41500,0.1700000000)
(0.42000,0.1600000000)
(0.42500,0.1500000000)
(0.43000,0.1400000000)
(0.43500,0.1300000000)
(0.44000,0.1200000000)
(0.44500,0.1100000000)
(0.45000,0.1000000000)
(0.45500,0.0900000000)
(0.46000,0.0800000000)
(0.46500,0.0700000000)
(0.47000,0.0600000000)
(0.47500,0.0500000000)
(0.48000,0.0400000000)
(0.48500,0.0300000000)
(0.49000,0.0200000000)
(0.49500,0.0100000000)
(0.50000,0.0000000000)
};
\addlegendentry{HSS deletion upper bound}
\end{axis}
\end{tikzpicture}
\caption{Binary deletion bounds for $L=3$. In the displayed comparison, the higher-order deletion
curve does not improve the minimum of the other two converses.}
\label{fig:comparison-deletion-L3}
\end{figure}

\subsection{Binary insertion codes}\label{sec:comparison-insertion}

For $L=1$, insertion-only and deletion-only unique decodability
are equivalent. We therefore focus on $L=2,3$.
By Theorem~\ref{thm:insertion-lower-bound},
\begin{equation}\label{eq:insertion-comparison-lower}
\cR_{2,\mathrm{ins}}^{(L)}(\gamma)
\ge \pos{\Rirun(\gamma)}.
\end{equation}
We compare this SCS-coverage bound, evaluated at
\eqref{eq:insertion-reference-point}, against the CDR union
bound and our improved CDR bound in Table~\ref{tab:insertion-comparison}.

Combining the converses in Section~\ref{sec:converses},
we obtain
\begin{equation}\label{eq:ins-comparison-upper}
\begin{split}
\cR_{2,\mathrm{ins}}^{(L)}(\gamma)
\le \min\bigl\{&
\mathsf U_2^{\mathrm{SL}}(\gamma/L),
\mathsf U_{2,\mathrm{del}}^{\mathrm{HSS}}(\gamma/L),\\
&\mathsf U_{2,\mathrm{ins}}^{\mathrm{HSS}}(\gamma),
G_2(p_{\mathrm{ins}})\bigr\},
\end{split}
\end{equation}
where $p_{\mathrm{ins}}$ is defined in
Theorem~\ref{thm:higher-order-elias}.
Figures~\ref{fig:comparison-insertion-L2}--\ref{fig:comparison-insertion-L3}
compare these upper bounds with our insertion lower bounds.
The higher-order Elias bound improves the HSS insertion
upper bound and gives zero rate for $\gamma\ge 3/4$
when $L=2$ and $\gamma\ge 7/8$ when $L=3$.

The insertion lower bounds and deletion upper bounds
also extend to mixed errors through
Proposition~\ref{prop:ins2del}.

\begin{table}[tbp]
\centering\small
\begin{tabular}{@{}rrrrr@{}}
\toprule
$L$ & $\gamma$ & CDR union bound & Improved CDR bound & SCS coverage\\
\midrule
2 & 0.05 & 0.409994 & 0.614992 & 0.615276 \\
2 & 0.10 & 0.249886 & 0.374829 & 0.376665 \\
2 & 0.20 & 0.019973 & 0.029959 & 0.039957 \\
3 & 0.05 & 0.497494 & 0.663326 & 0.663343 \\
3 & 0.10 & 0.341553 & 0.455404 & 0.455610 \\
3 & 0.20 & 0.119973 & 0.159964 & 0.161898 \\
\bottomrule
\end{tabular}
\caption{Insertion lower-bound values for list sizes}
\label{tab:insertion-comparison}
\end{table}

\begin{figure}[tbp]
\centering
\begin{tikzpicture}
\begin{axis}[comparison axis,xmax=1,xtick={0,.2,.4,.6,.8,1},xlabel={Insertion fraction $\gamma$},title={Binary insertion codes, $L=2$}]
\addplot[RunBlue,thick] coordinates {
(0.00000,1.000000000000)
(0.01000,0.888596634449)
(0.02000,0.806994442627)
(0.03000,0.736531655324)
(0.04000,0.673250776415)
(0.05000,0.615276709776)
(0.06000,0.561507550950)
(0.07000,0.511216634660)
(0.08000,0.463888294169)
(0.09000,0.419137378292)
(0.10000,0.376665097201)
(0.11000,0.336232760920)
(0.12000,0.297645177822)
(0.13000,0.260739645754)
(0.14000,0.225378364510)
(0.15000,0.191443037027)
(0.16000,0.158830923631)
(0.17000,0.127451891532)
(0.18000,0.097226164481)
(0.19000,0.068082576514)
(0.20000,0.039957196063)
(0.21000,0.012792227088)
(0.22000,0.000000000000)
(0.23000,0.000000000000)
(0.24000,0.000000000000)
(0.25000,0.000000000000)
(0.26000,0.000000000000)
(0.27000,0.000000000000)
(0.28000,0.000000000000)
(0.29000,0.000000000000)
(0.30000,0.000000000000)
(0.31000,0.000000000000)
(0.32000,0.000000000000)
(0.33000,0.000000000000)
(0.34000,0.000000000000)
(0.35000,0.000000000000)
(0.36000,0.000000000000)
(0.37000,0.000000000000)
(0.38000,0.000000000000)
(0.39000,0.000000000000)
(0.40000,0.000000000000)
(0.41000,0.000000000000)
(0.42000,0.000000000000)
(0.43000,0.000000000000)
(0.44000,0.000000000000)
(0.45000,0.000000000000)
(0.46000,0.000000000000)
(0.47000,0.000000000000)
(0.48000,0.000000000000)
(0.49000,0.000000000000)
(0.50000,0.000000000000)
(0.51000,0.000000000000)
(0.52000,0.000000000000)
(0.53000,0.000000000000)
(0.54000,0.000000000000)
(0.55000,0.000000000000)
(0.56000,0.000000000000)
(0.57000,0.000000000000)
(0.58000,0.000000000000)
(0.59000,0.000000000000)
(0.60000,0.000000000000)
(0.61000,0.000000000000)
(0.62000,0.000000000000)
(0.63000,0.000000000000)
(0.64000,0.000000000000)
(0.65000,0.000000000000)
(0.66000,0.000000000000)
(0.67000,0.000000000000)
(0.68000,0.000000000000)
(0.69000,0.000000000000)
(0.70000,0.000000000000)
(0.71000,0.000000000000)
(0.72000,0.000000000000)
(0.73000,0.000000000000)
(0.74000,0.000000000000)
(0.75000,0.000000000000)
(0.76000,0.000000000000)
(0.77000,0.000000000000)
(0.78000,0.000000000000)
(0.79000,0.000000000000)
(0.80000,0.000000000000)
(0.81000,0.000000000000)
(0.82000,0.000000000000)
(0.83000,0.000000000000)
(0.84000,0.000000000000)
(0.85000,0.000000000000)
(0.86000,0.000000000000)
(0.87000,0.000000000000)
(0.88000,0.000000000000)
(0.89000,0.000000000000)
(0.90000,0.000000000000)
(0.91000,0.000000000000)
(0.92000,0.000000000000)
(0.93000,0.000000000000)
(0.94000,0.000000000000)
(0.95000,0.000000000000)
(0.96000,0.000000000000)
(0.97000,0.000000000000)
(0.98000,0.000000000000)
(0.99000,0.000000000000)
(1.00000,0.000000000000)
};
\addlegendentry{SCS-coverage lower bound}
\addplot[BaselineOrange,thick,dashed] coordinates {
(0.00000,1.000000000000)
(0.01000,0.888593888293)
(0.02000,0.806973511445)
(0.03000,0.736464282696)
(0.04000,0.673098324361)
(0.05000,0.614992201454)
(0.06000,0.561037386971)
(0.07000,0.510502038304)
(0.08000,0.462866531157)
(0.09000,0.417742788846)
(0.10000,0.374829971579)
(0.11000,0.333888084264)
(0.12000,0.294721265735)
(0.13000,0.257166688405)
(0.14000,0.221086894072)
(0.15000,0.186364330794)
(0.16000,0.152897353140)
(0.17000,0.120597226312)
(0.18000,0.089385837681)
(0.19000,0.059193918488)
(0.20000,0.029959641032)
(0.21000,0.001627497184)
(0.22000,0.000000000000)
(0.23000,0.000000000000)
(0.24000,0.000000000000)
(0.25000,0.000000000000)
(0.26000,0.000000000000)
(0.27000,0.000000000000)
(0.28000,0.000000000000)
(0.29000,0.000000000000)
(0.30000,0.000000000000)
(0.31000,0.000000000000)
(0.32000,0.000000000000)
(0.33000,0.000000000000)
(0.34000,0.000000000000)
(0.35000,0.000000000000)
(0.36000,0.000000000000)
(0.37000,0.000000000000)
(0.38000,0.000000000000)
(0.39000,0.000000000000)
(0.40000,0.000000000000)
(0.41000,0.000000000000)
(0.42000,0.000000000000)
(0.43000,0.000000000000)
(0.44000,0.000000000000)
(0.45000,0.000000000000)
(0.46000,0.000000000000)
(0.47000,0.000000000000)
(0.48000,0.000000000000)
(0.49000,0.000000000000)
(0.50000,0.000000000000)
(0.51000,0.000000000000)
(0.52000,0.000000000000)
(0.53000,0.000000000000)
(0.54000,0.000000000000)
(0.55000,0.000000000000)
(0.56000,0.000000000000)
(0.57000,0.000000000000)
(0.58000,0.000000000000)
(0.59000,0.000000000000)
(0.60000,0.000000000000)
(0.61000,0.000000000000)
(0.62000,0.000000000000)
(0.63000,0.000000000000)
(0.64000,0.000000000000)
(0.65000,0.000000000000)
(0.66000,0.000000000000)
(0.67000,0.000000000000)
(0.68000,0.000000000000)
(0.69000,0.000000000000)
(0.70000,0.000000000000)
(0.71000,0.000000000000)
(0.72000,0.000000000000)
(0.73000,0.000000000000)
(0.74000,0.000000000000)
(0.75000,0.000000000000)
(0.76000,0.000000000000)
(0.77000,0.000000000000)
(0.78000,0.000000000000)
(0.79000,0.000000000000)
(0.80000,0.000000000000)
(0.81000,0.000000000000)
(0.82000,0.000000000000)
(0.83000,0.000000000000)
(0.84000,0.000000000000)
(0.85000,0.000000000000)
(0.86000,0.000000000000)
(0.87000,0.000000000000)
(0.88000,0.000000000000)
(0.89000,0.000000000000)
(0.90000,0.000000000000)
(0.91000,0.000000000000)
(0.92000,0.000000000000)
(0.93000,0.000000000000)
(0.94000,0.000000000000)
(0.95000,0.000000000000)
(0.96000,0.000000000000)
(0.97000,0.000000000000)
(0.98000,0.000000000000)
(0.99000,0.000000000000)
(1.00000,0.000000000000)
};
\addlegendentry{Improved CDR lower bound}
\addplot[LevGreen,semithick,dotted] coordinates {
(0.00000,1.000000000000)
(0.01000,0.956726606256)
(0.02000,0.925326366264)
(0.03000,0.899067291203)
(0.04000,0.876474660261)
(0.05000,0.856846427617)
(0.06000,0.839771727934)
(0.07000,0.824981903252)
(0.08000,0.812287840766)
(0.09000,0.801548490807)
(0.10000,0.792652982057)
(0.11000,0.785509468191)
(0.12000,0.780037648959)
(0.13000,0.776163482551)
(0.14000,0.773815332811)
(0.15000,0.772921158193)
(0.16000,0.772907804781)
(0.17000,0.772907804781)
(0.18000,0.772907804781)
(0.19000,0.772907804781)
(0.20000,0.772907804781)
(0.21000,0.772907804781)
(0.22000,0.772907804781)
(0.23000,0.772907804781)
(0.24000,0.772907804781)
(0.25000,0.772907804781)
(0.26000,0.772907804781)
(0.27000,0.772907804781)
(0.28000,0.772907804781)
(0.29000,0.772907804781)
(0.30000,0.772907804781)
(0.31000,0.772907804781)
(0.32000,0.772907804781)
(0.33000,0.772907804781)
(0.34000,0.772907804781)
(0.35000,0.772907804781)
(0.36000,0.772907804781)
(0.37000,0.772907804781)
(0.38000,0.772907804781)
(0.39000,0.772907804781)
(0.40000,0.772907804781)
(0.41000,0.772907804781)
(0.42000,0.772907804781)
(0.43000,0.772907804781)
(0.44000,0.772907804781)
(0.45000,0.772907804781)
(0.46000,0.770000000000)
(0.47000,0.765000000000)
(0.48000,0.760000000000)
(0.49000,0.755000000000)
(0.50000,0.750000000000)
(0.51000,0.745000000000)
(0.52000,0.740000000000)
(0.53000,0.735000000000)
(0.54000,0.730000000000)
(0.55000,0.725000000000)
(0.56000,0.720000000000)
(0.57000,0.715000000000)
(0.58000,0.710000000000)
(0.59000,0.705000000000)
(0.60000,0.700000000000)
(0.61000,0.695000000000)
(0.62000,0.690000000000)
(0.63000,0.685000000000)
(0.64000,0.680000000000)
(0.65000,0.675000000000)
(0.66000,0.670000000000)
(0.67000,0.665000000000)
(0.68000,0.660000000000)
(0.69000,0.655000000000)
(0.70000,0.650000000000)
(0.71000,0.645000000000)
(0.72000,0.640000000000)
(0.73000,0.635000000000)
(0.74000,0.630000000000)
(0.75000,0.625000000000)
(0.76000,0.620000000000)
(0.77000,0.615000000000)
(0.78000,0.610000000000)
(0.79000,0.605000000000)
(0.80000,0.600000000000)
(0.81000,0.595000000000)
(0.82000,0.590000000000)
(0.83000,0.585000000000)
(0.84000,0.580000000000)
(0.85000,0.575000000000)
(0.86000,0.570000000000)
(0.87000,0.565000000000)
(0.88000,0.560000000000)
(0.89000,0.555000000000)
(0.90000,0.550000000000)
(0.91000,0.545000000000)
(0.92000,0.540000000000)
(0.93000,0.535000000000)
(0.94000,0.530000000000)
(0.95000,0.525000000000)
(0.96000,0.520000000000)
(0.97000,0.515000000000)
(0.98000,0.510000000000)
(0.99000,0.505000000000)
(1.00000,0.500000000000)
};
\addlegendentry{Transferred Levenshtein upper bound}
\addplot[HSSCyan,semithick,dashdotted] coordinates {
(0.00000,1.000000000000)
(0.01000,0.990000000000)
(0.02000,0.980000000000)
(0.03000,0.970000000000)
(0.04000,0.960000000000)
(0.05000,0.950000000000)
(0.06000,0.940000000000)
(0.07000,0.930000000000)
(0.08000,0.920000000000)
(0.09000,0.910000000000)
(0.10000,0.900000000000)
(0.11000,0.890000000000)
(0.12000,0.880000000000)
(0.13000,0.870000000000)
(0.14000,0.860000000000)
(0.15000,0.850000000000)
(0.16000,0.840000000000)
(0.17000,0.830000000000)
(0.18000,0.820000000000)
(0.19000,0.810000000000)
(0.20000,0.800000000000)
(0.21000,0.790000000000)
(0.22000,0.780000000000)
(0.23000,0.770000000000)
(0.24000,0.760000000000)
(0.25000,0.750000000000)
(0.26000,0.740000000000)
(0.27000,0.730000000000)
(0.28000,0.720000000000)
(0.29000,0.710000000000)
(0.30000,0.700000000000)
(0.31000,0.690000000000)
(0.32000,0.680000000000)
(0.33000,0.670000000000)
(0.34000,0.660000000000)
(0.35000,0.650000000000)
(0.36000,0.640000000000)
(0.37000,0.630000000000)
(0.38000,0.620000000000)
(0.39000,0.610000000000)
(0.40000,0.600000000000)
(0.41000,0.590000000000)
(0.42000,0.580000000000)
(0.43000,0.570000000000)
(0.44000,0.560000000000)
(0.45000,0.550000000000)
(0.46000,0.540000000000)
(0.47000,0.530000000000)
(0.48000,0.520000000000)
(0.49000,0.510000000000)
(0.50000,0.500000000000)
(0.51000,0.490000000000)
(0.52000,0.480000000000)
(0.53000,0.470000000000)
(0.54000,0.460000000000)
(0.55000,0.450000000000)
(0.56000,0.440000000000)
(0.57000,0.430000000000)
(0.58000,0.420000000000)
(0.59000,0.410000000000)
(0.60000,0.400000000000)
(0.61000,0.390000000000)
(0.62000,0.380000000000)
(0.63000,0.370000000000)
(0.64000,0.360000000000)
(0.65000,0.350000000000)
(0.66000,0.340000000000)
(0.67000,0.330000000000)
(0.68000,0.320000000000)
(0.69000,0.310000000000)
(0.70000,0.300000000000)
(0.71000,0.290000000000)
(0.72000,0.280000000000)
(0.73000,0.270000000000)
(0.74000,0.260000000000)
(0.75000,0.250000000000)
(0.76000,0.240000000000)
(0.77000,0.230000000000)
(0.78000,0.220000000000)
(0.79000,0.210000000000)
(0.80000,0.200000000000)
(0.81000,0.190000000000)
(0.82000,0.180000000000)
(0.83000,0.170000000000)
(0.84000,0.160000000000)
(0.85000,0.150000000000)
(0.86000,0.140000000000)
(0.87000,0.130000000000)
(0.88000,0.120000000000)
(0.89000,0.110000000000)
(0.90000,0.100000000000)
(0.91000,0.090000000000)
(0.92000,0.080000000000)
(0.93000,0.070000000000)
(0.94000,0.060000000000)
(0.95000,0.050000000000)
(0.96000,0.040000000000)
(0.97000,0.030000000000)
(0.98000,0.020000000000)
(0.99000,0.010000000000)
(1.00000,0.000000000000)
};
\addlegendentry{Transferred HSS deletion upper bound}
\addplot[HSSGold,semithick,densely dashed] coordinates {
(0.00000,1.000000000000)
(0.01000,0.929062592195)
(0.02000,0.877982340964)
(0.03000,0.834309521798)
(0.04000,0.795398882908)
(0.05000,0.759994800970)
(0.06000,0.727358257981)
(0.07000,0.697001358870)
(0.08000,0.668577687438)
(0.09000,0.641828525898)
(0.10000,0.616553314386)
(0.11000,0.592592056176)
(0.12000,0.569814177157)
(0.13000,0.548111125604)
(0.14000,0.527391262715)
(0.15000,0.507576220530)
(0.16000,0.488598235427)
(0.17000,0.470398150875)
(0.18000,0.452923891787)
(0.19000,0.436129278993)
(0.20000,0.419973094022)
(0.21000,0.404418331456)
(0.22000,0.389431594048)
(0.23000,0.374982598101)
(0.24000,0.361043765068)
(0.25000,0.347589881391)
(0.26000,0.334597812882)
(0.27000,0.322046263158)
(0.28000,0.309915567951)
(0.29000,0.298187518866)
(0.30000,0.286845211520)
(0.31000,0.275872913981)
(0.32000,0.265255952222)
(0.33000,0.254980609946)
(0.34000,0.245034040586)
(0.35000,0.235404189687)
(0.36000,0.226079726193)
(0.37000,0.217049981402)
(0.38000,0.208304894529)
(0.39000,0.199834964045)
(0.40000,0.191631204007)
(0.41000,0.183685104786)
(0.42000,0.175988597645)
(0.43000,0.168534022704)
(0.44000,0.161314099898)
(0.45000,0.154321902601)
(0.46000,0.147550833596)
(0.47000,0.140994603154)
(0.48000,0.134647208990)
(0.49000,0.128502917903)
(0.50000,0.122556248918)
(0.51000,0.116801957799)
(0.52000,0.111235022774)
(0.53000,0.105850631373)
(0.54000,0.100644168273)
(0.55000,0.095611204038)
(0.56000,0.090747484702)
(0.57000,0.086048922089)
(0.58000,0.081511584829)
(0.59000,0.077131689996)
(0.60000,0.072905595320)
(0.61000,0.068829791924)
(0.62000,0.064900897541)
(0.63000,0.061115650183)
(0.64000,0.057470902201)
(0.65000,0.053963614738)
(0.66000,0.050590852515)
(0.67000,0.047349778950)
(0.68000,0.044237651557)
(0.69000,0.041251817635)
(0.70000,0.038389710202)
(0.71000,0.035648844168)
(0.72000,0.033026812726)
(0.73000,0.030521283946)
(0.74000,0.028129997563)
(0.75000,0.025850761940)
(0.76000,0.023681451190)
(0.77000,0.021620002463)
(0.78000,0.019664413367)
(0.79000,0.017812739532)
(0.80000,0.016063092291)
(0.81000,0.014413636491)
(0.82000,0.012862588400)
(0.83000,0.011408213733)
(0.84000,0.010048825764)
(0.85000,0.008782783537)
(0.86000,0.007608490158)
(0.87000,0.006524391173)
(0.88000,0.005528973022)
(0.89000,0.004620761557)
(0.90000,0.003798320643)
(0.91000,0.003060250807)
(0.92000,0.002405187963)
(0.93000,0.001831802178)
(0.94000,0.001338796509)
(0.95000,0.000924905877)
(0.96000,0.000588896001)
(0.97000,0.000329562367)
(0.98000,0.000145729250)
(0.99000,0.000036248772)
(1.00000,0.000000000000)
};
\addlegendentry{HSS insertion upper bound}
\addplot[EliasPurple,very thick,dashdotdotted] coordinates {
(0.00000,1.000000000000)
(0.01000,0.929057043921)
(0.02000,0.877946374184)
(0.03000,0.834204891113)
(0.04000,0.795179030925)
(0.05000,0.759607807650)
(0.06000,0.726748827106)
(0.07000,0.696112242378)
(0.08000,0.667350744241)
(0.09000,0.640205527695)
(0.10000,0.614476569090)
(0.11000,0.590004896012)
(0.12000,0.566661344251)
(0.13000,0.544339078636)
(0.14000,0.522948422060)
(0.15000,0.502413164984)
(0.16000,0.482667860556)
(0.17000,0.463655796768)
(0.18000,0.445327446318)
(0.19000,0.427639261428)
(0.20000,0.410552722859)
(0.21000,0.394033579592)
(0.22000,0.378051233822)
(0.23000,0.362578238266)
(0.24000,0.347589881391)
(0.25000,0.333063842267)
(0.26000,0.318979901137)
(0.27000,0.305319694990)
(0.28000,0.292066509813)
(0.29000,0.279205102985)
(0.30000,0.266721550601)
(0.31000,0.254603115590)
(0.32000,0.242838133264)
(0.33000,0.231415911579)
(0.34000,0.220326643876)
(0.35000,0.209561332278)
(0.36000,0.199111720207)
(0.37000,0.188970232768)
(0.38000,0.179129923940)
(0.39000,0.169584429670)
(0.40000,0.160327926131)
(0.41000,0.151355092493)
(0.42000,0.142661077668)
(0.43000,0.134241470556)
(0.44000,0.126092273398)
(0.45000,0.118209877885)
(0.46000,0.110591043729)
(0.47000,0.103232879425)
(0.48000,0.096132825006)
(0.49000,0.089288636559)
(0.50000,0.082698372367)
(0.51000,0.076360380496)
(0.52000,0.070273287721)
(0.53000,0.064435989659)
(0.54000,0.058847642022)
(0.55000,0.053507652886)
(0.56000,0.048415675909)
(0.57000,0.043571604433)
(0.58000,0.038975566396)
(0.59000,0.034627920009)
(0.60000,0.030529250159)
(0.61000,0.026680365488)
(0.62000,0.023082296107)
(0.63000,0.019736291938)
(0.64000,0.016643821637)
(0.65000,0.013806572078)
(0.66000,0.011226448395)
(0.67000,0.008905574550)
(0.68000,0.006846294421)
(0.69000,0.005051173411)
(0.70000,0.003523000552)
(0.71000,0.002264791118)
(0.72000,0.001279789738)
(0.73000,0.000571474012)
(0.74000,0.000143558627)
(0.75000,0.000000000000)
(0.76000,0.000000000000)
(0.77000,0.000000000000)
(0.78000,0.000000000000)
(0.79000,0.000000000000)
(0.80000,0.000000000000)
(0.81000,0.000000000000)
(0.82000,0.000000000000)
(0.83000,0.000000000000)
(0.84000,0.000000000000)
(0.85000,0.000000000000)
(0.86000,0.000000000000)
(0.87000,0.000000000000)
(0.88000,0.000000000000)
(0.89000,0.000000000000)
(0.90000,0.000000000000)
(0.91000,0.000000000000)
(0.92000,0.000000000000)
(0.93000,0.000000000000)
(0.94000,0.000000000000)
(0.95000,0.000000000000)
(0.96000,0.000000000000)
(0.97000,0.000000000000)
(0.98000,0.000000000000)
(0.99000,0.000000000000)
(1.00000,0.000000000000)
};
\addlegendentry{Higher-order Elias upper bound}
\end{axis}
\end{tikzpicture}
\caption{Binary insertion bounds for $L=2$. The SCS-coverage curve uses
fixed-run evaluations of \eqref{eq:insertion-list-lower-bound} at
\eqref{eq:insertion-reference-point}; the improved CDR curve is
$\pos{B_{\mathrm{ins},L}(\gamma)}$. Deletion converses use $\tau=\gamma/L$.}
\label{fig:comparison-insertion-L2}
\end{figure}

\begin{figure}[tbp]
\centering
\begin{tikzpicture}
\begin{axis}[comparison axis,xmax=1,xtick={0,.2,.4,.6,.8,1},xlabel={Insertion fraction $\gamma$},title={Binary insertion codes, $L=3$}]
\addplot[RunBlue,thick] coordinates {
(0.00000,1.000000000000)
(0.01000,0.902083492164)
(0.02000,0.830643658348)
(0.03000,0.769081908367)
(0.04000,0.713872723257)
(0.05000,0.663343702275)
(0.06000,0.616511410942)
(0.07000,0.572727313247)
(0.08000,0.531531510966)
(0.09000,0.492580985078)
(0.10000,0.455610170488)
(0.11000,0.420407468619)
(0.12000,0.386800373687)
(0.13000,0.354645592222)
(0.14000,0.323822221755)
(0.15000,0.294226889764)
(0.16000,0.265770196264)
(0.17000,0.238374050965)
(0.18000,0.211969640897)
(0.19000,0.186495852794)
(0.20000,0.161898030151)
(0.21000,0.138126980997)
(0.22000,0.115138176491)
(0.23000,0.092891096784)
(0.24000,0.071348691989)
(0.25000,0.050476934173)
(0.26000,0.030244442054)
(0.27000,0.010622164331)
(0.28000,0.000000000000)
(0.29000,0.000000000000)
(0.30000,0.000000000000)
(0.31000,0.000000000000)
(0.32000,0.000000000000)
(0.33000,0.000000000000)
(0.34000,0.000000000000)
(0.35000,0.000000000000)
(0.36000,0.000000000000)
(0.37000,0.000000000000)
(0.38000,0.000000000000)
(0.39000,0.000000000000)
(0.40000,0.000000000000)
(0.41000,0.000000000000)
(0.42000,0.000000000000)
(0.43000,0.000000000000)
(0.44000,0.000000000000)
(0.45000,0.000000000000)
(0.46000,0.000000000000)
(0.47000,0.000000000000)
(0.48000,0.000000000000)
(0.49000,0.000000000000)
(0.50000,0.000000000000)
(0.51000,0.000000000000)
(0.52000,0.000000000000)
(0.53000,0.000000000000)
(0.54000,0.000000000000)
(0.55000,0.000000000000)
(0.56000,0.000000000000)
(0.57000,0.000000000000)
(0.58000,0.000000000000)
(0.59000,0.000000000000)
(0.60000,0.000000000000)
(0.61000,0.000000000000)
(0.62000,0.000000000000)
(0.63000,0.000000000000)
(0.64000,0.000000000000)
(0.65000,0.000000000000)
(0.66000,0.000000000000)
(0.67000,0.000000000000)
(0.68000,0.000000000000)
(0.69000,0.000000000000)
(0.70000,0.000000000000)
(0.71000,0.000000000000)
(0.72000,0.000000000000)
(0.73000,0.000000000000)
(0.74000,0.000000000000)
(0.75000,0.000000000000)
(0.76000,0.000000000000)
(0.77000,0.000000000000)
(0.78000,0.000000000000)
(0.79000,0.000000000000)
(0.80000,0.000000000000)
(0.81000,0.000000000000)
(0.82000,0.000000000000)
(0.83000,0.000000000000)
(0.84000,0.000000000000)
(0.85000,0.000000000000)
(0.86000,0.000000000000)
(0.87000,0.000000000000)
(0.88000,0.000000000000)
(0.89000,0.000000000000)
(0.90000,0.000000000000)
(0.91000,0.000000000000)
(0.92000,0.000000000000)
(0.93000,0.000000000000)
(0.94000,0.000000000000)
(0.95000,0.000000000000)
(0.96000,0.000000000000)
(0.97000,0.000000000000)
(0.98000,0.000000000000)
(0.99000,0.000000000000)
(1.00000,0.000000000000)
};
\addlegendentry{SCS-coverage lower bound}
\addplot[BaselineOrange,thick,dashed] coordinates {
(0.00000,1.000000000000)
(0.01000,0.902083456260)
(0.02000,0.830643121285)
(0.03000,0.769079362396)
(0.04000,0.713865177210)
(0.05000,0.663326401292)
(0.06000,0.616477677308)
(0.07000,0.572668478492)
(0.08000,0.531436916584)
(0.09000,0.492438034530)
(0.10000,0.455404419181)
(0.11000,0.420122741568)
(0.12000,0.386418902875)
(0.13000,0.354148167471)
(0.14000,0.323188350286)
(0.15000,0.293434960706)
(0.16000,0.264797647235)
(0.17000,0.237197534499)
(0.18000,0.210565189049)
(0.19000,0.184839038656)
(0.20000,0.159964125362)
(0.21000,0.135891108608)
(0.22000,0.112575458730)
(0.23000,0.089976797467)
(0.24000,0.068058353424)
(0.25000,0.046786508521)
(0.26000,0.026130417175)
(0.27000,0.006061684210)
(0.28000,0.000000000000)
(0.29000,0.000000000000)
(0.30000,0.000000000000)
(0.31000,0.000000000000)
(0.32000,0.000000000000)
(0.33000,0.000000000000)
(0.34000,0.000000000000)
(0.35000,0.000000000000)
(0.36000,0.000000000000)
(0.37000,0.000000000000)
(0.38000,0.000000000000)
(0.39000,0.000000000000)
(0.40000,0.000000000000)
(0.41000,0.000000000000)
(0.42000,0.000000000000)
(0.43000,0.000000000000)
(0.44000,0.000000000000)
(0.45000,0.000000000000)
(0.46000,0.000000000000)
(0.47000,0.000000000000)
(0.48000,0.000000000000)
(0.49000,0.000000000000)
(0.50000,0.000000000000)
(0.51000,0.000000000000)
(0.52000,0.000000000000)
(0.53000,0.000000000000)
(0.54000,0.000000000000)
(0.55000,0.000000000000)
(0.56000,0.000000000000)
(0.57000,0.000000000000)
(0.58000,0.000000000000)
(0.59000,0.000000000000)
(0.60000,0.000000000000)
(0.61000,0.000000000000)
(0.62000,0.000000000000)
(0.63000,0.000000000000)
(0.64000,0.000000000000)
(0.65000,0.000000000000)
(0.66000,0.000000000000)
(0.67000,0.000000000000)
(0.68000,0.000000000000)
(0.69000,0.000000000000)
(0.70000,0.000000000000)
(0.71000,0.000000000000)
(0.72000,0.000000000000)
(0.73000,0.000000000000)
(0.74000,0.000000000000)
(0.75000,0.000000000000)
(0.76000,0.000000000000)
(0.77000,0.000000000000)
(0.78000,0.000000000000)
(0.79000,0.000000000000)
(0.80000,0.000000000000)
(0.81000,0.000000000000)
(0.82000,0.000000000000)
(0.83000,0.000000000000)
(0.84000,0.000000000000)
(0.85000,0.000000000000)
(0.86000,0.000000000000)
(0.87000,0.000000000000)
(0.88000,0.000000000000)
(0.89000,0.000000000000)
(0.90000,0.000000000000)
(0.91000,0.000000000000)
(0.92000,0.000000000000)
(0.93000,0.000000000000)
(0.94000,0.000000000000)
(0.95000,0.000000000000)
(0.96000,0.000000000000)
(0.97000,0.000000000000)
(0.98000,0.000000000000)
(0.99000,0.000000000000)
(1.00000,0.000000000000)
};
\addlegendentry{Improved CDR lower bound}
\addplot[LevGreen,semithick,dotted] coordinates {
(0.00000,1.000000000000)
(0.01000,0.968936597360)
(0.02000,0.945525661315)
(0.03000,0.925326366264)
(0.04000,0.907364352491)
(0.05000,0.891172437726)
(0.06000,0.876474660261)
(0.07000,0.863088955237)
(0.08000,0.850886529181)
(0.09000,0.839771727934)
(0.10000,0.829670855990)
(0.11000,0.820525430469)
(0.12000,0.812287840766)
(0.13000,0.804918406449)
(0.14000,0.798383293513)
(0.15000,0.792652982057)
(0.16000,0.787701102519)
(0.17000,0.783503527675)
(0.18000,0.780037648959)
(0.19000,0.777281791140)
(0.20000,0.775214735585)
(0.21000,0.773815332811)
(0.22000,0.773062191914)
(0.23000,0.772907804781)
(0.24000,0.772907804781)
(0.25000,0.772907804781)
(0.26000,0.772907804781)
(0.27000,0.772907804781)
(0.28000,0.772907804781)
(0.29000,0.772907804781)
(0.30000,0.772907804781)
(0.31000,0.772907804781)
(0.32000,0.772907804781)
(0.33000,0.772907804781)
(0.34000,0.772907804781)
(0.35000,0.772907804781)
(0.36000,0.772907804781)
(0.37000,0.772907804781)
(0.38000,0.772907804781)
(0.39000,0.772907804781)
(0.40000,0.772907804781)
(0.41000,0.772907804781)
(0.42000,0.772907804781)
(0.43000,0.772907804781)
(0.44000,0.772907804781)
(0.45000,0.772907804781)
(0.46000,0.772907804781)
(0.47000,0.772907804781)
(0.48000,0.772907804781)
(0.49000,0.772907804781)
(0.50000,0.772907804781)
(0.51000,0.772907804781)
(0.52000,0.772907804781)
(0.53000,0.772907804781)
(0.54000,0.772907804781)
(0.55000,0.772907804781)
(0.56000,0.772907804781)
(0.57000,0.772907804781)
(0.58000,0.772907804781)
(0.59000,0.772907804781)
(0.60000,0.772907804781)
(0.61000,0.772907804781)
(0.62000,0.772907804781)
(0.63000,0.772907804781)
(0.64000,0.772907804781)
(0.65000,0.772907804781)
(0.66000,0.772907804781)
(0.67000,0.772907804781)
(0.68000,0.772907804781)
(0.69000,0.770000000000)
(0.70000,0.766666666667)
(0.71000,0.763333333333)
(0.72000,0.760000000000)
(0.73000,0.756666666667)
(0.74000,0.753333333333)
(0.75000,0.750000000000)
(0.76000,0.746666666667)
(0.77000,0.743333333333)
(0.78000,0.740000000000)
(0.79000,0.736666666667)
(0.80000,0.733333333333)
(0.81000,0.730000000000)
(0.82000,0.726666666667)
(0.83000,0.723333333333)
(0.84000,0.720000000000)
(0.85000,0.716666666667)
(0.86000,0.713333333333)
(0.87000,0.710000000000)
(0.88000,0.706666666667)
(0.89000,0.703333333333)
(0.90000,0.700000000000)
(0.91000,0.696666666667)
(0.92000,0.693333333333)
(0.93000,0.690000000000)
(0.94000,0.686666666667)
(0.95000,0.683333333333)
(0.96000,0.680000000000)
(0.97000,0.676666666667)
(0.98000,0.673333333333)
(0.99000,0.670000000000)
(1.00000,0.666666666667)
};
\addlegendentry{Transferred Levenshtein upper bound}
\addplot[HSSCyan,semithick,dashdotted] coordinates {
(0.00000,1.000000000000)
(0.01000,0.993333333333)
(0.02000,0.986666666667)
(0.03000,0.980000000000)
(0.04000,0.973333333333)
(0.05000,0.966666666667)
(0.06000,0.960000000000)
(0.07000,0.953333333333)
(0.08000,0.946666666667)
(0.09000,0.940000000000)
(0.10000,0.933333333333)
(0.11000,0.926666666667)
(0.12000,0.920000000000)
(0.13000,0.913333333333)
(0.14000,0.906666666667)
(0.15000,0.900000000000)
(0.16000,0.893333333333)
(0.17000,0.886666666667)
(0.18000,0.880000000000)
(0.19000,0.873333333333)
(0.20000,0.866666666667)
(0.21000,0.860000000000)
(0.22000,0.853333333333)
(0.23000,0.846666666667)
(0.24000,0.840000000000)
(0.25000,0.833333333333)
(0.26000,0.826666666667)
(0.27000,0.820000000000)
(0.28000,0.813333333333)
(0.29000,0.806666666667)
(0.30000,0.800000000000)
(0.31000,0.793333333333)
(0.32000,0.786666666667)
(0.33000,0.780000000000)
(0.34000,0.773333333333)
(0.35000,0.766666666667)
(0.36000,0.760000000000)
(0.37000,0.753333333333)
(0.38000,0.746666666667)
(0.39000,0.740000000000)
(0.40000,0.733333333333)
(0.41000,0.726666666667)
(0.42000,0.720000000000)
(0.43000,0.713333333333)
(0.44000,0.706666666667)
(0.45000,0.700000000000)
(0.46000,0.693333333333)
(0.47000,0.686666666667)
(0.48000,0.680000000000)
(0.49000,0.673333333333)
(0.50000,0.666666666667)
(0.51000,0.660000000000)
(0.52000,0.653333333333)
(0.53000,0.646666666667)
(0.54000,0.640000000000)
(0.55000,0.633333333333)
(0.56000,0.626666666667)
(0.57000,0.620000000000)
(0.58000,0.613333333333)
(0.59000,0.606666666667)
(0.60000,0.600000000000)
(0.61000,0.593333333333)
(0.62000,0.586666666667)
(0.63000,0.580000000000)
(0.64000,0.573333333333)
(0.65000,0.566666666667)
(0.66000,0.560000000000)
(0.67000,0.553333333333)
(0.68000,0.546666666667)
(0.69000,0.540000000000)
(0.70000,0.533333333333)
(0.71000,0.526666666667)
(0.72000,0.520000000000)
(0.73000,0.513333333333)
(0.74000,0.506666666667)
(0.75000,0.500000000000)
(0.76000,0.493333333333)
(0.77000,0.486666666667)
(0.78000,0.480000000000)
(0.79000,0.473333333333)
(0.80000,0.466666666667)
(0.81000,0.460000000000)
(0.82000,0.453333333333)
(0.83000,0.446666666667)
(0.84000,0.440000000000)
(0.85000,0.433333333333)
(0.86000,0.426666666667)
(0.87000,0.420000000000)
(0.88000,0.413333333333)
(0.89000,0.406666666667)
(0.90000,0.400000000000)
(0.91000,0.393333333333)
(0.92000,0.386666666667)
(0.93000,0.380000000000)
(0.94000,0.373333333333)
(0.95000,0.366666666667)
(0.96000,0.360000000000)
(0.97000,0.353333333333)
(0.98000,0.346666666667)
(0.99000,0.340000000000)
(1.00000,0.333333333333)
};
\addlegendentry{Transferred HSS deletion upper bound}
\addplot[HSSGold,semithick,densely dashed] coordinates {
(0.00000,1.000000000000)
(0.01000,0.929062592195)
(0.02000,0.877982340964)
(0.03000,0.834309521798)
(0.04000,0.795398882908)
(0.05000,0.759994800970)
(0.06000,0.727358257981)
(0.07000,0.697001358870)
(0.08000,0.668577687438)
(0.09000,0.641828525898)
(0.10000,0.616553314386)
(0.11000,0.592592056176)
(0.12000,0.569814177157)
(0.13000,0.548111125604)
(0.14000,0.527391262715)
(0.15000,0.507576220530)
(0.16000,0.488598235427)
(0.17000,0.470398150875)
(0.18000,0.452923891787)
(0.19000,0.436129278993)
(0.20000,0.419973094022)
(0.21000,0.404418331456)
(0.22000,0.389431594048)
(0.23000,0.374982598101)
(0.24000,0.361043765068)
(0.25000,0.347589881391)
(0.26000,0.334597812882)
(0.27000,0.322046263158)
(0.28000,0.309915567951)
(0.29000,0.298187518866)
(0.30000,0.286845211520)
(0.31000,0.275872913981)
(0.32000,0.265255952222)
(0.33000,0.254980609946)
(0.34000,0.245034040586)
(0.35000,0.235404189687)
(0.36000,0.226079726193)
(0.37000,0.217049981402)
(0.38000,0.208304894529)
(0.39000,0.199834964045)
(0.40000,0.191631204007)
(0.41000,0.183685104786)
(0.42000,0.175988597645)
(0.43000,0.168534022704)
(0.44000,0.161314099898)
(0.45000,0.154321902601)
(0.46000,0.147550833596)
(0.47000,0.140994603154)
(0.48000,0.134647208990)
(0.49000,0.128502917903)
(0.50000,0.122556248918)
(0.51000,0.116801957799)
(0.52000,0.111235022774)
(0.53000,0.105850631373)
(0.54000,0.100644168273)
(0.55000,0.095611204038)
(0.56000,0.090747484702)
(0.57000,0.086048922089)
(0.58000,0.081511584829)
(0.59000,0.077131689996)
(0.60000,0.072905595320)
(0.61000,0.068829791924)
(0.62000,0.064900897541)
(0.63000,0.061115650183)
(0.64000,0.057470902201)
(0.65000,0.053963614738)
(0.66000,0.050590852515)
(0.67000,0.047349778950)
(0.68000,0.044237651557)
(0.69000,0.041251817635)
(0.70000,0.038389710202)
(0.71000,0.035648844168)
(0.72000,0.033026812726)
(0.73000,0.030521283946)
(0.74000,0.028129997563)
(0.75000,0.025850761940)
(0.76000,0.023681451190)
(0.77000,0.021620002463)
(0.78000,0.019664413367)
(0.79000,0.017812739532)
(0.80000,0.016063092291)
(0.81000,0.014413636491)
(0.82000,0.012862588400)
(0.83000,0.011408213733)
(0.84000,0.010048825764)
(0.85000,0.008782783537)
(0.86000,0.007608490158)
(0.87000,0.006524391173)
(0.88000,0.005528973022)
(0.89000,0.004620761557)
(0.90000,0.003798320643)
(0.91000,0.003060250807)
(0.92000,0.002405187963)
(0.93000,0.001831802178)
(0.94000,0.001338796509)
(0.95000,0.000924905877)
(0.96000,0.000588896001)
(0.97000,0.000329562367)
(0.98000,0.000145729250)
(0.99000,0.000036248772)
(1.00000,0.000000000000)
};
\addlegendentry{HSS insertion upper bound}
\addplot[EliasPurple,very thick,dashdotdotted] coordinates {
(0.00000,1.000000000000)
(0.01000,0.929062537277)
(0.02000,0.877981636475)
(0.03000,0.834306481191)
(0.04000,0.795390459562)
(0.05000,0.759976478667)
(0.06000,0.727324035557)
(0.07000,0.696943794665)
(0.08000,0.668487982281)
(0.09000,0.641696630497)
(0.10000,0.616368051437)
(0.11000,0.592341250383)
(0.12000,0.569484788387)
(0.13000,0.547689381790)
(0.14000,0.526862790151)
(0.15000,0.506926169955)
(0.16000,0.487811402823)
(0.17000,0.469459092327)
(0.18000,0.451817032035)
(0.19000,0.434839013525)
(0.20000,0.418483884730)
(0.21000,0.402714795965)
(0.22000,0.387498588928)
(0.23000,0.372805296218)
(0.24000,0.358607727364)
(0.25000,0.344881123395)
(0.26000,0.331602866292)
(0.27000,0.318752232795)
(0.28000,0.306310184407)
(0.29000,0.294259187169)
(0.30000,0.282583056092)
(0.31000,0.271266820188)
(0.32000,0.260296604780)
(0.33000,0.249659528441)
(0.34000,0.239343612341)
(0.35000,0.229337700214)
(0.36000,0.219631387433)
(0.37000,0.210214957938)
(0.38000,0.201079327986)
(0.39000,0.192215995810)
(0.40000,0.183616996458)
(0.41000,0.175274861161)
(0.42000,0.167182580686)
(0.43000,0.159333572212)
(0.44000,0.151721649311)
(0.45000,0.144340994699)
(0.46000,0.137186135441)
(0.47000,0.130251920354)
(0.48000,0.123533499374)
(0.49000,0.117026304685)
(0.50000,0.110726033427)
(0.51000,0.104628631842)
(0.52000,0.098730280698)
(0.53000,0.093027381883)
(0.54000,0.087516546060)
(0.55000,0.082194581274)
(0.56000,0.077058482442)
(0.57000,0.072105421635)
(0.58000,0.067332739087)
(0.59000,0.062737934872)
(0.60000,0.058318661191)
(0.61000,0.054072715218)
(0.62000,0.049998032460)
(0.63000,0.046092680592)
(0.64000,0.042354853729)
(0.65000,0.038782867100)
(0.66000,0.035375152098)
(0.67000,0.032130251669)
(0.68000,0.029046816032)
(0.69000,0.026123598685)
(0.70000,0.023359452694)
(0.71000,0.020753327241)
(0.72000,0.018304264409)
(0.73000,0.016011396191)
(0.74000,0.013873941713)
(0.75000,0.011891204650)
(0.76000,0.010062570827)
(0.77000,0.008387505996)
(0.78000,0.006865553773)
(0.79000,0.005496333730)
(0.80000,0.004279539632)
(0.81000,0.003214937812)
(0.82000,0.002302365672)
(0.83000,0.001541730311)
(0.84000,0.000933007273)
(0.85000,0.000476239394)
(0.86000,0.000171535771)
(0.87000,0.000019070818)
(0.88000,0.000000000000)
(0.89000,0.000000000000)
(0.90000,0.000000000000)
(0.91000,0.000000000000)
(0.92000,0.000000000000)
(0.93000,0.000000000000)
(0.94000,0.000000000000)
(0.95000,0.000000000000)
(0.96000,0.000000000000)
(0.97000,0.000000000000)
(0.98000,0.000000000000)
(0.99000,0.000000000000)
(1.00000,0.000000000000)
};
\addlegendentry{Higher-order Elias upper bound}
\end{axis}
\end{tikzpicture}
\caption{Binary insertion bounds for $L=3$, with the same benchmarks
as Figure~\ref{fig:comparison-insertion-L2}. The higher-order insertion
upper bound vanishes at $\gamma=7/8$.}
\label{fig:comparison-insertion-L3}
\end{figure}
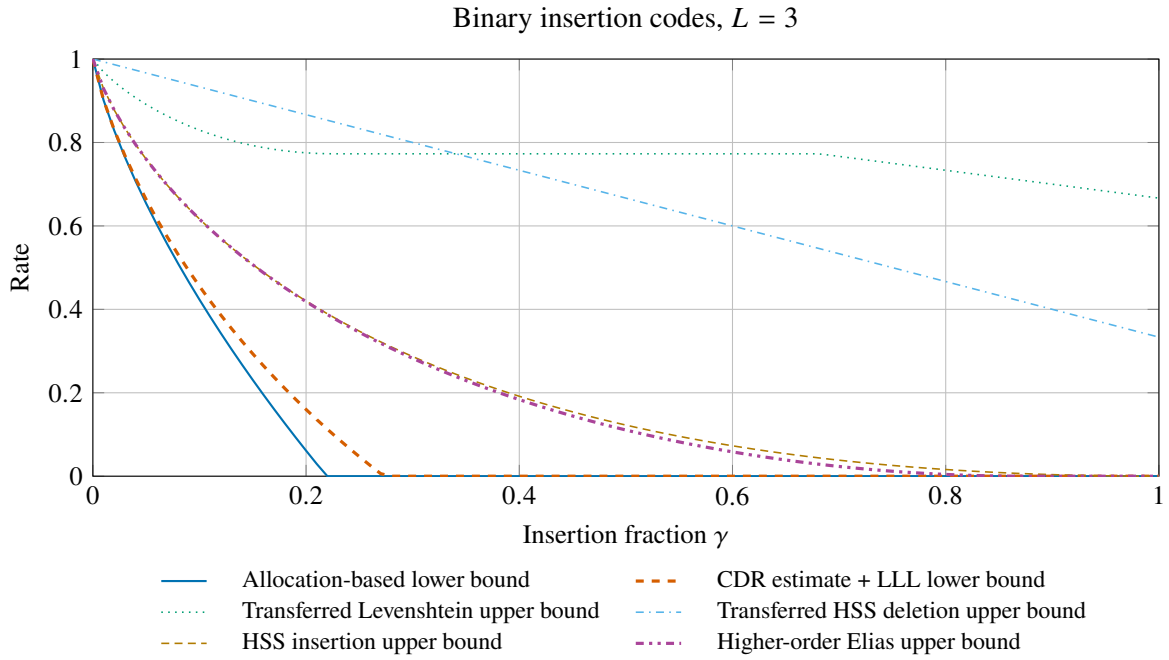

\let\SavedFloatBarrier\FloatBarrier
\let\FloatBarrier\relax

\section{Conclusion and Open Questions}\label{sec:discussion}

\let\FloatBarrier\SavedFloatBarrier

We have obtained improved asymptotic bounds for adversarial
insertion--deletion codes with prescribed finite list sizes.
In particular, we strictly improve Levenshtein's classical
rate lower bound for binary unique-deletion correction in
the linear-deletion regime. For list decoding, we exploit
structural restrictions on common subsequences and
supersequences to obtain stronger achievable-rate bounds
using run-restricted sampling and the Lov\'asz local lemma.
We also establish a reduction between insertion and deletion
budgets, together with new Levenshtein-type and higher-order
Elias-type converses.

Several questions remain. Our generating functions count
witness--tuple configurations rather than distinct bad tuples.
Can a sharper treatment of witness multiplicity, or stronger
structural restrictions on witnesses, improve the achievable
rates? Another challenge is to construct explicit codes
with comparable rates and prescribed finite list sizes,
together with efficient encoding and decoding algorithms.
Finally, direct counting arguments for mixed insertion
and deletion errors may yield stronger bounds than those
obtained through the budget reduction.

\section*{Acknowledgment of AI assistance}
\addcontentsline{toc}{section}{Acknowledgment of AI assistance}

We used ChatGPT 6 to assist with mathematical exploration,
checking calculations and arguments, and preparation of the manuscript. The tool materially affected
Section 6, Section 8 and Appendix A. In particular, it assisted with
the transfer-matrix analyses in
Sections 6, the coefficient-asymptotic arguments in
Appendix A, and the development of code for the numerical
evaluations and comparisons in Section 8. The authors verified the correctness and originality of all content, including references.

\appendix
\numberwithin{theorem}{section}
\numberwithin{equation}{section}
\section{Proofs for Critical-Point Coefficient Evaluation}
\label{app:critical}\label{app:acsv}
This appendix proves the Critical-Point Coefficient Lemma and then
establishes Theorem~\ref{thm:deletion-beta} for a fixed run density.
This exact-exponent result is separate from the admissible-stationary-point
coding bounds in Sections~\ref{sec:unique} and~\ref{sec:deletions}.
The argument uses the large-power structure and nonnegative coefficients,
rather than invoking a general rational-function ACSV theorem.
For analytic-combinatorics context, see
\cite{pemantle2008twenty,melczer2021invitation,goyal2024gilbert,lenz2025multivariate}.

\subsection{A critical-point coefficient lemma and its proof}
\label{app:critical-proof}\label{sec:critical-tool}
Consider $[\vu^{\vj_N}]B(\vu)K(\vu)^N$, where $B$ is a fixed
boundary contribution and $K$ is the repeated transition series. Here $\ln$ denotes the natural
logarithm; coefficient exponents are expressed in base two. The support
of a power series is the set of exponent vectors of its nonzero
monomials. Its support-difference lattice consists of all integer linear
combinations of differences of such vectors.

\begin{lemma}[Critical-Point Coefficient Lemma]
\label{lem:critical-point-coefficient}\label{lem:coefficient-powers}
Let $s\ge1$ and let $B(\vu)$ and $K(\vu)$ be fixed, nonzero power series
in $\vu=(u_1,\ldots,u_s)$ with nonnegative real coefficients. Set
$F(z,\vu)=B(\vu)/(1-zK(\vu))$. Let
$\vj_N\in\Z_{\ge0}^s$ satisfy $\vj_N/N\to\vmu$.
Suppose the following conditions hold.
\begin{enumerate}[label=\textnormal{(\alph*)},leftmargin=2em,itemsep=.25em]
\item \emph{Critical-point equations.} There is $\vu^*>0$ such that
\begin{equation}\label{eq:coefficient-critical-equations}
\frac{u_i^*K_{u_i}(\vu^*)}{K(\vu^*)}=\mu_i\quad(1\le i\le s).
\end{equation}
Set $z^*=1/K(\vu^*)$.
\item Both power series converge in a neighborhood of $\vu^*$.
\item \emph{Full-rank support-difference lattice.} With
$\supp(K)=\{\vvec\in\Z_{\ge0}^s:[\vu^{\vvec}]K(\vu)>0\}$, the lattice
\[
\Lambda(K)=\operatorname{span}_{\Z}
\{\vvec-\vvec':\vvec,\vvec'\in\supp(K)\}
\]
has rank $s$.
\item \emph{Lattice compatibility.} For some $\vvec_0\in\supp(K)$, there
is a fixed monomial $b_{\vw}\vu^{\vw}$ of $B$, with $b_{\vw}>0$, such that
$\vj_N-\vw\in N\vvec_0+\Lambda(K)$ for all sufficiently large $N$.
\end{enumerate}
Then $c_N=[\vu^{\vj_N}]B(\vu)K(\vu)^N
=[z^N\vu^{\vj_N}]F(z,\vu)$ is eventually positive and
\begin{equation}\label{eq:coefficient-power-estimate}
\lim_{N\to\infty}\frac1N\log_2 c_N
=\log_2K(\vu^*)-\langle\vmu,\log_2\vu^*\rangle
=-\log_2z^*-\langle\vmu,\log_2\vu^*\rangle.
\end{equation}
\end{lemma}

The direction in the lemma need only converge; no bounded-displacement
assumption is imposed. Interior convergence is required for the power
series themselves, not merely their rational continuations. The proof
uses a single fixed positive monomial of $B$ to meet the lattice condition.

\begin{proof}
Write $K(\vu)=\sum_{\vvec}\kappa_{\vvec}\vu^{\vvec}$, put $\vt^*=\ln\vu^*$, and define
\[
\mathcal L(\vt)=\ln K(e^{\vt}),\qquad
\Pr_{\vt}(\vV=\vvec)=\kappa_{\vvec}\exp\bigl(\langle\vt,\vvec\rangle-\mathcal L(\vt)\bigr).
\]
Interior convergence gives an exponential moment locally and permits
termwise differentiation. Therefore
\begin{equation}\label{eq:tilted-mean-covariance}
\nabla\mathcal L(\vt)=\E_{\vt}\vV,\qquad
\nabla^2\mathcal L(\vt)=\Sigma_{\vt}=\operatorname{Cov}_{\vt}(\vV).
\end{equation}
In particular, the $i$th mean is
$\E_{\vt}V_i=u_i\partial_{u_i}\ln K(\vu)$ for $\vu=e^{\vt}$:
termwise differentiation multiplies each monomial by its $i$th exponent.
These are derivatives with respect to logarithmic coordinates, not
expectations of logarithms.
If $\va^{\mathsf T}\Sigma_{\vt}\va=0$, then $\langle \va,\vV\rangle$ is constant
on the support. Thus $\va$ annihilates every support difference, and the
full-rank assumption implies $\va=\vzero$. Hence $\Sigma_{\vt}$ is positive
definite. This derives the needed nondegeneracy from the support condition.

We first choose a moving critical point.
Fix the monomial $b_{\vw}\vu^{\vw}$ from the lattice-compatibility assumption and
set $\vq_N=\vj_N-\vw$ and $\vmu_N=\vq_N/N$. Since $\vmu_N\to\vmu$
and $\nabla\mathcal L(\vt^*)=\vmu$, the inverse function theorem
and \eqref{eq:tilted-mean-covariance} give, for all large $N$, a nearby
$\vt_N$ with
\begin{equation}\label{eq:moving-critical-point}
\nabla\mathcal L(\vt_N)=\vmu_N,\qquad \vt_N\to\vt^*.
\end{equation}
In particular, the boundary-shifted target is exactly centered under
the tilted distribution at $\vt_N$. No estimate of the form
$\vq_N-N\vmu=O(1)$ is used.

We next establish a uniform local estimate on the lattice.
Choose a small compact neighborhood $U$ of $\vt^*$ strictly inside the
logarithmic convergence domain. Interior convergence gives, for some
$\epsilon>0$,
\[
\sup_{\vt\in U}\E_{\vt}e^{\epsilon\|\vV\|}<\infty.
\]
The covariance matrices vary continuously and their smallest eigenvalues
are bounded away from zero on $U$. Let
$\varphi_{\vt}(\vtheta)=\E_{\vt}e^{\mathrm i\langle\vtheta,\vV\rangle}$ on
$\R^s/(2\pi\Z^s)$, and write $\Lambda=\Lambda(K)$.
The points where $|\varphi_{\vt}(\vtheta)|=1$ form the same finite set
for every $\vt\in U$:
\[
H=\{\vh\bmod2\pi\Z^s:e^{\mathrm i\langle \vh,\vlambda\rangle}=1
       \text{ for all }\vlambda\in\Lambda\},
\qquad |H|=[\Z^s:\Lambda].
\]
Indeed, all support probabilities are positive, so equality in the
triangle inequality requires a constant phase on the support. Away
from fixed small neighborhoods of $H$, continuity and compactness give
$|\varphi_{\vt}|\le\xi<1$ uniformly over $U$.

Near zero, the common exponential moment yields the uniform expansion
\begin{equation}\label{eq:uniform-characteristic-expansion}
\ln\varphi_{\vt}(\vtheta)
=\mathrm i\langle\nabla\mathcal L(\vt),\vtheta\rangle
 -\tfrac12\vtheta^{\mathsf T}\Sigma_{\vt}\vtheta+O(\|\vtheta\|^3).
\end{equation}
On a smaller neighborhood we also have
$|\varphi_{\vt}(\vtheta)|\le e^{-c\|\vtheta\|^2}$ uniformly, for some $c>0$.
Since the support is contained in $\vvec_0+\Lambda$, every $\vh\in H$ satisfies
\[
\varphi_{\vt}(\vh+\vtheta)
=e^{\mathrm i\langle \vh,\vvec_0\rangle}\varphi_{\vt}(\vtheta).
\]
In Fourier inversion for the sum of $N$ independent copies at $\vt_N$,
the contribution near $\vh$ differs from that near zero by the factor
\begin{equation}\label{eq:lattice-phase-one}
e^{\mathrm i\langle \vh,N\vvec_0-\vq_N\rangle}=1,
\end{equation}
using $\vq_N\in N\vvec_0+\Lambda$. Thus these periodic contributions add.

Writing $\vV^{(1)},\ldots,\vV^{(N)}$ for independent copies at $\vt_N$,
\[
\Pr_{\vt_N}(\vV^{(1)}+\cdots+\vV^{(N)}=\vq_N)
=\frac1{(2\pi)^s}\int_{[-\pi,\pi]^s}
\varphi_{\vt_N}(\vtheta)^N e^{-\mathrm i\langle\vtheta,\vq_N\rangle}\,d\vtheta.
\]
Near zero put $\vtheta=N^{-1/2}\vz$. By \eqref{eq:moving-critical-point},
the linear term cancels exactly. On bounded $\vz$-sets,
\eqref{eq:uniform-characteristic-expansion} gives a Gaussian integrand;
the uniform quadratic bound controls the remaining local tails. The
complement of the phase neighborhoods contributes $O(\xi^N)$.
Combining the $[\Z^s:\Lambda]$ equal phase contributions gives
\begin{equation}\label{eq:moving-lattice-local-limit}
\Pr_{\vt_N}(\vV^{(1)}+\cdots+\vV^{(N)}=\vq_N)
=\frac{[\Z^s:\Lambda]}{(2\pi N)^{s/2}\sqrt{\det\Sigma_{\vt_N}}}
 (1+o(1)).
\end{equation}
This also proves eventual attainability of the compatible target.
Uniformity follows from convergence on a neighborhood and the fixed
full-rank support; it is not an additional assumption.

We now compare the coefficient with the corresponding evaluated series.
Exponential tilting gives the exact identity
\[
[\vu^{\vq_N}]K(\vu)^N
=e^{N\mathcal L(\vt_N)-\langle\vt_N,\vq_N\rangle}
 \Pr_{\vt_N}(\vV^{(1)}+\cdots+\vV^{(N)}=\vq_N).
\]
The probability in \eqref{eq:moving-lattice-local-limit} has logarithm
$o(N)$. Since $\vt_N\to\vt^*$ and $\vq_N/N\to\vmu$, this
coefficient has normalized natural-logarithmic limit
$\mathcal L(\vt^*)-\langle\vmu,\vt^*\rangle$.
Nonnegativity supplies the sandwich
\begin{equation}\label{eq:boundary-sandwich}
b_{\vw}[\vu^{\vj_N-\vw}]K(\vu)^N
\le c_N
\le B(\vu^*)K(\vu^*)^N(\vu^*)^{-\vj_N}.
\end{equation}
Both sides have the same exponential rate. Dividing by $\ln2$ proves
\eqref{eq:coefficient-power-estimate}, together with eventual positivity.
All estimates hold along arbitrary integer powers tending to infinity.
\end{proof}
The polynomial estimate above is centered at the moving point $\vt_N$.
Under mere convergence of the direction, no fixed-$\vu^*$ polynomial
multiplier is asserted; the lemma's conclusion is only the exponent.

\subsection{Fixed-density deletion coefficient evaluation}
\label{app:deletion-critical}
Here the run density $\rho$ is fixed before the first two stationary
equations are solved. Define
\begin{equation}\label{eq:beta-definition}
\beta_L(\delta,\rho)=\limsup_{n\to\infty}\frac1n
\log_2 M_L(n,\lfloor\delta n\rfloor;\lfloor\rho n\rfloor).
\end{equation}
The direct coefficient estimate gives $\beta_L\le\widehat\beta_L$.
The following result strengthens that estimate at an admissible
fixed-density solution, independently of the conditional coding
statement in Theorem~\ref{thm:deletion-stationary}.

\begin{theorem}[Fixed-density deletion coefficient exponent]
\label{thm:deletion-beta}
Fix $L\ge1$, $0<\delta<1$, and $0<\rho<1$. Suppose
$(x^*,y^*)\in\Ddomain$ satisfies
\eqref{eq:deletion-critical-x}--\eqref{eq:deletion-critical-y} for
this fixed $\rho$. Such a pair is unique and is the unique minimizer
of $\Phi_{L,\delta,\rho}$ on $\Ddomain$. The limsup above is a limit,
and
\begin{equation}\label{eq:beta}
\widehat\beta_L(\delta,\rho)=\beta_L(\delta,\rho)
=(1-\delta)\log_2\ldel{L}(x^*,y^*)
 -(L+1)\delta\log_2x^*-(L+1)\rho\log_2y^*.
\end{equation}
The same coefficient exponent holds for integers $0\le d_n<n$,
$1\le r_n\le n$ with $d_n/n\to\delta$ and $r_n/n\to\rho$.
The additional equation $\rho=y/(1+y)$ is not imposed here.
\end{theorem}
\begin{proof}
Put $\vu=(x_1,\ldots,x_{L+1},y_1,\ldots,y_{L+1})$, and
$\vu^*=(x^*,\ldots,x^*,y^*,\ldots,y^*)$.
Apply Lemma~\ref{lem:critical-point-coefficient} with
$B=\mathcal A_L$, $K=\mathcal Q_L$, power $N_n=n-d_n-1$, and
$\vj_n=(d_n,\ldots,d_n,r_n,\ldots,r_n)$.
We verify its assumptions explicitly.

The required convergence holds on a neighborhood of the chosen point.
The gap-restricted series is a subseries of the unrestricted
local choices, whose coefficients dominate its coefficients. These
local series involve only $(1-x_i)^{-1}$ and
$(1-x_i(1+y_i))^{-1}$, together with polynomial factors. Since
$x^*(1+y^*)<1$, there is $\epsilon>0$ with
$(x^*+\epsilon)(1+y^*+\epsilon)<1$. The series for both
$\mathcal Q_L$ and $\mathcal A_L$ therefore converge absolutely when
$|x_i|<x^*+\epsilon$ and $|y_i|<y^*+\epsilon$.
This supplies a neighborhood of $\vu^*$, not just convergence at a
single point.

Next we verify the individual coordinate mean equations.
On the diagonal, the chain rule gives
\[
\partial_x\ldel{L}(x,y)=\sum_{i=1}^{L+1}
 \left.\partial_{x_i}\mathcal Q_L\right|_{\rm diag},
\qquad
\partial_y\ldel{L}(x,y)=\sum_{i=1}^{L+1}
 \left.\partial_{y_i}\mathcal Q_L\right|_{\rm diag}.
\]
Symmetry under simultaneous permutations of the pairs $(x_i,y_i)$ makes
the ${L+1}$ terms in each sum equal. Hence, for every $i$,
\begin{align*}
x_i^*\partial_{x_i}\ln\mathcal Q_L(\vu^*)
&=\frac1{L+1} x^*\partial_x\ln \ldel{L}(x^*,y^*)
 =\frac\delta{1-\delta},\\
y_i^*\partial_{y_i}\ln\mathcal Q_L(\vu^*)
&=\frac1{L+1} y^*\partial_y\ln \ldel{L}(x^*,y^*)
 =\frac\rho{1-\delta}.
\end{align*}
These are the individual coordinate expectations under the auxiliary
distribution in \eqref{eq:tilted-mean-covariance}. They equal the
coordinates of $\lim\vj_n/N_n$, since $N_n/n\to1-\delta$.
Thus the factor ${L+1}$ in the symmetric equations comes from changing all
${L+1}$ copies of a variable simultaneously.

The support differences determine the relevant integer lattice.
Let $\ve_{x_i},\ve_{y_i}$ be the coordinate unit vectors and put
$\vone_y=\sum_{i=1}^{L+1} \ve_{y_i}$. The restricted transition series has the nonnegative expansion
\[
\mathcal Q_L=
\sum_{I\subsetneq[{L+1}]}\prod_{i\in I}
       \left(\sum_{a\ge1}x_i^a y_i^2\right)
+\left(\prod_i y_i\right)
 \sum_{\substack{\va\in\Z_{\ge0}^{L+1}\\\min_i a_i=0}}
       \prod_i x_i^{a_i}.
\]
In the first sum, $I$ is the set of words with a nonempty gap; it is
proper because at least one gap is empty. The second sum encodes the
same condition for a flip. Since $L+1\ge2$, a single nonempty gap is
allowed. In particular, this expansion contains the monomials $1$, $x_i y_i^2$, $x_i^2y_i^2$, and
$\prod_i y_i$. Their exponent vectors are respectively
$\vzero$, $\vvec_i=\ve_{x_i}+2\ve_{y_i}$, $\vvec_i'=2\ve_{x_i}+2\ve_{y_i}$, and
$\vone_y$. In the integer span of support differences,
\[
\ve_{x_i}=\vvec_i'-\vvec_i,
\qquad
2\ve_{y_i}=(\vvec_i-\vzero)-(\vvec_i'-\vvec_i),
\qquad
\vone_y=\vone_y-\vzero.
\]
Conversely, the $y$-exponents in the first sum are all even, and
those in the flip sum are all odd. Therefore
\begin{equation}\label{eq:deletion-lattice}
\Lambda(\mathcal Q_L)=\Z^{L+1}\times
\{\vb\in\Z^{L+1}:b_1\equiv\cdots\equiv b_{L+1}\pmod2\}.
\end{equation}
For the reverse inclusion in this equality, if $b_i=2q_i+\varepsilon$
with a common $\varepsilon\in\{0,1\}$, write
$(\va,\vb)=\sum_i a_i \ve_{x_i}+\sum_i q_i(2\ve_{y_i})+
\varepsilon\vone_y$.
The $2(L+1)$ independent vectors $\ve_{x_i},2\ve_{y_i}$ show full rank.
The index $[\Z^{2(L+1)}:\Lambda(\mathcal Q_L)]$, the number of cosets,
is $2^L$: a coset is determined by the $L$ relative parities
$b_i-b_1$ for $2\le i\le {L+1}$.

The boundary contribution contains a monomial with the required lattice compatibility.
Take $\vvec_0=\vzero$ and $\vw=\vone_y$. Selecting $S_0$ for every codeword and
the empty suffix gives the monomial $\prod_i y_i$ in $\mathcal A_L$.
Its full multivariate coefficient is
\[
[x_1^0\cdots x_{L+1}^0y_1\cdots y_{L+1}]\mathcal A_L=1:
\]
the initial matches already contribute $\prod_i y_i$, and any
nonempty suffix has positive $x$-degree and cannot contribute. Since all $y$-coordinates of
\[
\vj_n-\vw=(d_n,\ldots,d_n,r_n-1,\ldots,r_n-1)
\]
have the same parity, this vector belongs to \eqref{eq:deletion-lattice}.
The compatibility condition follows.

The coefficient lemma now applies to \eqref{eq:deletion-gf-power}. Its scalar product
term is
\[
\langle\vmu,\log_2\vu^*\rangle
=\frac{(L+1)\delta}{1-\delta}\log_2x^*
 +\frac{(L+1)\rho}{1-\delta}\log_2y^*.
\]
Multiplying its limit by $N_n/n\to1-\delta$ proves \eqref{eq:beta},
including for the stated converging integer sequences.
For uniqueness, consider $\ln \ldel{L}(e^a,e^b)$ on the convex domain
$a<0$, $b\in\R$. Its Hessian is the covariance of the projected
support. This support contains $(0,0),(1,2),(2,2)$, which are not
collinear, so the Hessian is positive definite. For this fixed $\rho$, subtracting the linear term
$(L+1)\delta a+(L+1)\rho b$ from
$(1-\delta)\ln\ldel{L}(e^a,e^b)$ leaves a strictly convex objective.
An admissible stationary point is therefore its unique global minimizer
on $a<0$, $b\in\R$, and also on the smaller domain corresponding
to $\Ddomain$. This is a fixed-density statement and does not imply
uniqueness for the joint three-variable system.
\end{proof}

\clearpage
\phantomsection

\end{document}